\documentclass[a4paper,12pt]{article}
\pdfoutput=1 %arXiv fully supports and automatically recognizes PDFLaTeX. You can ensure pdflatex processing by setting the flag \pdfoutput=1 within the first 5 lines of the preamble of the main .tex file. You should not need any other special flag.
\usepackage[margin=1in,includefoot]{geometry}
\usepackage{amsmath,amssymb,amsthm}
\usepackage[pdfstartview=FitH,pdfpagemode=None]{hyperref}
\usepackage{cite}
\usepackage{authblk}
\usepackage{xfrac} %per frazioni storte (\frac{}{})
\usepackage{cancel}
\usepackage{comment}
\usepackage{color}
\usepackage{soul}

\numberwithin{equation}{section}

\def\be{\begin{equation}}
\def\ee{\end{equation}}
\def\bea{\begin{eqnarray}}
\def\eea{\end{eqnarray}}
\newcommand{\Tr}{{\rm Tr}}

\renewcommand{\to}{\rightarrow}

\def\nb{\nonumber}
\def\p{\partial}

\def\a{\alpha}
\def\b{\beta}

\def\G{\Gamma}

\def\r{\rho}
\def\s{\sigma}

\def\th{\theta}

\def\mc{\mathcal}

\def\und{\underline}

\def\mc{\mathcal}

\newcommand{\sign}{\text{sign}}

\newtheorem{theorem}{Theorem}[section]
\newtheorem{corollary}{Corollary}[theorem]
\newtheorem{lemma}[theorem]{Lemma}

\newcommand{\rmd}{\,\mathrm{d}}

\def\({\left (}
\def\){\right )}

\let\benn\[
\let\eenn\]

\def\[{\left [}
\def\]{\right ]}

\let\oldlgraf\{ %rinomino la parentesi graffa con un altro nome, così posso usare \{ per la ridefinzione del comando
\renewcommand{\{}{\left \oldlgraf}

\let\oldrgraf\}
\renewcommand{\}}{\right \oldrgraf}

\title{\bfseries\LARGE Hagedorn temperature in holography: \protect\\ 
world-sheet and effective approaches}

\author[a]{Francesco Bigazzi}
\author[a,b]{Tommaso Canneti}
\author[a,b]{Federico Castellani}
\author[a,b]{\\Aldo L.\ Cotrone}            
\author[c,d]{Wolfgang M\"uck}

\affil[a]{Istituto Nazionale di Fisica Nucleare, Sezione di Firenze %
\protect\\ Via G. Sansone 1; 50019 Sesto Fiorentino (Firenze), Italy}
\affil[b]{Dipartimento di Fisica e Astronomia, Universit\`a di Firenze %
\protect\\ Via G. Sansone 1; 50019 Sesto Fiorentino (Firenze), Italy}
\affil[c]{Dipartimento di Fisica ``Ettore Pancini'', Universit\`a degli Studi di Napoli ``Federico II''% 
\protect\\ Via Cintia; 80126 Napoli, Italy}
\affil[d]{Istituto Nazionale di Fisica Nucleare, Sezione di Napoli %
\protect\\ Via Cintia; 80126 Napoli, Italy}
\date{\small bigazzi@fi.infn.it, canneti@fi.infn.it, federico.castellani@unifi.it, \\ cotrone@fi.infn.it, mueck@na.infn.it}
\begin{document}
\maketitle

\begin{abstract}
We provide general results on the Hagedorn temperature of planar, strongly coupled confining gauge theories holographically dual to type II superstring models on curved backgrounds with Ramond-Ramond and Kalb-Ramond fluxes and non-trivial dilaton. For exact backgrounds the Hagedorn temperature is determined up to next-to-next-to-next-to-leading order (NNNLO) in an expansion in $\alpha'$; in all the other cases the results can be safely trusted up to NNLO. To reach these goals we exploit two complementary approaches. On the one hand, we perform an extrapolation to the Hagedorn regime of world-sheet results obtained from the semiclassical quantization of string configurations winding around the compact Euclidean time direction. En passant, we provide a detailed derivation of the fermionic part of the world-sheet spectrum, which is hard to find in the literature. On the other hand, we perturbatively solve the equations of motion for the thermal scalar field corresponding to the lightest mode of the winding string, which in flat space becomes tachyonic above the Hagedorn temperature. The interplay between different approaches is surely convenient, but we provide insights about a possible derivation of the whole NNLO correction to the Hagedorn temperature from a pure world-sheet perspective; furthermore, we determine the effective mass of the thermal scalar from the world-sheet in full generality. 
\end{abstract}

%\baselineskip=15.5pt
%\pagestyle{plain}
%\setcounter{page}{1}
%\newfont{\namefont}{cmr10}
%\newfont{\addfont}{cmti7 scaled 1440}
%\newfont{\boldmathfont}{cmbx10}
%\newfont{\headfontb}{cmbx10 scaled 1728}
%\renewcommand{\theequation}{{\rm\thesection.\arabic{equation}}}
%\font\cmss=cmss10 \font\cmsss=cmss10 at 7pt
%\renewcommand{\thefootnote}{\arabic{footnote}}
%
%\begin{comment}
%\vspace{1cm}
%
%\begin{titlepage}
%\vskip 2cm
%\begin{center}
%{\Large{\bf Hagedorn temperature in holography: \\ 
%world-sheet and effective approaches
%}}
%\end{center}
%
%\vskip 10pt
%\begin{center}
%Francesco Bigazzi$^{a}$, Tommaso Canneti$^{a,b}$, Federico Castellani$^{a,b}$, \\
%Aldo L. Cotrone$^{a,b}$, Wolfgang M\"uck$^{c,d}$
%\end{center}
%\vskip 10pt
%\begin{center}
%\vspace{0.2cm}
%\textit{$^a$ INFN, Sezione di Firenze; Via G. Sansone 1; I-50019 Sesto Fiorentino (Firenze), Italy.
%}\\
%\textit{$^b$ Dipartimento di Fisica e Astronomia, Universit\'a di Firenze; Via G. Sansone 1;\\ I-50019 Sesto Fiorentino %(Firenze), Italy.
%}\\
%\textit{$^c$ Dipartimento di Fisica ``Ettore Pancini'', Universit\'a degli Studi di Napoli ``Federico II''; Via Cintia; 80126 %Napoli, Italy.
%}\\
%\textit{$^d$ INFN, Sezione di Napoli; Via Cintia; 80126 Napoli, Italy.
%}\\
%\vskip 20pt
%{\small{
%bigazzi@fi.infn.it, canneti@fi.infn.it, federico.castellani@unifi.it, cotrone@fi.infn.it,mueck@na.infn.it}
%}
%\end{center}
%
%\vspace{25pt}
%
%\begin{center}
% \textbf{Abstract}
%\end{center}
%
%\noindent 
%\end{titlepage}

\newpage
\tableofcontents

%%%%%%%%%%%%%%%%%%%%%%%%%%%%%%%%%%%%%%%%%%%%%%%%%%%%%%%%%%%%%%%%%%%%%%%%%%%%%%%%%%%%%%%%%%%%%%%%%%%%%%%%%%%%%%%%%%%

\section{Introduction}
Solvable models of strings display a density of states which exponentially grows with the energy in the UV, regardless of whether the target space is flat or not \cite{Huang:1970iq, Sundborg:1984uk, Tye:1985jv, Bowick:1985az, Matsuo:1986es, Russo:2002rq,Canneti:2024iyn}. This implies that the one-loop partition function of the theory, at finite temperature $T$, diverges when $T>T_H$, where $T_H$ is the so-called Hagedorn temperature. Correspondingly, the lowest string state winding once around the compact time direction becomes tachyonic \cite{Atick:1988si}. This state is a complex scalar field, often referred to as the ``thermal scalar'', in the target space. When $T$ approaches $T_H$ its mass goes to zero and the field can be accounted for in the low energy string effective action. In this framework, the requirement that the linearized equation of motion for the thermal scalar admits a normalizable solution, gives the Hagedorn temperature.

Confining gauge theories like $SU(N)$ Yang-Mills display a Hagedorn behavior too. In first order confinement-deconfinement phase transitions, the Hagedorn temperature can be understood as the limiting temperature above which the (metastable) confining branch ceases to exist. Lattice results (see e.g. \cite{Bringoltz:2005xx,Caselle:2015tza}) in fact, indicate that in $SU(N\ge3)$ Yang-Mills theories $T_H$ is (slightly) larger than the critical temperature for deconfinement $T_c$.

Classes of planar confining gauge theories at large 't Hooft coupling admit a dual holographic description in terms of superstring theories on curved backgrounds with fluxes, so that, in principle, we can obtain the related $T_H$ from string theory computations. In practice, the strategy consists on expanding the world-sheet sigma model to quadratic order in quantum fluctuations around a certain reference winding string configuration. The latter is chosen in such a way to probe the asymptotic nearly-flat region of the backgrounds, which holographically corresponds to the deep IR regime of the dual gauge theories. In this limit, one can get the string spectrum and thus, focusing on the ground state, the Hagedorn temperature. 
To leading order (LO) in the $\alpha'$ expansion, this analysis gives $T_H = \sqrt{T_s/(4\pi)}$, where $T_s$ is the confining string tension \cite{Bigazzi:2022gal}. This result agrees with the value of the Hagedorn temperature inferred from the effective action of the thermal scalar field, along the lines of the Horowitz-Polchinski construction \cite{Horowitz:1997jc}, and reduces (for $T_s = (2\pi\alpha')^{-1}$) to the known value of the superstring Hagedorn temperature in flat spacetime. Moreover, as it was noticed in \cite{Bigazzi:2022gal} focusing on an explicit example, $T_H$ turns out to be parametrically larger than $T_c$ at strong coupling.

The above results have been extended to next-to-leading order (NLO) in the strong 't Hooft coupling expansion (which is the $\alpha'$ expansion in string theory) in two complementary ways. From one side, by solving, to first order in a suitable perturbative regime, the equation of motion for the thermal scalar field on a class of holographic confining backgrounds and on global $AdS_{d+1}$ spacetimes (dual to $d$-dimensional conformal field theories on $S^{d-1}$ spheres) \cite{Urbach:2022xzw,Urbach:2023npi}. From the other side, focusing on a specific example, by computing the NLO correction to the mass of the winding string ground state on the world-sheet \cite{Bigazzi:2023oqm}. Remarkably, the two computations turned out to perfectly agree. The reason is that the NLO correction to the Hagedorn temperature is due to the contribution of the zero modes of the (massive) bosonic world-sheet fields, which are in turn accounted for by the effective approach. 

Very recently, the above analysis has been extended up to NNNLO in global $AdS_{d+1}$ setups \cite{Ekhammar:2023glu,Bigazzi:2023hxt,Ekhammar:2023cuj}. In the $AdS_5$ and $AdS_4$ cases, respectively dual to ${\cal N}=4$ SYM on $S^3$ and to the ABJM theory on $S^2$, the holographic results on the Hagedorn temperature to NNNLO remarkably agree with those obtained, numerically, by means of Quantum Spectral Curve (QCS) integrability methods \cite{Harmark:2017yrv,Harmark:2018red,Harmark:2021qma,Ekhammar:2023glu,Ekhammar:2023cuj}. These are arguably among the most impressive precision tests of the holographic correspondence at finite temperature. The interplay of effective and world-sheet methods in this context revealed to be a very crucial ingredient.

At NNLO in the strong coupling expansion, in fact, the Hagedorn temperature computed using the thermal scalar effective action turned out to miss a term, containing a $d \log2$ factor, which is instead precisely accounted for by the world-sheet approach, being related to the quadratic contribution of the non-zero modes to the ground state mass \cite{Bigazzi:2023hxt}. Specularly, the world-sheet computations missed the NNLO term inferred from the effective thermal scalar theory, the reason being that such a term could only be accounted for by quartic interactions involving the zero modes of the world-sheet scalars. The striking agreement (in the above mentioned $AdS$ cases) between the further NNNLO correction, computed using just the effective approach, and QSC results, suggests that no world-sheet contribution might arise at that order.

%In \cite{Bigazzi:2023hxt}, the authors computed the Hagedorn temperature $T_H$ for $d$-dimensional CFTs compactified on $S^{(d-1)}$, up to next-to-next-to-next-to-leading order (NNNLO)\footnote{\textcolor{red}{Actually, besides the NNLO, the proposal is not justified from first principles. Maybe, we should review the discussion in the light of the new remarks about the crucial role of the quadratic world-sheet theory.}} in the holographic limit.
%In the dual description, the final result has been deduced thanks to the interplay of an effective approach inspired by the Horowitz-Polchinski construction \cite{Horowitz:1997jc,Urbach:2022xzw,Ekhammar:2023glu} and a world-sheet method relying on the expansion of the stringy sigma model up to second order in quantum fluctuations around a reference configuration \cite{Bigazzi:2022gal,Bigazzi:2023oqm}. 

The analysis carried out in \cite{Bigazzi:2023hxt} was structured in such a way to be extended to the entire class of holographic backgrounds dual to confining gauge theories. However, the complete scrutiny of the general cases, involving not only Ramond-Ramond fluxes but also non trivial Kalb-Ramond and dilaton fields, was missing. From one side, the effective approach was pushed beyond NLO only in the $AdS_{d+1}$ cases. From the other, on the world-sheet, it was (reasonably) assumed - skipping the details of an explicit check - that, just as in the standard semiclassical limit, a crucial conformal anomaly cancellation condition, allowing for the winding ground state zero-point energy to stay finite, was holding also in the Hagedorn regime. In cases where the dilaton field does not contribute to it, this condition translates into the equality between the sum of the squared masses of the bosonic fields on the world-sheet and that of the fermionic ones. This is what we will refer to as the mass-matching condition. Once this condition is satisfied, the above mentioned NNLO $\log2$ term, which arises from the expansion of the ground state zero-point energy, is automatically obtained.  

However, as already outlined in \cite{Bigazzi:2022gal,Bigazzi:2023oqm}, at the Hagedorn point the semiclassical string picture breaks down, essentially because the length of the string along the time circle becomes too small. This in turn requires a careful extrapolation to the Hagedorn regime of the results obtained by means of semiclassical tools. This extrapolation consistently reproduces the NLO results \cite{Urbach:2022xzw,Urbach:2023npi} obtained by means of the effective approach in a class of examples. Moreover it allows for the already mentioned spectacular agreement between $AdS$ and QSC results to NNNLO. This makes us confident that such an extrapolation should work in full generality. 

In the present work we make several progresses along these lines. On the world-sheet side, we clarify how the extrapolation to the Hagedorn configuration has to be made in the most general case, allowing for the mass-matching condition to be safely guaranteed. Working with general confining backgrounds, we will show how the NNLO $\log 2$ term, which is proportional to the sum of the squared bosonic masses once the mass-matching condition is fulfilled, can in turn be expressed, in covariant form, in terms of the pullback on the Hagedorn configuration of a combination of the background Ricci tensor and the Kalb-Ramond field strength. This in turn allows us to deduce, from the world-sheet, a general expression for the effective mass of the thermal scalar field, which precisely agrees with that recently proposed in \cite{Harmark:2024ioq}. Moreover, we will extend the effective approach to NNNLO for generic holographic confining backgrounds with Ramond-Ramond fluxes and non trivial Kalb-Ramond and dilaton fields. The interplay of world-sheet and effective methods considered in this paper allows us to deduce a general formula for the Hagedorn temperature, to NNNLO in the strong coupling expansion. At each order, the expansion coefficients are given in terms of covariant geometrical objects of the dual holographic backgrounds. We will discuss several applications of our results on a class of relevant examples. In turn, for the global $AdS_{d+1}$ cases, we will extend the world-sheet analysis to include quartic interaction terms involving the bosonic zero modes, getting the full expression for $T_H$ at NNLO using just world-sheet computations.

This work is organized as follows. In section \ref{sec:results}, for readers' convenience, we will collect most of our general results on the Hagedorn temperature. In section \ref{sec:quadmodel}, as a review, we will consider type II Green-Schwarz (GS) superstrings on generic supergravity backgrounds, deriving the spectrum of quadratic world-sheet fluctuations around a classical string configuration. In particular, we will recover general expressions for the traces of the bosonic and fermionic squared mass matrices, expressed  in terms of the pullback of certain covariant combinations of the background fields. In section \ref{sec:extrapol} we will specialize the above results to Euclidean-time-winding strings on holographic confining backgrounds and discuss how the extrapolation to the Hagedorn regime is carried out in full generality. We will also see how
the interplay of world-sheet and effective methods allows to determine $T_H$ in a perturbative expansion at strong coupling. We will also notice how the world-sheet results in the Hagedorn regime naturally provide the general expression (\ref{geneffmass}) for the effective mass of the thermal scalar field. In section \ref{sec:effnlo} we will develop the effective approach for general confining backgrounds and present our perturbative solution of the thermal scalar equation of motion. In section \ref{sec:examples} we apply our results to several examples of holographic confining theories including a class of recently discovered models; in the global AdS cases, we also perform the computation of the whole NNLO correction to $T_H$ coming from the quartic world-sheet sigma-model. Section \ref{sec:conclu} contains our concluding remarks. Other useful material is contained in the appendices. In appendix \ref{app:notations} we collect our notations and conventions. In appendix \ref{app:sugraeom} we write down the supergravity equations of motion in string frame. In appendix \ref{app:useful} we provide useful formulae for Gamma matrix manipulations. In appendix \ref{app:trace} we discuss the spectrum of fermionic world-sheet fields in a specific example.

%%%%%%%%%%%%%%%%%%%%%%%%%%%%%%%%%%%%%%%%%%%%%%%
\section{Overview of results for the Hagedorn temperature}
\label{sec:results} 
In this section we collect results of this paper about the Hagedorn temperature of confining theories (and CFTs on spheres) with supergravity duals.
The backgrounds include the metric $g_{\mu\nu}$, RR-fields, dilaton $\phi$ and the NSNS two-form $B_{(2)}$ with field strength $H_3$.
We consider a metric, with compact Euclidean time t, whose radial holographic direction $\rho$ attains the value $\rho_0$ at the ``tip'' of the geometry, which corresponds to the IR regime of the dual field theory. 
For simplicity, let us say that $\rho_0$ is the minimal value of $\rho$.
In the known cases, the geometry has a shrinking cycle of dimension $d-1$ at $\rho_0$.
If there is no such cycle, then $d=1$.

Then, the square of the inverse Hagedorn temperature, $\beta_H^2=1/T_H^2$, can be expressed as 
\be 
%{\boxed{
		\beta_H^2 = \frac{4\pi}{T_s} \( 1 + c_0 \, \alpha'^{\frac12} + \tilde c_1 \, \alpha'\ + \tilde c_2 \, \alpha'^{\frac32} + ...\) 
%}}
\,,
\ee
where $T_s$ is the ``string tension'' defined by $2\pi \alpha'T_s =g_{tt}(\rho_0)$.

The coefficients $c_0,\tilde c_i$ are written in terms of derivatives of the following functions
\bea
&& G(\rho) \equiv \hat V^{\mu} \hat V^{\nu}\left(g_{\mu\nu} + B_{\mu\rho}B_{\nu\eta}g^{\rho\eta} \right)\,, \qquad g^{\rho\rho}\sqrt{|g|}\,(\rho)\,, \nonumber\\
&&\hat R(\rho) \equiv \hat V^{\mu} \hat V^{\nu}\left(R_{\mu\nu} -\frac{1}{2} |H_{(3)}|_{\mu\nu}^2  \right)\,, \qquad   \phi(\rho)\,.
\eea  
In these expressions $\hat V$ is the vector tangent to the compactified thermal direction, normalized as $\hat V^2=1$ at the IR point of the metric $\rho_0$, and
\be
|H_{(3)}|_{\mu\nu}^2 = \frac12 H_{\mu \alpha \beta} H_{\nu}{}^{\alpha \beta} \, .
\ee
Moreover, let  $W$ be the normalized vector with non-vanishing component only along the direction $\rho$
\be 
W= W^{\rho} \partial_{\rho} = \frac{1}{\sqrt{g_{\rho\rho}}} \partial_{\rho} \,.
\ee
Then, one can define the following coefficients
\bea 
%\tilde g^{(0)}_n & = & \frac{1}{n!} W^{\mu_1}...W^{\mu_n} \nabla_{\mu_1}...\nabla_{\mu_n} \tilde g\, |_{\rho_0}\,, \\
\tilde G^{(0)}_n & = & \frac{1}{n!} W^{\mu_1}...W^{\mu_n} \nabla_{\mu_1}...\nabla_{\mu_n} G\, |_{\rho_0}\,, \\
\hat g_m &=& \frac{1}{m!} W^{\mu_1}...W^{\mu_m} \nabla_{\mu_1}...\nabla_{\mu_m} \left(g^{\rho\rho}\sqrt{g} \right) \, |_{\rho_0}\,, \\
\phi_n &=& \frac{1}{n!} W^{\mu_1}...W^{\mu_n} \nabla_{\mu_1}...\nabla_{\mu_n} \phi\, |_{\rho_0}\,, \\
\hat R_0 & = & \hat R |_{\rho_0}\,.
\eea
Since there is a vanishing $(d-1)$-cycle in the geometry, the first  $d-1$ derivatives $\hat g_m$ are zero, beginning with $m=0$; let us call $C$ the value of the (non-zero) $d-$th derivative, $C=\hat g_d $.
Then calling $n=m-d$, for $n\geq 0$ one can define the coefficients
\be 
\tilde g_{n} =  \frac{1}{C\,(n+d)!} W^{\mu_1}...W^{\mu_{n+d}} \nabla_{\mu_1}...\nabla_{\mu_{n+d}} \left(g^{\rho\rho}\sqrt{g} \right)\, |_{\rho_0}\,.
\ee	

Finally, the coefficients $c_0, \tilde c_i$ have the following form
\begin{align}
& c_0 =  - \frac{d}{\sqrt{2}} \(\tilde  G_2^{(0)}\)^{\hspace{-2pt}\frac12} \,,\\
& \tilde c_1  = - \frac{d(d+2)}{8} \, \frac{\tilde G^{(0)}_4}{ \tilde G^{(0)}_2} - \frac{d}{2} (\tilde g_2 - 2 \phi_2) + \frac{d^2}{4} \tilde G^{(0)}_2 -  2 \hat R_0 \log{2} \,,\\
& \tilde c_2 = - \frac{d(d+2)(d+4)}{16\sqrt{2}} \, \frac{ \tilde G^{(0)}_6}{\(\tilde G^{(0)}_2\)^{\hspace{-2pt}\frac32}} 
+ \frac{d(10+9d+2d^2)}{32\sqrt{2}} \, \frac{ \(\tilde G^{(0)}_4 \)^{\hspace{-2pt}2}}{\(\tilde G^{(0)}_2\)^{\hspace{-2pt}\frac52}} - \frac{d(d+2)}{2\sqrt{2}} \, \frac{ \tilde G^{(0)}_4\,\phi_2}{\(\tilde G^{(0)}_2\)^{\hspace{-2pt}\frac32}} + \nonumber \\
& \qquad + \frac{d^2(d+2)}{16\sqrt{2}} \, \frac{ \tilde G^{(0)}_4}{\(\tilde G^{(0)}_2\)^{\hspace{-2pt}\frac12}} + \frac{d}{\sqrt{2}} \, \frac{\phi_2^2}{\(\tilde G^{(0)}_2\)^{\hspace{-2pt}\frac12}} + \frac{d(d+1)}{4\sqrt{2}} \, \frac{ \tilde g^2_2}{\(\tilde G^{(0)}_2\)^{\hspace{-2pt}\frac12}} + \frac{d^2}{4\sqrt{2}} \,  \tilde g_2\,\(\tilde G^{(0)}_2\)^{\hspace{-2pt}\frac12} + \nonumber \\
& \qquad + \frac{d(d+2)}{2\sqrt{2}} \, \frac{ (\phi_4 -\tilde g_4)}{\(\tilde G^{(0)}_2\)^{\hspace{-2pt}\frac12}} - \frac{d^3}{16\sqrt{2}} \(\tilde G^{(0)}_2\)^{\hspace{-2pt}\frac32} + \frac{3}{\sqrt{2}}  \, d \, \(\tilde G^{(0)}_2\)^{\hspace{-2pt}\frac12} \hat R_0 \log{2}  \,.
\end{align}
Notice that the above expression for $\tilde c_2$ provides the full NNNLO correction to $\beta_H^2$ only for exact backgrounds.

The RR-supported pp-wave backgrounds, where there is not a unique holographic direction $\rho$, can be included in this formalism by a trivial generalization, as will be briefly discussed in section \ref{sec:ppwaves}. 

%%%%%%%%%%%%%%%%%%%%%%%%%%%%%%
%%%%%%%%%%%%%%%%%%%%%%%%%%%%%%
\section{World-sheet perspective}
\label{sec:quadmodel}

Before considering the extrapolation to the Hagedorn configuration, in this section we will review some general world-sheet results obtained expanding the string sigma model to quadratic order around a classical (bosonic) configuration $X^p(\tau,\sigma)$, (where $p=0,\ldots, 9$ are the target space indices) on a generic supergravity background. 
Besides the 10-d metric $g_{pq}$, the field content of the latter will in general include a dilaton $\phi$, Ramond-Ramond fields $G_{(n)}$ and the Kalb-Ramond field $B_{(2)}$.

\subsection{Type-II superstrings in curved backgrounds}
  
%In the standard (non-democratic) formulation of type II supergravities, the background Lagrangian displays kinetic terms $G_{(n)}\wedge*G_{(n)}$, where $n=2,4$ in the IIA case and $n=1,3,5$ in the IIB one, along with Chern-Simons interactions. Nevertheless, the full ten-dimensional theory features also higher forms by means of the electromagnetic dual, that is
%\be
%\label{dualityconstraint}
%* G_{(n)} = (-1)^{\frac{n(n-1)}{2}} G_{(10-n)} \, , \quad \forall \, n = 1,...,9 \, .
%\ee
In the democratic formalism \cite{Bergshoeff:2001pv}, the dynamics of the supergravity fields is encoded in the pseudo action\footnote{The ``pseudo" indicates that this action does not describe the dynamics of the physical degrees of freedom alone, but also includes kinetic terms for the electromagnetic duals of the RR forms (see \eqref{dualityconstraint}). It is precisely this redundancy that allows us to write a simpler background Lagrangian without Chern-Simons interactions.} 
\be
S_{II} = - \frac{1}{(2\pi)^7 \alpha'^4} \int \[e^{-2\phi}\( R \, \rmd\text{vol} + 4 \, \rmd\phi \wedge *\rmd\phi - \frac12 H_{(3)} \wedge *H_{(3)}\) + \frac14 \sum_n G_{(n)} \wedge *G_{(n)}\] \, .
\ee
Here, $R$ is the Ricci scalar, $H_{(3)}=\rmd B_{(2)}$ and  the sum runs over $n=2,4,6,8$ in the IIA case\footnote{Notice that we are excluding the massive theory which features an additional parameter $G_{(0)}$.} and $n=1,3,5,7,9$ in the IIB one. The RR fields are given by
\begin{subequations}
\begin{align}
&G_{(n)}=\rmd C_{(n-1)}\, , \quad \text{for } n<3 \, ,\\
&G_{(n)}=\rmd C_{(n-1)} - H_{(3)} \wedge C_{(n-3)}\, , \quad \text{for } n\ge3 \, ,
\end{align}
\end{subequations}
being $C_{n}$ the standard RR potentials. Crucially, the equations of motion must be equipped with the duality constraints 
\be
\label{dualityconstraint}
* G_{(n)} = (-1)^{\frac{n(n-1)}{2}} G_{(10-n)} \, , \quad \forall \, n = 1,...,9 \, ,
\ee
in order to avoid a double counting of the degrees of freedom. %Notice that no Chern-Simons interaction term shows up in the pseudo action. Finally, remember that 

The explicit form of the GS superstring Polyakov action (to quadratic order in the fermionic fields $\psi$) on a generic supergravity background is given by \cite{Cvetic:1999zs}\footnote{See also \cite{Grisaru:1985fv} for a seminal work in superspace formalism, \cite{Martucci:2003gc} for the GS action in the Nambu-Goto formalism, and \cite{Wulff:2013kga} for an extension to fourth-order in the fermionic fields.}%\footnote{This version of the action holds for a generic $g_{\alpha\beta}$, given that $\{\Gamma_\alpha, \Gamma_\beta\} = 2 \, g_{\alpha\beta} \, \mathbb{I}_{32} + \mathcal O \(\bar\psi \psi\)$. Indeed, up to quadratic corrections in the fermionic fields, the latter implies $\varepsilon^{\lambda\gamma} \Gamma_\lambda \Gamma_\gamma \Gamma^{(11)} g^{\alpha\beta} \Gamma_\alpha A_\beta = 2 \, \Gamma^{(11)} \varepsilon^{\alpha\beta} \Gamma_\alpha A_\beta$, whatever $g_{\alpha\beta}$ and $A_\beta$ are.}
\be
\label{actionv2}
S= -\frac{1}{4\pi\alpha'} \hspace{-2pt} \int \hspace{-4pt} d\tau d\sigma  \, \( \sqrt{-h} \,h^{\alpha\beta} \, \mathcal G_{\alpha\beta} - \varepsilon^{\alpha\beta} \, \mathcal B_{\alpha\beta}\) \, ,
\ee
where $h^{\alpha\beta}$ is the auxiliary (intrinsic) metric on the world-sheet, the Levi-Civita symbol is such that $\varepsilon^{\tau\sigma}=+1$ and 
%Moreover, $\mathcal G_{\alpha\beta}$ and $\mathcal B_{\alpha\beta}$ are respectively the pull-backs of the super-metric $\mathcal G_{pq}$ and the super NSNS two-form potential $\mathcal B_{pq}$. To second order in the fermion fields, they are given by
\begin{subequations}
\begin{align}
&\mathcal G_{\alpha\beta} = g_{\alpha\beta} - i \, \bar \psi \, \Gamma_{(\alpha} D_{\beta)} \, \psi \, , \\
&\mathcal B_{\alpha\beta} = B_{\alpha\beta} - i \, \bar \psi \, \Gamma^* \, \Gamma_{[\alpha}D_{\beta]} \, \psi \, ,
\end{align}
\end{subequations}
where $g_{\alpha\beta}=\partial_{\alpha}x^p\partial_{\beta}x^q g_{pq}$ is the induced metric and $B_{\alpha\beta}$ is the pullback of the NSNS two-form potential $B_{pq}$. In the democratic formalism, the covariant derivative acting on the fermionic fields $\psi$ is given by\footnote{This can be also deduced as the pullback of the covariant derivative appearing in the gravitino supersymmetry variation \cite{Bergshoeff:2001pv}.} 
\be
\label{covderivative}
D_\beta = \partial_\beta  + \frac14 \, \cancel{\omega}_{\beta} + \frac18 \, \mathcal S \, \cancel{H}_{\beta} + \frac{1}{16} \, e^\phi  \, \sum_{n} \, \frac{1}{n!} \, \mathcal S_n \, \, \cancel{G}_{(n)} \, \Gamma_\beta \,  ,
\ee
where, in the IIA case,
\begin{subequations}
\begin{align}
&\Gamma^* = \mathcal S = \Gamma^{(11)} \, , \quad \mathcal S_n = \( - \Gamma^{(11)}\)^{\frac{n(n-1)}{2}} \, , \\
&\label{chiaralIIA}\psi = \psi_1 + \psi_2 \, , \quad \Gamma^{(11)} \psi_I = (-1)^{I+1} \psi_I \, ,
\end{align}
\end{subequations}
while, in the IIB one,
\begin{subequations}
\begin{align}
&\Gamma^* = \sigma_3 \otimes \Gamma^{(11)} \, , \quad \mathcal S = \sigma_3 \otimes \mathbb{I}_{32} \, , \\
& \mathcal S_n = \frac12 \[ \(1+(-1)^{\frac{n(n-1)}{2}}\) i \sigma_2 - \(1-(-1)^{\frac{n(n-1)}{2}}\) \sigma_1 \] \otimes \mathbb{I}_{32} \, \, ,  \\
&\label{chiralIIB}\psi = \binom{\psi_1}{\psi_2} \, , \quad \Gamma^{(11)} \psi_I = \psi_I \, , \\
& \sigma_1 = \(\begin{matrix} 0 & 1 \\ 1 & 0\end{matrix}\)  \, , \quad \sigma_2 = \(\begin{matrix} 0 & -i \\ i & 0\end{matrix}\)  \, , \quad \sigma_3 = \(\begin{matrix} 1 & 0 \\ 0 & -1\end{matrix}\)  \, .
\end{align}
\end{subequations}
In both frameworks, $\Gamma^{(11)}%=-\Gamma^{\underline 0} \cdots \Gamma^{\underline 9}
$ is the ten-dimensional chirality matrix (see appendix \ref{app:notations}) and $\psi_I$, $I=1,2$, are ten-dimensional Majorana-Weyl spinors, 
%Since the background is taken to be a solution of the 10-d supergravity equations of motion, the GS action is also invariant under k-symmetry. Despite explicitly gauge fixing the latter is a subtle issue in general, this would guarantee that the physical fermionic degrees of freedom correspond to those of 8 two-dimensional Majorana spinors. These are the supersymmetric partners of the 8 physical bosons. 
while $\omega$ denotes the spin-connection 
\be
\label{Clifford}
\omega_{p\underline{ab}} = \eta_{\underline{ac}} \, \eta_{\underline{bd}} \, g^{ql} e^{\underline d}{}_{l} \(- \partial_p e^{\underline c}{}_q + e^{\underline c}{}_k \Gamma^{k}{}_{pq}\) \, , \quad \Gamma^k{}_{pq} = \frac12 \, g^{kl} \(\partial_p g_{lq} + \partial_q g_{lp} -\partial_l g_{pq}\) \, .
\ee 
In this section, latin underlined (normal) indices stand for flat (curved) ten-dimensional ones, and Greek indices stand for world-sheet ones. The link among them is given by the pullback on the world-sheet and a proper choice of the vielbein $e^{\underline p}{}_q$ such that
\be
\label{cliffordv2}
\Gamma_{\alpha} = \partial_\alpha x^q \, \Gamma_q = \partial_\alpha x^q \, e^{\underline p}{}_q \, \Gamma_{\underline p} \, , \quad \{\Gamma_\alpha, \Gamma_\beta\} = 2 \, g_{\alpha\beta} \, \mathbb{I}_{32} \, , \quad \{\Gamma_{p}, \Gamma_{q}\} = 2 \, g_{pq} \, \mathbb{I}_{32} \, .
\ee
The ``slash'' notation summarizes the contractions 
\be
\cancel{\omega}_\beta = \omega_{\beta\underline{ab}} \, \Gamma^{\underline{ab}} \, , \quad \cancel{H}_{\beta} = H_{\beta\underline{ab}} \, \Gamma^{\underline{ab}} \, , \quad \cancel{G}_{(n)} = G_{\underline{a_1 ... a_n}} \, \Gamma^{\underline{a_1 ... a_n}} \, .
\ee

\subsection{Quadratic fluctuations: general considerations}

The total action in \eqref{actionv2} is classically Weyl invariant. Expanding around a classical string configuration $X\equiv X^{p}(\tau,\sigma)$, the requirement that Weyl invariance is preserved also at the quantum level notoriously gives rise to beta-function conditions which are in turn interpreted, to leading order in $\alpha'$, as the equations of motion of the supergravity background. As a consequence of these conditions, the world-sheet zero point energy is finite. As we will review in the following, this can be checked explicitly by computing the masses of the small fluctuations around a reference string configuration.

Following the same strategy outlined, for instance, in \cite{Drukker:2000ep,Forini:2015mca,Gautason:2021vfc}, we will work in conformal gauge, keeping the auxiliary world-sheet metric independent from the induced one, hence working off-shell with respect to the Virasoro constraints. Within this framework, the intrinsic metric is partly gauge-fixed as
\be
\label{gaugefixing}
h_{\alpha\beta} = e^{2 \, \Omega(\tau,\sigma)} \, \eta_{\alpha\beta} \, ,
\ee
which implies
\be
\label{flatgauge}
\sqrt{-h} \, h^{\alpha\beta} = \eta^{\alpha\beta} \, .
\ee

Since the background is taken to be a solution of the 10-d supergravity equations of motion, the GS action is also invariant under $\kappa$-symmetry. The latter is a non-linear symmetry in the Lagrangian description and is not evident in the field equations. The only exception is the special case in which the auxiliary metric would be conformal to the induced metric, which would ensure the presence of a projector onto a single sixteen-component spinor in the action. Nevertheless, this is not the case here. Therefore, it is necessary to insert by hand a constraint which guarantees that the physical fermionic degrees of freedom correspond to those of eight two-dimensional Majorana spinors, which are the supersymmetric partners of the eight physical bosons.\footnote{In flat space all the spurious degrees of freedom can be readily eliminated working in light-cone gauge. In more general cases, implementing the latter can be non-trivial (see for instance \cite{Metsaev:2000yf}).} This is a subtle issue which will be discussed elsewhere \cite{w}.
%In this way, just the two-dimensional diffeomorphism invariance has been fixed and the intrinsic metric is generically free to fluctuate.

%It is in this gauge %and in particular with reference to the flat world-sheet metric $\eta_{\alpha\beta}$,
%that in the two following subsections we will provide the masses of the bosonic and fermionic fluctuations around the classical configuration. As usual not all the string fluctuations will correspond to physical observables. We defer the counting of the physical degrees of freedom to the final subsection.

%{\blu This paragraph is not clear to me. The above gauge choice \emph{is} a matter of conveneince, because it simplifies the kinetic terms for the field equations of the fluctuations. However, it is not impossible to do differently.}
%Actually, the above gauge-fixing is not just a matter of convenience. So far, we referred to \emph{quantum fluctuations} as the classical solutions to the Euler-Lagrange equations from the quadratic Polyakov action promoted to the quantum word. Between the lines, this requires that such equations of motion must be solved. Of course, we cannot do that in general.

In conformal gauge, the fluctuations generically satisfy coupled equations of motion in two-dimensional flat space. %In this way, they can be expressed as usual as a superposition of normal modes. Each of them will correspond to a ladder operator identified by a certain frequency. In turn, the latter is fixed by the mass of the fluctuation at hand.  As we will see, the Hagedorn temperature of the model is basically fixed by the flat-space zero-point energy of $8+8$ massive bosonic and fermionic modes. 
As we will see in the following, these equations can be diagonalized, and thus the system can be easily canonically quantized, if the background fields in the NSNS sector satisfy %certain conditions on the background fields are satisfyed.
%We will show that a generic NSNS sector of the supergravity background represents an obstruction to {\blu canonical} quantization. {\blu (Path integral quantization using the heat kernel should be possible in general.)} 
a certain set of conditions. As we will see, these requests are met in all the backgrounds of interest we have in mind. Nevertheless, in what follows, we will discuss the problem in full generality applying these conditions just at the very end of each section.  Indeed, despite we fix the world-sheet metric as above, we will keep the covariance in the target space to be as general as possible.

\subsection{Bosonic sector}

From the total action \eqref{actionv2} with the gauge choice in \eqref{gaugefixing}, the field equations for the bosonic modes read\footnote{Here, we have introduced the components $H_{kpq} = 3 \, \partial_{[ k} B_{pq]}$ of the field strength $H_{(3)}=dB_{(2)}$. Notice that $\varepsilon^{\alpha\beta} \partial_\alpha x^p \partial_\beta x^q \(\partial_p B_{kq} - \frac12 \partial_k B_{pq}\) = - \frac12 \varepsilon^{\alpha\beta} \partial_\alpha x^p \partial_\beta x^q H_{kpq}$.}
\be
-\partial_{\alpha} \( \eta^{\alpha\beta} \, \partial_\beta x^q \, g_{kq}(x) \) + \frac12 \, \partial_\alpha x^p \, \partial_\beta x^q \, \eta^{\alpha\beta} \, \partial_k g_{pq}(x) = \frac12 \, \varepsilon^{\alpha\beta} \, \partial_\alpha x^p \, \partial_\beta x^q \, H_{kpq}(x) \, .
\ee
Let us parameterize the fluctuations around the reference configuration $X$ as 
\be
x^k = X^k + \xi^k \, .
\ee
The equations of motion for $X$ read\footnote{In the absence of extrinsic curvature of the background world-sheet, which is the case for all the configurations we have in mind, \eqref{classicaleom} also holds in the Nambu-Goto formulation. See \cite{Forini:2015mca, Gautason:2021vfc} and \cite{w} for related comments.}
%\footnote{The background embedding X must satisfy the classical field equation that follows from the Nambu- Goto action \cite{Forini:2015mca, Gautason:2021vfc} and is independent of the auxiliary metric in order to avoid a diffeomorphism anomaly. In the absence of extrinsic curvature of the background world-sheet, which is the case in all configurations we have in mind, \eqref{classicaleom} also holds. This will be discussed in more detail in \cite{w}.}  
%\footnote{They can be also expressed in terms of the extrinsic curvature of the background. In the gauge-fixing \eqref{gaugefixing}, it reads $K^l_{.\,\alpha\beta} = \partial_\alpha \partial_\beta X^l + \partial_\alpha X^p \partial_\beta X^q \Gamma^l_{pq}(X)$. In the general case, it appears also a contribution coming from the Christoffel symbols related to the intrinsic world-sheet metric.} 
%{\blu Are these identical to the equations of motion of the Nambu-Goto string? I would expect that one wants the background embedding $X$ to be a classical solution of the NG field equations.}
\be \label{classicaleom}
-\eta^{\alpha\beta} \partial_\alpha \partial_\beta X^k = \partial_\alpha X^p \partial_\beta X^q \(  \frac12 \, \varepsilon^{\alpha\beta} H^k{}_{pq}(X) + \eta^{\alpha\beta}\, \Gamma^{k}{}_{pq}(X)\) \, , \quad \forall \, k,p,q = 0,\ldots,9 \,,
\ee
and the linearized equations of motion for the fluctuations are
\be \label{linearboseom}
-\eta^{\alpha\beta} \partial_\alpha \partial_\beta \xi^k = \eta^{\alpha\beta} \(\Lambda_\alpha\)^{k}{}_l \, \partial_\beta \xi^l + \Omega^k{}_l \, \xi^l \, , \quad \forall \, k,l = 0,\ldots,9 \,,
\ee
where
\begin{subequations}
\begin{align}
\label{Lambdabose}
&\hspace{-4.5pt}\(\Lambda_\alpha\)^{k}{}_l = 2 \, \partial_\gamma X^p \(\frac12 \, \varepsilon^{\gamma\lambda} H^k{}_{pl}(X) + \eta^{\gamma\lambda} \, \Gamma^{k}{}_{pl}(X)\) \eta_{\lambda\alpha} \, ,\\
&\Omega_{kl} = \partial_\alpha X^p \partial_\beta X^q \,  \partial_l \hspace{-2pt} \(\frac12 \, \varepsilon^{\alpha\beta} H_{kpq}(X) + \eta^{\alpha\beta} \, \Gamma_{kpq}(X)\) + \eta^{\alpha\beta} \partial_\alpha \partial_\beta X^q \, \partial_l g_{kq}(X)\,. 
\end{align}
\end{subequations}
Here, and throughout this section, the spacetime indices are raised (lowered) with $g^{pq}(X)$ ($g_{pq}(X)$). %Moreover, they span all the possible ten directions. 

%As we discussed before, this is not the end of the story. Indeed, we need to translate the just found equations of motion in a set of flat-space Klein-Gordon equations.
%To get rid of the term proportional to $\partial_\beta \xi^q$, we can proceed as follows. 

Let us now rotate the fluctuations according to a linear transformation $R$ whose matrix representation is such that
\be
\label{R}
\xi^k = R^k{}_l \, \zeta^l \, , \quad \partial_\lambda R^k{}_l = - \partial_{\alpha} X^p \, \eta_{\beta\lambda} \(\frac12 \, \varepsilon^{\alpha\beta} H^k{}_{pq}(X) + \eta^{\alpha\beta} \, \Gamma^{k}{}_{pq}(X)\) R^q{}_l \, .
\ee
Then, the equations of motion take the Klein-Gordon form
\be
\(-\eta^{\alpha\beta}\partial_{\alpha}\partial_\beta + R^{-1} \mathcal M_b \, R \) \zeta = 0 \, , 
\ee
where %the mass operator is similar to the matrix $\mathcal M_b$ such that
\be
\(\mathcal M_b\)^k{}_l = - \eta^{\alpha\beta} \partial_\alpha X^p \partial_\beta X^q \(R^k{}_{plq}(X) + \frac14 H^k{}_{pm}(X) H^m{}_{ql}(X)\) - \varepsilon^{\alpha\beta} \partial_\alpha X^p \partial_\beta X^q \nabla_{[l}^{\vphantom{k}} H^k{}_{p]q}(X) \, ,
\ee
being
\be
\nabla_{l}^{\vphantom{k}} H^k{}_{pq}(X) = \partial_l^{\vphantom{k}} H^k{}_{pq}(X) + \Gamma^k{}_{lm}(X) H^m{}_{pq}(X) - \Gamma^m{}_{lq}(X) H^k{}_{pm}(X) - \Gamma^m{}_{pl}(X) H^k{}_{mq}(X) \, .
\ee

%Notice the presence of mass terms which are quadratic in the Kalb-Ramond field, although it appears only linearly in the Polyakov action. This is because the first derivatives of the map $R$ are linear in the components of $H_{(3)}$. Further, they are also linear in the Christoffel symbols. Similarly, this is what produces the Levi-Civita connection terms in the covariant derivative of $H_{(3)}$, as well as the quadratic terms in the Christoffel symbols contained in the Riemann tensor
%\be
%R^k{}_{plq}(X) = \partial_l^{\vphantom k} \Gamma^k{}_{pq} (X) - \partial_q^{\vphantom k} \Gamma^k_{.\,lp} (X) + \Gamma^k{}_{lm} (X) \, \Gamma^m{}_{pq} (X) - \Gamma^k_{.\,qm} (X) \, \Gamma^m_{.\,lp} (X) \, .
%\ee

%Thanks to the similarity relation, $\mathcal M_b$ and the mass operator do have the same eigenvalues. This means that the masses of the bosonic modes do not depend on the explicit expression for $R$, once the definitional equation in \eqref{R} is solved.

The computation of the masses $\mu_b$ of the bosonic fluctuations is thus reduced to the diagonalization of $\mathcal M_b$. The problem can be addressed case by case given an explicit expression of the Kalb-Ramond field. Nevertheless, the sum of the bosonic squared masses, which coincides with the trace of $\mathcal M_b$, is immediately given by
\be
\label{pullbackEeq}
\sum_b \mu_b^2 = -\eta^{\alpha\beta} \partial_\alpha X^p \partial_\beta X^q \[R_{pq}(X) - \frac12 |H_{(3)}|_{pq}^2(X)\] - \frac12 \, \varepsilon^{\alpha\beta} \, \partial_\alpha X^p \partial_\beta X^q \, \nabla_k^{\vphantom k} H^k{}_{pq}(X) \, ,
\ee
where
\be
R_{pq}(X) = R^k{}_{pkq}(X) \, , \quad  |H_{(3)}|_{pq}^2(X) = \frac12 H_{pkl}(X) H_q{}^{kl}(X) \, .
\ee
This agrees with the result derived in \cite{Gautason:2021vfc} from a beta-function perspective.

As a final step, we must verify that a solution of the differential equation \eqref{R} exists. Let us demand the NSNS sector to satisfy $\rmd \Lambda = 0$, with $\Lambda$ defined in (\ref{Lambdabose}), that is
\be \label{rotablebos}
\partial_\alpha \[\partial_\gamma X^p \(\frac12 \, \varepsilon^{\gamma\lambda} H^k{}_{pl}(X) + \eta^{\gamma\lambda} \, \Gamma^{k}{}_{pl}(X)\) \eta_{\lambda\beta}\] = 0 \, , \quad \forall k \, , l \, .
\ee
Under these assumptions, equation \eqref{R} is integrable and is solved by
%Nevertheless, starting from \eqref{R}, a direct computation gives\footnote{\textcolor{red}{Actually, another term appears. It looks like $\varepsilon^{\alpha\beta} \Pi_\alpha \Pi_\beta R$, defining the matrix $\Pi_\alpha$ such that $\(\Pi_\lambda\)^k_{.\,q} = \eta_{\lambda\beta} \partial_\alpha X^p \(\frac12 \varepsilon^{\alpha\beta} H^k{}_{pq}(X) + \eta^{\alpha\beta} \Gamma^k{}_{pq}(X) \)$. Since $\[\Pi_\alpha, \Pi_\beta\]=0$, this contribution should vanish. Do you agree? In any case it should not be a problem, at most there would be another condition.}}
%\be \label{bosintcondexplicit}
%\varepsilon^{\alpha\beta} \partial_\alpha \partial_\beta R^k{}_l = \[\partial_\alpha X^p \partial_\beta X^m \, \partial_m \(\frac12 \eta^{\alpha\beta} H^k{}_{pq}(X) + \varepsilon^{\alpha\beta} \Gamma^k{}_{pq}(X)\) + \frac12 \, \eta^{\alpha\beta} \partial_\alpha \partial_\beta X^p \, H^k{}_{pq}(X) \] \hspace{-4pt} R^q{}_l \, .
%\ee

%Clearly, it is identically zero if both the background configuration of the field-strength $H_{(3)}$ and of the Levi-Civita connection vanish. If so, however, no rotation should have been introduced. 

%To be more general, let us focus on II supergravity backgrounds satisfying conditions \ref{cond1} and \ref{cond2}. This guarantees the existence of the linear map $R$ and in turn the existence of a set of bosonic fluctuations solving two-dimensional Klein-Gordon equations for whatever background belonging to this class. Indeed, under these assumptions, \eqref{bosintcond} is trivially satisfied and it follows that
\be
R = \text{Exp}\[- \partial_{\alpha} X^p \, \eta_{\beta\lambda} \(\frac12 \, \varepsilon^{\alpha\beta} H_p(X) + \eta^{\alpha\beta} \, \Gamma_p(X)\) \sigma^\lambda\] \, , \quad \sigma^\lambda=\{\tau,\sigma\} \, ,
\ee
being $H_p$ and $\Gamma_p$ matrices with entries $H^k{}_{pq}$ and $\Gamma^{k}{}_{pq}$. Let us stress that \eqref{rotablebos} is not given by first principles. Rather, it is an extra condition which removes the obstructions to the canonical quantization of the bosonic fluctuations.\footnote{If we interpret $\Lambda$ as a connection in the equations of motion \eqref{linearboseom}, then \eqref{rotablebos} guarantees that its curvature vanishes.}

\subsection{Fermionic sector}
\label{sec:fermsector}

%{\blu In the entire subsection, you never mention kappa symmetry. However, kappa symmetry is crucial for arguing that the number of 2-d fermions should be 8, not 16. The calculations you do give 16 fermions. The ``projection'' you do in words before (2.46) is probably not consistent with the field equations, or it gives the wrong masses. At least this is what we agreed on in our recent discussion.}

The fermionic part of the action in \eqref{actionv2} is already quadratic in the fluctuations so that every function of the coordinates has to be evaluated on the reference string configuration. %From now on, each of them will tacitly refer to its background configuration.%Let us stress that what follows does not rely on a particular choice of the reference configuration. 

Making use of the chiral components $\psi_I$, $I=1,2$, of the Majorana spinor $\psi$, we can enclose both the IIA fermionic action and the IIB one in a single expression, that is
\be
\label{genSF}
S_F^{(2)} = \frac{i}{4\pi\alpha'} \hspace{-4pt} \int \hspace{-2pt} d\tau d\sigma \, \bar \psi_I  \(\sqrt{-h} \, h^{\alpha\beta} \hspace{-2pt} + \hspace{-2pt} (-1)^{I+1} \varepsilon^{\alpha\beta}\) \[\delta_{IJ} \Gamma_\alpha \nabla^{(I)}_\beta \, \psi_J \hspace{-2pt} + \hspace{-2pt} (-1)^{I+1} \(\sigma_1\)_{IJ} \Gamma_\alpha \widetilde \Gamma_I \Gamma_\beta \, \psi_J \] \, ,
\ee
where the sum over $I,J=1,2$ is understood. Here,
\be
\nabla^{(I)}_\beta = \partial_\beta + \frac14 \, \cancel{\omega}_{\beta} + \frac18 \, (-1)^{I+1} \, \cancel{H}_{\beta}
\ee
and
\begin{align}
&\widetilde \Gamma_I = \frac{1}{16} \, e^{\phi} \sum_n \, (-1)^{f_{(n,I)}} \, \frac{1}{n!} \, \cancel{G}_{(n)} \, , \hspace{-40pt} && n = 2,4,6,8 \quad \text{for IIA, } \\
& \, \hspace{-40pt} && n = 1,3,5,7 \quad \text{for IIB,} \nonumber
\end{align}
with
\be
f_{(n,I)} = I \[\frac{n(n-1)}{2} + n +1\] + n+1 \, .
\ee

Since
\be
\(\eta^{\alpha\beta} \hspace{-2pt} + \hspace{-2pt} (-1)^{I+1} \varepsilon^{\alpha\beta}\) \Gamma_\alpha \, \nabla^{(I)}_\beta = \((-1)^{I+1} \Gamma_\tau + \Gamma_\sigma\) \((-1)^I \nabla^{(I)}_\tau + \nabla^{(I)}_\sigma\)
\ee
and
\be
\(\eta^{\alpha\beta} \hspace{-2pt} + \hspace{-2pt} (-1)^{I+1} \varepsilon^{\alpha\beta}\) \Gamma_\alpha \widetilde \Gamma_I \Gamma_\beta = \((-1)^{I+1} \Gamma_\tau + \Gamma_\sigma\) \widetilde \Gamma_I \((-1)^I \Gamma_\tau + \Gamma_\sigma\) \,,
\ee
we can simplify the action \eqref{genSF} and derive the equations of motion in the form
\be
\label{diraceom}
\begin{cases}
\(\nabla^{(1)}_\tau - \nabla^{(1)}_\sigma\) \psi_1 + \widetilde \Gamma_1 \(\Gamma_\tau - \Gamma_\sigma\) \psi_2 = 0 \vspace{5pt}\\
\(\nabla^{(2)}_\tau + \nabla^{(2)}_\sigma\) \psi_2 - \widetilde \Gamma_2 \(\Gamma_\tau + \Gamma_\sigma\) \psi_1 = 0
\end{cases} \, .
\ee
Notice that, since we are working off-shell with respect to the Virasoro constraints, the operators $\((-1)^{I+1} \Gamma_\tau + \Gamma_\sigma\)$ are generically not nilpotent. They can be thus removed from the equations of motion.

By applying $\(\partial_\tau + \partial_\sigma\)$ to the first and $\(\partial_\tau - \partial_\sigma\)$ to the second equation and iterating, we can write 
\be
\label{kgeom}
\begin{cases}
\displaystyle -\eta^{\alpha\beta} \partial_\alpha \partial_\beta \psi_1 + \eta^{\alpha\beta} C_{1\alpha} \, \partial_\beta \psi_1 + B_1 \, \psi_2 + \tilde A_1 \, \psi_1 = 0 \vspace{5pt} \\
\displaystyle -\eta^{\alpha\beta} \partial_\alpha \partial_\beta \psi_2 + \eta^{\alpha\beta} C_{2\alpha} \, \partial_\beta \psi_2+ B_2 \, \psi_1 + \tilde A_2 \, \psi_2 = 0
\end{cases} ,
\ee
where
\begin{subequations}
\begin{align}
&\hspace{-4pt}\begin{cases}
\displaystyle \tilde A_1 = - \(\eta^{\alpha\beta} - \varepsilon^{\alpha\beta}\) \[\widetilde \Gamma_1 \Gamma_\alpha \widetilde \Gamma_2 \Gamma_\beta + \frac14 \, \partial_\beta \(\frac12 \, \cancel{H}_\alpha + \cancel{\omega}_\alpha\)\] \vspace{5pt}\\
\displaystyle \tilde A_2 = - \(\eta^{\alpha\beta} + \varepsilon^{\alpha\beta}\) \[\widetilde \Gamma_2 \Gamma_\alpha \widetilde \Gamma_1 \Gamma_\beta - \frac14 \, \partial_\beta \(\frac12 \, \cancel{H}_\alpha - \cancel{\omega}_\alpha\)\]
\end{cases},\\
&\label{BI} B_I =\(\eta^{\alpha\beta}+(-1)^I \varepsilon^{\alpha\beta}\) \[(-1)^I \partial_\beta \(\widetilde \Gamma_I \, \Gamma_\alpha\) - \frac14 \, \widetilde \Gamma_I \, \Gamma_\alpha \(\frac12 \, \cancel{H}_\beta + (-1)^I \cancel{\omega}_\beta\)\] \, , \\
\label{Cfermi}
&C_{I\alpha} = \frac14 (-1)^I \eta_{\alpha\beta }\(\eta^{\beta\lambda} + (-1)^{I+1} \varepsilon^{\beta\lambda}\) \(\frac12 \, \cancel{H}_\lambda + (-1)^{I+1} \cancel{\omega}_\lambda\) \, .
\end{align}
\end{subequations}

Similarly to the previous section, we aim to get rid of the terms proportional to $\partial_\beta \psi_I$ through a $(\tau,\sigma)$-dependent rotation of the spinors. In particular, introducing
\be \label{fermionicrotation}
\psi_I = R_I \, \chi_I \, , \quad \partial_\beta R_I = \frac18 (-1)^{I} \eta_{\beta\lambda} \(\eta^{\lambda\gamma} + (-1)^{I+1} \varepsilon^{\lambda\gamma}\) \(\frac12 \, \cancel{H}_\gamma + (-1)^{I+1} \cancel{\omega}_\gamma\) R_I \, ,
\ee
the fermionic equations of motion reduce to\footnote{We make use of $\eta_{\beta\lambda} \(\eta^{\alpha\beta}  \pm \varepsilon^{\alpha\beta}\) \(\eta^{\gamma\lambda} \pm \varepsilon^{\gamma\lambda}\)=0$.}
\be
\label{finalkgeom}
\(-\eta^{\alpha\beta} \partial_\alpha \partial_\beta  + \mathcal R^{-1} \mathcal M_f \mathcal R\) \chi = 0 \, ,
\ee
where
\be
\mathcal R = \(\begin{matrix} R_1 & 0 \\ 0 & R_2\end{matrix}\) \, , \quad \mathcal M_f = \(\begin{matrix} A_1 & B_1 \\ B_2 & A_2\end{matrix}\) \, , \quad \chi =\(\begin{matrix} \chi_1 \\ \chi_2\end{matrix}\) 
\ee
and
\be
\begin{cases}
\label{A1A2}
\displaystyle A_1 = - \(\eta^{\alpha\beta} - \varepsilon^{\alpha\beta}\) \[\widetilde \Gamma_1 \Gamma_\alpha \widetilde \Gamma_2 \Gamma_\beta + \frac18 \, \partial_\beta \(\frac12 \, \cancel{H}_\alpha + \cancel{\omega}_\alpha\)\] \vspace{5pt} \, ,\\
\displaystyle A_2 = - \(\eta^{\alpha\beta} + \varepsilon^{\alpha\beta}\) \[\widetilde \Gamma_2 \Gamma_\alpha \widetilde \Gamma_1 \Gamma_\beta - \frac18 \, \partial_\beta \(\frac12 \, \cancel{H}_\alpha - \cancel{\omega}_\alpha\)\] \, .
\end{cases}
\ee
%Notice that this is the natural notation for the IIB case. Nevertheless, no one prohibits stacking the defined-chirality components of $\chi$ in that way, even in the IIA case. Again, one can resolve each fermionic mass $\mu_f$ by diagonalizing the square mass-operator. However, for our purposes, the knowledge of its trace is enough.
At this point, we can compute the trace of the fermionic mass-squared matrix from \eqref{finalkgeom}. Since at this stage we have not fixed the $\kappa$-symmetry, the spinor $\chi$, which has 32 real components, can be viewed as sixteen 2-d spinors. Because the mass of each 2-d spinor multiplies a $2\times2$ unit matrix in the 2-d Dirac equation, tracing over $\mathcal{M}_f$ gives twice the sum of the fermionic mass-squares. Hence, we obtain
%\footnote{This is equivalent to define $\sum_f \mu_f^2$ as the reduced trace of $\mathcal M_f$ over the subspace of chirality $1$ or $2$. Notice that, from the cyclicity of the trace, it follows that $\Tr A_1 = \Tr A_2$.} 
%The above Klein-Gordon equation comes from the iteration of the Dirac equations in \eqref{diraceom}. This procedure causes information to be lost. Indeed, let us stress that we have embedded both $\chi_1$ and $\chi_2$ in a single spinor $\chi$. Nevertheless, the original version of the equations of motion in \eqref{diraceom} couples $\chi_1$ to $\chi_2$ and vice-versa. Therefore, just one of them is enough to encode the physical degrees of freedom. This is equivalent to rearrange them in eight two-dimensional world-sheet spinors.
%All in all, we define the sum over the %physical
%fermionic masses squared as the reduced trace over the subspace of chirality $1$,\footnote{Similarly, we could have chosen the subspace of chirality $2$. Indeed, from the ciclycity of the trace, it follows that $\Tr A_1 = \Tr A_2$, as it should be.} that is
\be \label{reducedtrace}
\sum_f \mu^2_f = \frac12\Tr \mathcal M_f = \frac12 \left( \Tr  A_1 + \Tr A_2\right) = \Tr  A_1 \, ,
\ee
where the last equality follows from the cyclicity of the trace.
Using the results collected in appendix \ref{app:useful}, we compute the trace in appendix \ref{app:redtrace}, getting
\be \label{finalsummuf}
\sum_f\mu_f^2 = - \frac12 \, e^{2\phi} \hspace{-2pt} \[\eta^{\alpha\beta} \, \partial_\alpha X^p \partial_\beta X^q \sum_n  \hspace{-2pt}  \{G_{(n)}\}^2_{pq} \hspace{-2pt} - \frac12 \, \varepsilon^{\alpha\beta} \, \partial_\alpha X^p \partial_\beta X^q \sum_n \hspace{-2pt} \[*\(G_{(n)} \wedge *G_{(n+2)}\)\]_{pq}\] \hspace{-2pt},
\ee
where
\begin{subequations}
\begin{align}
&\{G_{(n)}\}^2_{pq}= \frac12 |G_{(n)}|^2_{pq} \, , \quad \forall \, n = 1,...,9 \, ,\\
\intertext{in the democratic formalism, while}
&\{G_{(n)}\}^2_{pq}= |G_{(n)}|^2_{pq} - \frac12 g_{pq} |G_{(n)}|^2\,,\quad (n\neq 5) \, , \\
&\{G_{(5)}\}^2_{pq}= \frac12 |G_{(5)}|^2_{pq} \,  ,
\end{align}
\end{subequations}
in the non-democratic one. The above expression agrees with known results in the literature (see e.g. \cite{Gautason:2021vfc}). For an explicit realization of the above formulae in a relevant example, see appendix \ref{app:wym}.

%Let us stress that we performed the computation in full-generality, admitting the presence of non-trivial RR fluxes, dilaton, Kalb-Ramond field and spin-connection in the background. Nevertheless, we verify that the last two do not affect the above equation. Indeed, the Kalb-Ramond field and the spin connection contribute to the mass operator through traceless terms in its diagonal entries (see \eqref{A1A2} and \eqref{tracelessHw}). Moreover, notice that also the first derivative of the RR fields and the dilaton do not play any role, since they inhabit the off-diagonal blocks of the mass operator.

%Therefore, this result applies for any background. Of course, it holds for those which display just one RR field-strength. For instance, this is the case of the supergravity solution dual to $(\mathcal N = 4)$ SYM on $S^3$ or Witten-Yang-Mills theory. The latter has been addressed explicitly in Appendix \ref{app:wym}. In this way, we provide the reader with a concrete example in which our general formulae are applied. Less trivially, the same works also to backgrounds featuring more than one RR field-strengths having legs along whatever directions. This is the case of $ABJM$ theory on $S^2$, which is holographically captured by global-$AdS_4 \times CP^3$ featuring both $G_{(2)}$ flux on $CP^1 \subset CP^3$ two-cycle and $G_{(4)}$ flux on $AdS_4$~\cite{Aharony:2008ug}.

As in the bosonic case, we have to check whether the fermionic rotation $\mathcal R$ exists or not. For backgrounds satisfying $\rmd C = 0$ with $C$ defined in (\ref{Cfermi}), that is
\be
\partial_\alpha \[\eta_{\beta\lambda}\(\eta^{\lambda\gamma} + (-1)^{I+1} \varepsilon^{\lambda\gamma}\) \(\frac12 \, \cancel{H}_\gamma + (-1)^{I+1} \cancel{\omega}_\gamma\)\] = 0 \, , \quad \forall I \, ,
\ee
the solution of \eqref{fermionicrotation} is given by
\be
R_I = \text{Exp}\[\frac18 (-1)^{I} \eta_{\alpha\beta} \(\eta^{\beta\lambda} + (-1)^{I+1} \varepsilon^{\beta\lambda}\) \(\frac12 \, \cancel{H}_\lambda + (-1)^{I+1} \cancel{\omega}_\lambda\) \sigma^{\alpha}\] \,  , \quad \sigma^{\alpha} = \{\tau,\sigma\} \, .
\ee

The comments made at the end of the previous section apply here too. Moreover, since $\mathcal R$ clearly commutes with $\Gamma^{(11)}$, it still holds that
\be
\Gamma^{(11)} \chi_I = (-1)^{(\mathfrak p+1)(I+1)} \chi_I \, ,
\ee
being $\mathfrak p = 0$ for the IIA case, while $\mathfrak p = 1$ for the IIB one.

\subsection{Mass-matching condition}
\label{sub:mass}
To sum up, in the previous sections we have found that
\be
\label{bosma}
\sum_b \mu_b^2 = -\eta^{\alpha\beta} \partial_\alpha X^p \partial_\beta X^q \[R_{pq} - \frac12 |H_{(3)}|_{pq}^2\] - \frac12 \, \varepsilon^{\alpha\beta} \, \partial_\alpha X^p \partial_\beta X^q \,\nabla_k^{\vphantom k} H^{k}{}_{pq}
\ee
and
\be
\label{ferma}
\sum_f\mu_f^2 = - \frac12 \, e^{2\phi} \hspace{-2pt} \[\eta^{\alpha\beta} \, \partial_\alpha X^p \partial_\beta X^q \sum_n  \hspace{-2pt}  \{G_{(n)}\}^2_{pq} \hspace{-2pt} - \frac12 \, \varepsilon^{\alpha\beta} \, \partial_\alpha X^p \partial_\beta X^q \sum_n \hspace{-2pt} \[*\(G_{(n)} \wedge *G_{(n+2)}\)\]_{pq}\] \hspace{-2pt}.
\ee
Notice that\footnote{This holds thanks to the total antisymmetric nature of $H_{(3)}$ and metric compatibility $\nabla_\kappa g_{pq}=0$. Moreover, we used that $\Gamma^{k}{}_{k q} = \partial_q \log\(\sqrt{-g}\)$, where here $g$ is the determinant of the background metric.}
\be
\nabla_k H^{k}{}_{pq} = \frac{1}{\sqrt{-g}} \, \partial_k \( \sqrt{-g} \, H^{kml}\) g_{mp} \, g_{lq} \, ,
\ee
which implies (see \eqref{covdiv})
\be
e^{2\phi}\[*\rmd\(e^{-2\phi}*H_{(3)}\)\]_{pq} = 2 \, \nabla_{k}^{\vphantom k} \phi \, H^k{}_{pq} - \nabla_{k}^{\vphantom k} H^k{}_{pq} \, .
\ee
Therefore, the last contributions to the above expressions corresponds to the Hodge dual of the Maxwell equation in \eqref{demomaxwell} (in particular, to the terms which do not depend on the derivative of the dilaton). 

From the above expressions, using the string-frame type II Einstein equations in Ricci form in \eqref{Einsteq}, it thus follows the generalized mass-matching formula
\be
\label{massmatchplusdilaton}
\sum_b \mu_b^2 - \sum_f\mu_f^2 = 2 \, \eta^{\alpha\beta}\partial_{\alpha}X^p\partial_{\beta}X^q \, \nabla_p^{\vphantom k} \nabla_q^{\vphantom k} \phi -\varepsilon^{\alpha\beta} \partial_{\alpha}X^p\partial_{\beta}X^q \, \nabla_k^{\vphantom k} \phi \, H^k{}_{pq} \,.
\ee
This is the same result derived in \cite{Gautason:2021vfc} by means of a beta-function computation.

Let us stress that the intrinsic metric is generically parametrized as in \eqref{gaugefixing}. As a consequence, the scalar curvature of the world-sheet $R_{(h)}$ is non-trivial, that is
\be
R_{(h)} = - 2 \, e^{-2\Omega(\tau,\sigma)} \, \eta^{\alpha\beta} \partial_\alpha \partial_\beta \, \Omega(\tau,\sigma) \, .
\ee
Then, the Polyakov action should be supplemented with the Fradkin-Tseytlin counterterm action \cite{Fradkin:1984pq, Fradkin:1985ys}, which introduces a dilaton coupling on the world-sheet. At the level of the trace of the world-sheet stress-energy tensor, this compensates for the anomaly coming from the one-loop fluctuations of the string \cite{Callan:1989nz}, that is exactly what survives in the difference \eqref{massmatchplusdilaton}.% This guarantees that the logarithmic divergences in the one-loop partition function cancel between bosons and fermions.

For the class of backgrounds we will deal with in the following, the right hand side of \eqref{massmatchplusdilaton} vanishes.\footnote{The contribution to the anomaly from the Fradkin-Tseytlin action vanishes as well.} Therefore, %whose dilaton $\phi$ does not depend on the embedding directions and with vanishing internal derivative of the Kalb-Ramond field strength $H_{(3)}$ along $\nabla \phi$
we recover the mass-matching condition
\be
\label{massmatch}
\sum_b \mu_b^2 - \sum_f\mu_f^2 = 0\,.
\ee
As we will see, it is this condition which allows the total Casimir energy of the world-sheet fluctuations to remain finite, avoiding (logarithmic) UV divergences that would otherwise spoil the conformal invariance of the world-sheet theory.\footnote{See \cite{Papadopoulos:2002bg} for a discussion about the UV finiteness of the world-sheet zero point energy in the presence of a non-trivial dilaton.}

%Let us finally observe that the sums (\ref{bosma}) and (\ref{ferma}) run over all the world-sheet fields, not just over the physical ones. The reason is that we have worked in conformal gauge without further fixing other bosonic or fermionic symmetries. For instance, the sum of the squared masses of the eight physical (transverse) world-sheet scalars is actually given by the right-hand-side of (\ref{bosma}) modulo a term proportional to $R_{h}$ 
%\cite{Gautason:2021vfc}. Then, the right-hand-sides of the expression  \eqref{bosma} turn out to be equal to the sum of the squared masses of the {\it physical} degrees of freedom if the conformal factor in the gauge fixing \eqref{gaugefixing} is a constant.

%\st{Let us stress that the number of physical world-sheet fermions \emph{must} be equal to eight, as many as the physical world-sheet bosons. A precise realization of this statement in the Hagedorn limit, where a standard semiclassical expansion is generically not available, is deferred to a future work}  \cite{w}. \st{That is crucial, since it guarantees the remaining anomaly to vanish. Indeed, we must not forget about the presence of bosonic ghosts coming from the gauge-fixing of the intrinsic metric. It is well known that their masses are proportional to the scalar curvature. Furthermore, other $R_{(h)}$ contributions to the trace of the stress-energy tensor are produced by the Jacobian of the transformation from GS fermions to two-dimensional fermions on the world-sheet.}

%{\blu I would substitute the last two paragraphs with:}

The above results must be supplemented with some considerations on the conformal and $\kappa$-symmetry gauge fixings, which we have so far omitted from our discussion. Schematically, our treatment of the quadratic fluctuations implies that we should still perform a path integration over the fields 
\be
	\int \mathcal{D}\xi\,\mathcal{D}b\, \mathcal{D}c\, \frac{\mathcal{D}\chi}{\mathcal{V}_\kappa}~. 
\ee
The ten bosonic fields $\xi$ contain eight (physical) transverse modes and two longitudinal modes, the contributions of which are exactly canceled by the integration over the $b$-$c$ ghost pair that arises in conformal gauge \cite{Drukker:2000ep, Forini:2015mca}. It is well-known that their masses are proportional to the world-sheet scalar curvature $R_{h}$. Therefore, for a flat background world-sheet, only the transverse fluctuations need to be counted and, as a result, the sum (\ref{bosma}) runs just over the latter. Similarly, the integration over the fermionic fields $\chi$ should be considered modulo $\kappa$-symmetry, which we indicated by ${\cal V}_{\kappa}$ in the denominator. Of course, $\kappa$-symmetry cannot be treated in the BRST formalism, but it is reasonable to expect that  it compensates the contribution of eight 2-d fermions with masses again proportional to the world-sheet scalar curvature. Other $R_{(h)}$ contributions to the trace of the stress-energy tensor are produced by the Jacobian of the transformation from GS fermions to the 2-d fermions on the world-sheet \cite{Gautason:2021vfc}. Therefore, the effective number of 2-d fermions is eight and the mass-matching condition \eqref{massmatch} is not violated.

Having in mind a flat world-sheet metric and a light-cone-gauge-like approach, these arguments would {\it effectively} allow us to perform a canonical quantization of just the physical degrees of freedom satisfying  \eqref{massmatch}. 

\section{Hagedorn configuration on the world-sheet}
\label{sec:extrapol}

Holography allows to map the low energy regime of classes of planar finite temperature confining gauge theories, into type II supergravity backgrounds with fluxes whose string frame metrics, written in a convenient coordinate frame, generically asymptote to
\be \label{asymetric}
ds^2 \approx 2\pi\alpha' T_s \(1+\frac{y^2}{l^2}\) \(dt^2 + \eta_{ij} dx^i dx^j\) + d\vec y\,^2 + ds^2_{\mathcal M} \, , \quad i, \, j =1,...,p, \quad y \approx 0 \, ,
\ee
where $\vec y = \{y_1,...,y_d\}$, $y^2 = y_1^2 + ... + y_d^2$ is dual to the gauge theory energy scale. Here, $t$ is the compact Euclidean time direction of length $\beta=1/T$ and a Wick rotation has been made on one of the other space directions $x^i$, $\alpha'$ is the string length squared, $T_s$ is the confining string tension of the dual gauge theory, $l$ is the curvature scale and $\mathcal M$ is a $(9-d-p)$-dimensional compact transverse space (which looks flat for the string configuration at $y=0$ we will focus on). Setting $p=0$, $T_s=1/2\pi\alpha'$ and $l=R_{AdS}$, the above geometries describe the IR regime of $d$-dimensional CFTs on $S^{(d-1)}$.\footnote{Notice that we are keeping track of the $y^2$-contributions to the metric coefficients which are relevant for the quadratic world-sheet theory.} We will refer to the latter as ``confining" in the sense described in \cite{Witten:1998zw}, as they display a first order confinement/deconfinement phase transition, with the critical temperature set by the $S^{(d-1)}$ radius . 

Working in Polyakov's formalism in conformal gauge, a classical string configuration $X^p(\tau,\sigma)$ probing the above backgrounds and winding around the compact time circle with zero momentum will in general need to have some momentum in the flat or transverse directions, in order to satisfy the classical Virasoro constraints. It is around such a configuration that a semiclassical approach can be employed, as in the previous section, solving for the related quadratic quantum fluctuations. This analysis will allow to get the spectrum, to second order in the world-sheet oscillators, and hence, in particular, the mass of the lowest lying winding string mode \cite{Bigazzi:2023oqm}. 

The Hagedorn regime is in this sense critical, as it corresponds to the limiting configuration with zero momenta 
%The consistency of the world-sheet sigma model results allowed to skip some details about the computations, like the explicit derivation of the masses of the fermionic fluctuations in the Hagedorn regime. 
%Indeed, their sum can be inferred from the conformal anomaly cancellation condition, once one compute the masses of the bosonic fluctuations at the Hagedorn point
\be
\label{hagpoint}
X^0=\pm \frac{\beta_H}{2\pi} \, \sigma \, , \quad X^i=0 \, , \quad \forall i \neq 0 \, ,
\ee
being ``$0$'' the (compact) thermal direction of the target space and $\beta_H$ the inverse Hagedorn temperature. In this regime the background induced metric is degenerate and, as outlined in \cite{Bigazzi:2022gal,Bigazzi:2023oqm},  the background configuration turns out to be of the same order as the fluctuations around it, invalidating the semiclassical approximation.

Adopting a pragmatic approach, we do not claim to expand the world-sheet action around \eqref{hagpoint}. Indeed, this procedure would be ill-defined due to the degeneracy of the background induced bosonic metric.
What we can do is to compute the sigma model masses in a near-Hagedorn regime where the world-sheet is still regular, regardless of what the background string configuration and spacetime metric are. Then, we take the $\beta\to\beta_H$ limit, paying attention to its regularity, extrapolating the sigma model results, in particular those regarding the masses of the world-sheet fields, to the configuration \eqref{hagpoint}. 

Clearly, this extrapolation is not guaranteed to work a priori. However, both the agreement of world-sheet results on the Hagedorn temperature to NLO \cite{Bigazzi:2023oqm} with those obtained using the complementary effective approach \cite{Urbach:2023npi}, and the already mentioned outstanding agreement (even at subleading orders) with numerical predictions coming from integrability and quantum spectral curve methods in global $AdS$ cases, allow us to be encouraged performing the above extrapolation in a very general framework. 

Most importantly for our purposes, the results collected in section \ref{sub:mass} do not show any degeneracy once extrapolated to the configuration \eqref{hagpoint}. The generalized mass-matching condition (\ref{massmatchplusdilaton}) generically reduces to (\ref{massmatch}) at the Hagedorn point, provided the dilaton is time-independent.

Once we verify \emph{a posteriori} that the Hagedorn limit can be taken safely, we end up with a two-dimensional consistent world-sheet theory of, in general, eight massive bosonic modes and eight massive fermionic modes whose masses depend on $T_H$. Crucially, the mass-matching condition at the Hagedorn point guarantees its UV finiteness, since it makes the zero point energy of the world-sheet sigma model
\be
\Delta = -\frac12 \sum_b \mu_b - \sum_b \sum_{n=1}^\infty \sqrt{n^2 + \mu_b^2} + \sum_f \sum_{n=1}^\infty \sqrt{\(n-\frac12\)^2 + \mu_f^2}
\ee
a finite quantity.\footnote{The presence of a time-dependent dilaton is accompanied by a two-dimensional renormalization which cancels the would-be divergence in the $\mathcal O(\mu^2_b)$ term of $\Delta$ \cite{Papadopoulos:2002bg}. If the dilaton is trivial,  \eqref{massmatch} does the job.} %For the details, see section \emph{UV finiteness and the role of the dilaton} in .

Now that we have a method to compute the sigma model masses at the Hagedorn point, we can generalize the proposal in \cite{Bigazzi:2023hxt} about the higher order corrections to the (inverse) Hagedorn temperature $\beta_H$ in the holographic limit. This is the main result of this work.  

Proceeding step by step, in \cite{Bigazzi:2023hxt} the authors provided a formula for $\beta_H$ in gauge theories whose low energy regime is holographically captured by the asymptotic backgrounds in \eqref{asymetric}. %, along with a (possibly running with $y$) dilaton, Ramond-Ramond field strengths flowing over the transverse space w.~r.~t.~the Minkowskian sector and a trivial (at least around $y=0$) Kalb-Ramond field.
More precisely, it has been given by the implicit relation 
\be \label{genproposal}
\frac{g_{00}(0)}{\alpha'} \, \left(\frac{\beta_H}{2\pi}\right)^2 = 2 \[ \, \Delta + \Delta {\cal E} \, \] \, ,\quad (g_{00}(0)=2\pi\alpha' T_s)\,,
\ee
where, up to second order in the world-sheet mass parameters (taking into account the mass-matching condition \eqref{massmatch}),
\begin{equation}\label{formul}
\Delta = 1 -\frac12 \sum_{b=1}^8\mu_b+ \log{2} \sum_{b=1}^8 \mu_b^2+ ... \,,
\end{equation}
and $\Delta{\cal E}$ accounts for NNLO and NNNLO contributions derived from the effective approach. Notice, crucially, that once the mass-matching condition is satisfied, we do not need any detail about the fermionic masses in \eqref{formul}.

The extrapolation to the Hagedorn point of the results we have collected above show that the formula (\ref{formul}) holds in general cases where the dual backgrounds have time-independent dilaton, generic RR forms and a $B$-field with no legs along the time direction. 

From formulae (\ref{bosma}) and (\ref{hagpoint}) in the Hagedorn regime we get
\be
\sum_{b=1}^8 \mu_b^2 = - \left(\frac{\beta_H}{2\pi}\right)^2 \left[ R_{00}(0)-\frac{1}{2} |H_{(3)}(0)|_{00}^2 \right]\,,
\ee
so that the Hagedorn temperature is obtained as the solution of the equation
\be
\label{complete}
-\sum_{b=1}^8\mu_b+2\,\Delta {\cal E}= \alpha' m^2
\ee
where
\be
\alpha' m^2\equiv -2 + \frac{g_{00}(0)}{\alpha'} \, \left(\frac{\beta_H}{2\pi}\right)^2 + 2\log{2}\left(\frac{\beta_H}{2\pi}\right)^2 \left[ R_{00}(0)-\frac{1}{2} |H_{(3)}(0)|_{00}^2 \right]\,.
\ee
In cases with $B$-field with legs along the time direction the relation above is easily generalized taking into account that the canonical momentum density gets shifted
\be
\pi^{\tau}_p = \frac{\partial{\cal L}}{\partial (\partial_{\tau} x^p)} = \frac{1}{2\pi\alpha'}\left(\partial_{\tau} x^q g_{pq} +\partial_{\sigma} x^q B_{pq}\right)\,,
\ee
where ${\cal L}$ is the Polyakov action coupled to the B-field. Thus, at the flat tip of our backgrounds ($y=0$),
\be
P_i = \int \hspace{-4pt} d\sigma \, \pi^{\tau}_i = p^j g_{ij}(0) + \left(\frac{\beta}{2\pi}\right) \frac{B_{0i}}{\alpha'}\,.
\ee
Taking this into account, and setting $p^i=0$, $\beta=\beta_H$ in the Hagedorn regime, the combination $m^2$ in the above equations will get shifted by a term quadratic in the B-field.\footnote{The gauge invariance of the corresponding quantum condition on the string ground state is guaranteed by the fact that now the ground state is charged with respect to the B field. In the effective approach this is realized by the fact that the thermal scalar is charged under an external $A_p=B_{0p}$ gauge field, so that the standard derivative is replaced by a covariant derivative (see \cite{Urbach:2023npi}).}

All in all, our final proposal for the (inverse) Hagedorn temperature in any holographic confining gauge theory whose dual generically features both RR and NSNS fluxes is the solution of \eqref{complete} equipped with
\be
\label{geneffmass}
\alpha' m^2\equiv -2 + V^\mu V^\nu \hspace{-2pt} \[ \frac{g_{\mu\nu} \hspace{-4pt}+g^{pq} B_{p\mu} B_{q\nu}}{\alpha'} + 2\log{2} \( R_{\mu\nu}-\frac{1}{2} |H_{(3)}|_{\mu\nu}^2 \)\] ,
\ee
being $V$ the tangent vector to the thermal direction of the background with modulo $\beta_H/2\pi$. 

Notice that the combination $m^2$ singled out by the world-sheet approach precisely corresponds to the expression for the effective mass (once evaluated at the origin of the background) of the thermal scalar as recently conjectured in \cite{Harmark:2024ioq} (see eq.~(2.28) there), automatically providing, universally, the value $C=2\log{2}$ in the notation of \cite{Harmark:2024ioq}.

%The following part of this work has been structured as follows: section \ref{sec:quadmodel} is devoted to a full-glory world-sheet derivation of the generalized mass-matching condition in \eqref{massmatchplusdilaton}, which applies to any II supergravity model featuring whatever kind of fields content the reader can imagine; instead, in section \ref{sec:effnlo}, we focus on backgrounds whose fields content admits an expansion in even power of the only holographic coordinate and we compute $\beta_H$ up to NNLO in the holographic limit within this class of models (if the background is an exact solution of supergravity in $\alpha'$, then we present also the final result at NNNLO); in the appendixes, we collect all the notations and the details about the computations of the main body.

%%%%%%%%%%%%%%%%%%%%%%%%%%%%%%%%%%%%%%%%%%%%%%%%%%
%%%%%%%%%%%%%%%%%%%%%%%%%%%%%%%%%%%%%%%%%%%%%%%%%%% 

%%%%%%%%%%%%%%%%%%%%%%%%%%%%%%%%%%%%%%%%%%%%%%%%%%%%%%

\section{Effective approach on general backgrounds}
\label{sec:effnlo}

The near Hagedorn behavior of the string model can be captured by an effective field theory for a complex scalar field $\chi$ describing the stringy winding mode. 
Reducing the theory on the compact directions, the Hagedorn singularity happens when the winding mode $\chi$ becomes massless. This is the picture recently adopted in \cite{maldanotes,Urbach:2022xzw,Urbach:2023npi,Ekhammar:2023glu,Bigazzi:2023hxt,Ekhammar:2023cuj,Harmark:2024ioq} for the computation of $T_H$ in particular examples, following the perspective in \cite{Horowitz:1997jc}.

The effective approach is very convenient and allows us to compute the higher order corrections in a very simple way. However, the final result is affected by arbitrary coefficients which cannot be fixed by the approach itself. On the other hand, the world-sheet method is very powerful: in principle, we could derive the full result without any ambiguity from a pure sigma-model computation. For example, in section \ref{pureworld-sheet} we show how the NNLO correction predicted by the effective approach can be interpreted as the VEV of the contributions to the world-sheet canonical Hamiltonian which are quartic in the zero modes part of the bosonic fluctuations, at least in a very specific case. Unfortunately, this method is also very demanding.

In this section, we generalize the effective approach for the computation of $T_H$ to a whole general class of backgrounds. For the sake of simplicity here we will not introduce the extra contributions predicted by the sigma model. In other words here we will neglect higher derivative corrections to the thermal scalar effective action. The missing terms will be accounted for in section \ref{sec:examples}, relying on the interplay of the two approaches. %: the missing coefficients in the effective one can be deduced from a comparison with the results coming from a proper world-sheet theory.

Consider a background with string-frame metric of the form
\be \label{metricnlo}
ds^2 = 2\pi \alpha' T_s \tilde{g}^{(0)}(r) \left(dt^2 + \eta_{ij} dx^i dx^j \right) +dr^2 + r^2 \tilde{g}^{(\Omega)}(r)d\Omega_{d-1}^2 + \tilde{g}^{(\mathcal M)}(r) ds_{\cal M}^2 \,,
\ee
supported by some RR-forms, the NSNS two-form $B_{(2)}$ and a running dilaton $\phi(r)$. Here, $i=1,...,p$, the radial variable is $r \in [0,\infty)$ and the form of the metric above is supposed to be valid around $r=0$, where the internal non-shrinking manifold ${\cal M}$ is $(9-p-d)-$dimensional.
This form of the metric turns out to be quite general.
Nevertheless, even if the structure of the metric, for example its ``internal part'' (i.e. the non-Minkowski-like part), is not precisely as in (\ref{metricnlo}), the final results of this section will be valid as well.
The only mandatory requirements are that there is a holographic direction $\rho$ ($r$ above) with some ``IR value'' $\rho_0$ ($r=0$ above)\footnote{For ``IR'' we mean the region of the holographic radial direction corresponding to the IR regime of the dual field theory.} such that $g_{tt}(\rho_0)\neq 0$ so as to have a finite value of the ``string tension'' $T_s$.
Then, in many of the known cases, the internal part of the metric in the IR has some vanishing cycle as in (\ref{metricnlo}).
Finally, clearly the end results will not be dependent on the coordinates chosen: they can be written in a covariant form, as we will exhibit in section \ref{sec:covariant}.

Starting from scratch, the $(p+1)$-dimensional string frame effective action which describes the infrared dynamics of a scalar field
\be
\chi = \chi(\vec x \, , r)
\ee
on the background in \eqref{metricnlo} is
\be\label{essechi}
S_{\chi} \approx \beta \, V_{S^{d-1}} \, V_{\mathcal M} \hspace{-4pt} \int \hspace{-4pt} d^p\vec x \, dr \, \sqrt{-g} \, e^{-2\,\phi} \{ g^{mn} \, \partial_m \chi^* \, \partial_n \chi + m^2_{eff}(r) \, \chi^* \chi\}\,,
\ee
where
\be
\alpha' m^2_{eff}(y) = \frac{\beta^2  (g_{tt}(r)+g^{mn} B_{m t} B_{n t})-8 \pi^2 \alpha'}{4\pi^2\alpha'}\, .
\ee
Here, $\beta$ is the length of the compactified temporal direction, $V_{S^{d-1}}=\int d^{d-1}\theta \sqrt{g_{S^{d-1}}}$ is the volume of the unit $S^{d-1}$, $V_{\mathcal M}=\int d^{9-p-d}\varphi \sqrt{g_{\mathcal M}}$ is the volume of the manifold $\mathcal M$ and $g$ is the determinant of the whole ten-dimensional metric \eqref{metricnlo}, once extracted $g_{S^{d-1}}$ and $g_{\mathcal M}$. This action can be easily derived by the reduction of the ten-dimensional one on the compact directions. This is the reason why $g$ contains also the $r$-dependent contributions coming from all directions.

The equation of motion for $\chi$ is thus
\be
- \frac{1}{\sqrt{-g} \, e^{-2\,\phi}} \partial_m \[\sqrt{-g} \, e^{-2 \, \phi} \, g^{mn} \, \partial_n \chi\] + m^2_{eff}(r) \, \chi = 0 \, .
\ee
With the ansatz
\be
\chi(\vec x , r) = e^{i \vec p \cdot \vec x} \, w(r) \, ,
\ee
and in the limit $M^2 = - \eta_{ij} p^i p^j\to0$, the above equation reduces to
\be\label{eqw}
-\frac12 w'' + \frac12 w' \left(2\phi'-\partial_r \log{\sqrt{-g}}  \right) + \frac12 m^2_{eff} w = 0\,.
\ee

%%%%%%%%%%%%%%%%%%%%%%%%%%%%%%%%%%%%%%%%%%%%
\subsection{Expansion}

Let us specialize to the case where the metric functions, $B_{(2)}$ and the dilaton admit an expansion in even powers of the radius in the variable $r$:\footnote{We are also considering the case where the constant in front of the $r^2 d\Omega_{d-1}^2$ term is equal to one. The end result is not going to depend on this choice. Indeed, the dependence on $g$ is enclosed in a logarithmic derivative and so all the $r$-independent factors do not contribute.} 
\bea
&& \tilde{g}^{(0)}(r)= 1+\tilde{g}^{(0)}_{2} r^2 + \tilde{g}^{(0)}_{4} r^4 + ...\,, \\ 
&& \tilde{g}^{(\Omega)}(r)= 1+\tilde{g}^{(\Omega)}_{2} r^2 + ... \,, \\
&& \tilde g^{(\cal M)}(r) = 1+\tilde{g}^{(\cal M)}_{2} r^2 + ...\,, \\
&& \phi (r)= \phi_0 + \phi_2 r^2+ ... \,,\\
&& g^{mn} B_{m t} B_{n t} = b_2 \, r^2 + b_4 \, r^4 + ...\,.
\eea
Note that, comparing with \eqref{asymetric}, we have 
\be 
\tilde g^{(0)}_2 = l^{-2}\,.
\ee
Introducing the notations
\be
C \equiv \(2\pi\alpha'T_s\)^{\frac{p+1}{2}} \, , \qquad \tilde{g}_2 \equiv \frac12 \left((p+1)\tilde{g}^{(0)}_{2} + (d-1)\tilde{g}^{(\Omega)}_{2} + (9-d-p)\tilde{g}^{(\cal M)}_{2} \right) \,,
\ee
around $r=0$ we can parametrize\footnote{If the metric is given in a different radial variable $\rho$ from the $r$ variable in (\ref{metricnlo}), with the relation $dr^2=F(\rho)^2d\rho^2$, by exploiting the fact that $\sqrt{|g(r)|}dr=\sqrt{|g(\rho)|}d\rho$ and that $g_{rr}(r)=1$, a convenient way to extract the coefficients $d, \tilde g_2, ...,$ is to simply expand $\sqrt{|g(\rho)|}/F(\rho)|_{\rho=f(r)}$, where $f(r)$ is found by integrating $dr=F(\rho)d\rho$.
If the metric in $\rho$ variables has no $d\rho dy^i$ components ($i=1, ..., 9$), this simplifies to $\sqrt{|g_9(\rho)|}|_{\rho=f(r)}$, where $g_9(\rho)$ is the determinant of the nine dimensional part of the metric (i.e. omitting the radius).}
\be\label{sqrtg}
\sqrt{-g}= C r^{d-1} \left(1+\tilde{g}_2 r^2 + \tilde g_4 r^4 + ... \right)\,.
\ee

Since at leading order the Hagedorn temperature is given by the relation $\beta_H^2 = 4\pi/T_s$, we can parameterize the corrections we want to calculate in the near-Hagedorn regime as
\be \label{nearHag}
{\boxed{\beta_H^2 \simeq \frac{4\pi}{T_s} \( 1 + c_0 \, \alpha'^{\frac12} + c_1 \, \alpha'\ + c_2 \, \alpha'^{\frac32} + ...\) 
}}\,.
\ee
Moreover, the fluctuation $r$ scales as
\be\label{erre}
r \sim \sqrt{\frac{\alpha'}{\mu_b}} \sim \sqrt{l_s l} \,,
\ee
where $l_s\equiv\sqrt{\alpha'}$. The above scaling relation follows from the world-sheet analysis in the previous sections and in \cite{Bigazzi:2023oqm} (with $r^2 = y^i y^i$, and $y^i$ being the zero modes of the massive scalar fields on the world-sheet). Briefly, the on-shell expression for $r$ is given by an expansion in zero and non-zero modes. The effective approach takes into account just the former ones. Their normalization depends on $\beta$ through the world-sheet bosonic mass $\mu_b$. In the near Hagedorn regime, this affects the scaling of $r$ as above. 

Therefore, we can consistently expand the equation of motion for $w(r)$ in powers of $\alpha'$.
This allows to perform a standard quantum-mechanical perturbative analysis.

Note however that the measure of integration in the action $S_{\chi}$ (\ref{essechi}), which enters the scalar products, has in turn the perturbative expansion in (\ref{sqrtg}). 
A convenient way to keep track of this expansion is to include the expanded part of $\sqrt{-g}$, i.e. $\sqrt{-g}/C r^{d-1}$, in the equation for the mode $w(r)$.
In this way, all the internal products, at every order, will be performed with the same measure $r^{d-1}$.

Thus, multiplying everything by $\sqrt{-g}/C r^{d-1}$ and considering the expansion of $\beta_H$ in (\ref{nearHag}) and the scaling of $r$ in (\ref{erre}), the equation of motion (\ref{eqw}) for $w(r)$ can be expanded as
\begin{align}
	0 = 
	&  \left[1+\tilde{g}_2 r^2 + \tilde g_4 r^4 + ... \right] \Bigl[ - \frac12 \, w''(r) - \frac12 \, \frac{d-1}{r} \, w'(r) + \frac{1}{\alpha'} \, \tilde G_2^{(0)} r^2 \, w(r) + \frac{c_0}{\sqrt{\alpha'}} \, w(r) + \nonumber \\
	& + \(2\, \phi_2 - \tilde g_2\) \, r \, w'(r) + c_1 \, w + \frac{c_0}{\sqrt{\alpha'}} \, \tilde G_2^{(0)}\, r^2 \, w(r) + \frac{1}{\alpha'} \, \tilde G_4^{(0)} r^4 \, w(r) +\nonumber \\
	& + \(4\, \phi_4 - 2\tilde g_4 + \tilde g_2^2\) \, r^3 \, w'(r) + c_2 \alpha'^{\frac12}\, w(r) + c_1 \, \tilde G_2^{(0)}\, r^2 \, w(r) + \frac{c_0}{\sqrt{\alpha'}} \, \tilde G_4^{(0)} r^4 \, w(r) + \nonumber \\ 
	& + \frac{1}{\alpha'} \, \tilde G_6^{(0)} r^6 \, w(r) + ... \Bigr]\, ,
\end{align}
being 
\be
\tilde G_a^{(0)} = \tilde g_a^{(0)} + \frac{b_a}{2\pi\alpha'T_s} \, , \quad a=2, 4, 6, ... \, .
\ee

To clarify the expansion, we can introduce the zeroth-order parameter
\be
\xi = \(\frac{2 \, \tilde G^{(0)}_2}{\alpha'}\)^{\hspace{-4pt}\frac14} r \, .
\ee
In this way, we can sharply separate the different orders in the $\alpha'$-expansion as
\be \label{eigenproblem}
H w(\xi) = E w(\xi) \, , \quad H=H_0 + H_1 + H_2 + ...\, , \quad E = E_0 + E_1 + E_2 + ...\, ,
\ee
where
\begin{align}
	&H_0 = -\frac12 \frac{\partial^2}{\partial \xi^2} - \frac12 \frac{d-1}{\xi} \frac{\partial}{\partial \xi} + \frac12 \xi^2  \, , \\
	&H_1 = \(\frac{2 \, \tilde G^{(0)}_2}{\alpha'}\)^{\hspace{-4pt}-\frac12} \tilde g_2 \xi^2 \left(H_0 - E_0 \right) + \Delta H_1 \equiv \nonumber \\
	& \quad \ \equiv \(\frac{2 \, \tilde G^{(0)}_2}{\alpha'}\)^{\hspace{-4pt}-\frac12} \tilde g_2 \xi^2 \left(H_0 - E_0 \right) + \nonumber \\
	& \quad \ + \(\frac{2 \, \tilde G^{(0)}_2}{\alpha'}\)^{\hspace{-4pt}-\frac12} \[ \(2\, \phi_2 - \tilde g_2\) \xi \, \frac{\partial}{\partial \xi} + \frac{c_0}{\sqrt{2}} \(\tilde G_2^{(0)}\)^{\hspace{-2pt}\frac12} \xi^2 + \frac12 \frac{\tilde G^{(0)}_4}{\tilde G^{(0)}_2}\, \xi^4 \] \, , \label{acca1}\\
	&H_2 = \(\frac{2 \, \tilde G^{(0)}_2}{\alpha'}\)^{\hspace{-4pt}-1} \tilde g_4 \xi^4 \left(H_0 - E_0 \right) + \(\frac{2 \, \tilde G^{(0)}_2}{\alpha'}\)^{\hspace{-4pt}-\frac12} \tilde g_2 \xi^2 \left(H_1 - E_1 \right)  + \Delta H_2 \equiv \nonumber \\
	& \quad \ \equiv \(\frac{2 \, \tilde G^{(0)}_2}{\alpha'}\)^{\hspace{-4pt}-1} \tilde g_4 \xi^4 \left(H_0 - E_0 \right) + \(\frac{2 \, \tilde G^{(0)}_2}{\alpha'}\)^{\hspace{-4pt}-\frac12} \tilde g_2 \xi^2 \left(H_1 - E_1 \right)  + \nonumber \\
	& \quad \ +  \(\frac{2 \, \tilde G^{(0)}_2}{\alpha'}\)^{\hspace{-4pt}-1} \[ \(4\, \phi_4 - 2\tilde g_4 + \tilde g_2^2\) \, \xi^3 \, \frac{\partial}{\partial \xi} + c_1 \, \tilde G_2^{(0)}\, \xi^2  + \frac{c_0}{\sqrt{2 \tilde G_2^{(0)}}} \, \tilde G_4^{(0)} \xi^4  + \frac{1}{2} \, \frac{\tilde G_6^{(0)}}{\tilde G_2^{(0)}} \xi^6 \] \, ,
\end{align}
and  
\be\label{energies}
E_0 = - \(2 \tilde G_2^{(0)}\)^{-\frac12} c_0 \, , \qquad E_1 = - \(2 \, \tilde G^{(0)}_2 \)^{-\frac12} \alpha'^{\frac12} c_1 \,, \qquad E_2 = - \(2 \, \tilde G^{(0)}_2\)^{-\frac12} \alpha' c_2\,.
\ee

\subsubsection*{Order zero}

If we want the ``unperturbed'' problem $H_0 w = E_0 w$ to admit a normalizable solution, namely\footnote{Notice that the normalization constant
	\be
	\alpha_0 = \left(\frac{C}{2} \Gamma\[d/2\]  \right)^{\hspace{-3pt}-\frac12} \,,
	\ee
	is chosen such that $\int_0^{\infty} d\xi C \xi^{d-1} w_0^2 = 1$.}
\be
w_0 = \alpha_0 \, e^{-\frac{\xi^2}{2}}\,,
\ee
we have to impose
\be
E_0 = \frac{d}{2} \, .
\ee
In other words, as already noted in \cite{maldanotes, Urbach:2023npi, Ekhammar:2023glu,Bigazzi:2023oqm}, the leading term of \eqref{eigenproblem} is the eigenvalue equation for the ground state of a $d$-dimensional harmonic oscillator in flat space with unit frequency. 
Using (\ref{energies}), this fixes the NLO correction to the inverse Hagedorn temperature (\ref{nearHag}) to
\be\label{c0}
\boxed{
	c_0 = - \frac{d}{\sqrt{2}} \(\tilde  G_2^{(0)} \)^{\hspace{-2pt}\frac12}
}\,.
\ee

\subsubsection*{First order}

At the next level, the correction $E_1$ in (\ref{eigenproblem}) is given by the expectation value of the first-order perturbation $H_1$ on the ground state $w_0$.
Since $w_0$ is the eigenfunction of $H_0$ with eigenvalue $E_0$, only the $\Delta H_1$ part of $H_1$ (\ref{acca1}) gives a non-trivial contribution, so\footnote{Let us stress that $\Delta H_1$ is not a self-adjoint operator at this order. Instead, the complete perturbations $H_1$ and $H_2$ are self-adjoint as $H_0$.}
\be
\label{E1}
E_1 =   \langle w_0 | \Delta H_1 | w_0 \rangle  \, .
\ee
Thus, using again (\ref{energies}), we can fix the NNLO coefficient in the expansion of $\beta_H$ (\ref{nearHag}) as
\be
c_1 = - \left \langle w_0 \left | \(2\, \phi_2 - \tilde g_2\) \xi \, \frac{\partial}{\partial \xi} - \frac{d}{2} \tilde G^{(0)}_2 \xi^2 + \frac12 \frac{\tilde G^{(0)}_4}{\tilde G^{(0)}_2}\, \xi^4 \right | w_0 \right \rangle  \,,
\ee
that is
\be
c_1=C \alpha_0^2 \int_0^{\infty} \hspace{-8pt} d\xi \, \xi^{d-1}\[ \(2\phi_2 - \tilde g_2 \) \xi^2 +  \frac{d}{2} \, \tilde G^{(0)}_2 \, \xi^2 - \frac12 \frac{\tilde G^{(0)}_4}{\tilde G^{(0)}_2}\, \xi^4 \,\]  e^{-\xi^2} \, .
\ee
The result of the integration is
\be\label{c1}
\boxed{
	c_1=- \frac{d(d+2)}{8} \, \frac{\tilde G^{(0)}_4}{ \tilde G^{(0)}_2} - \frac{d}{2} (\tilde g_2 - 2 \phi_2) + \frac{d^2}{4} \tilde G^{(0)}_2 
} \,.
\ee

\subsubsection*{Second order}

One can then proceed to calculate the second correction $E_2$ in (\ref{eigenproblem}) by the formula
\be\label{e2}
E_2 = \langle w_0 | H_2 | w_0 \rangle + \sum_{n \neq 0} \frac{|\langle w_n | H_1 | w_0 \rangle |^2}{E_0-E_{(n)}}\,,
\ee
where $w_n$ are the eigenfunctions of the harmonic oscillator Hamiltonian $H_0$,
\be
w_n = \alpha_n L^{d/2-1}_{n/2} (\xi)\ e^{-\frac{\xi^2}{2}  }\,,
\ee
where $L^{d/2-1}_{n/2} (\xi)$ are the associated Laguerre polynomials and $\alpha_n$ the normalization constants, and their energy levels are
\be
E_{(n)} = \left(n + \frac{d}{2} \right)\,.
\ee
As above, the $(H_0-E_0)$ parts of $H_2$ and $H_1$ do not contribute to (\ref{e2}) and we are left with\footnote{Generically the unperturbed ground state $w_0$ is not an eigenfunction of the perturbation $H_1$.}
\be\label{e2b}
E_2 = \left \langle w_0 \left | \(\frac{2 \, \tilde G^{(0)}_2}{\alpha'}\)^{\hspace{-4pt}-\frac12} \tilde g_2 \xi^2 \left(\Delta H_1 - E_1 \right)  + \Delta H_2  \right | w_0 \right \rangle + \sum_{n \neq 0} \frac{|\langle w_n | \Delta H_1 | w_0 \rangle |^2}{E_0-E_{(n)}}\,.
\ee
Since $\Delta H_1$ is quartic and with even powers of $\xi$, only $w_2$ and $w_4$ contribute to the result.
Then equation (\ref{energies}) gives the NNNLO correction $c_2$ to the inverse Hagedorn temperature in (\ref{nearHag}).
Performing the calculation, one gets the rather unpleasant expression
\begin{align}
	\boxed{\begin{aligned}[t]
	c_2  & =   - \frac{d(d+2)(d+4)}{16\sqrt{2}} \, \frac{ \tilde G^{(0)}_6}{\(\tilde G^{(0)}_2\)^{\hspace{-2pt}\frac32}} 
	+ \frac{d(10+9d+2d^2)}{32\sqrt{2}} \, \frac{ \(\tilde G^{(0)}_4 \)^{\hspace{-2pt}2}}{\(\tilde G^{(0)}_2\)^{\hspace{-2pt}\frac52}} - \frac{d(d+2)}{2\sqrt{2}} \, \frac{ \tilde G^{(0)}_4\,\phi_2}{\(\tilde G^{(0)}_2\)^{\hspace{-2pt}\frac32}} + \\
	& \quad + \frac{d^2(d+2)}{16\sqrt{2}} \, \frac{ \tilde G^{(0)}_4}{\(\tilde G^{(0)}_2\)^{\hspace{-2pt}\frac12}} + \frac{d}{\sqrt{2}} \, \frac{\phi_2^2}{\(\tilde G^{(0)}_2\)^{\hspace{-2pt}\frac12}} + \frac{d(d+1)}{4\sqrt{2}} \, \frac{ \tilde g^2_2}{\(\tilde G^{(0)}_2\)^{\hspace{-2pt}\frac12}} + \frac{d^2}{4\sqrt{2}} \,  \tilde g_2\,\(\tilde G^{(0)}_2\)^{\hspace{-2pt}\frac12} + \\
	& \quad + \frac{d(d+2)}{2\sqrt{2}} \, \frac{ (\phi_4 -\tilde g_4)}{\(\tilde G^{(0)}_2\)^{\hspace{-2pt}\frac12}} - \frac{d^3}{16\sqrt{2}} \(\tilde G^{(0)}_2\)^{\hspace{-2pt}\frac32}\,.	 \label{c2}
\end{aligned}}
\end{align}

%%%%%%%%%%%%%%%%%%%%%%%%%%%%%%%%%%%%%%%%% 
\subsection{A more covariant form of the coefficients} 
\label{sec:covariant}

The coefficients $\tilde G^{(0)}_n, \tilde g_n$, etc., and so the $c_i$ coefficients in $\beta_H^2$ (\ref{nearHag}), have a covariant expression, so that one does not need to put the metric in the explicit form (\ref{metricnlo}).
As suggested in \cite{Harmark:2024ioq}, one can consider the vector associated to the compactified thermal circle, let us call it $\hat V$, normalized such that $\hat V^2=1$ at the IR point of the metric. Thus one can define the scalar quantity
\be 
G \equiv \hat V^{\mu} \hat V^{\nu}\left(g_{\mu\nu} + B_{\mu p}B_{\nu q}g^{p q} \right)\,.
\ee  
Let us also consider the normalized vector $W$ with non-vanishing component only along the direction $\rho$ on which the metric functions depend explicitly %\footnote{This could be extended to the dependence on various coordinates, but it will not be required in the cases considered in this paper.}
\be 
W= W^{\rho} \partial_{\rho} = \frac{1}{\sqrt{g_{\rho\rho}}} \partial_{\rho} \,.
\ee
Then, the $\tilde G^{(0)}_n$ coefficients are simply given by
\be 
\tilde G^{(0)}_n = \frac{1}{n!} W^{\mu_1}...W^{\mu_n} \nabla_{\mu_1}...\nabla_{\mu_n} G\, |_{IR}\,,
\ee
where $|_{IR}$ means that after the derivatives and contractions are performed, the result must be evaluated at the IR point of the metric,  say at the bottom $\rho_0$ of the space.\footnote{Note that even in the $B_{(2)}=0, V^{\mu}=\delta^{\mu }_t$ case, $\nabla_{\mu_n} G$ is the derivative of the function $G$, not the (covariant) derivative of the metric tensor component $g_{tt}$, which is automatically zero since we are using the standard Levi-Civita connection.}

For what concerns the coefficients $\tilde g_n$, one first considers the quantities
\be 
\hat g_m \equiv \frac{1}{m!} W^{\mu_1}...W^{\mu_m} \nabla_{\mu_1}...\nabla_{\mu_m} \left(g^{\rho\rho}\sqrt{|g|} \right) \, |_{IR}\,.
\ee	 
Call $d-1$ the number of such derivatives which are zero, beginning with $n=0$, and $C$ the value of the (non-zero) $d-$th derivative, $C=\hat g_d $.
Then calling $n=m-d$, for $n\geq 0$ one defines
\be 
\tilde g_{n} =  \frac{1}{C\,(n+d)!} W^{\mu_1}...W^{\mu_{n+d}} \nabla_{\mu_1}...\nabla_{\mu_{n+d}} \left(g^{\rho\rho}\sqrt{|g|} \right)\, |_{IR}\,.
\ee	
Finally, one can write\footnote{Obviously one could rewrite the results in terms of the covariant derivatives of the function $e^{-2\phi}g^{\rho\rho}\sqrt{g}$.}
\be 
\phi_n = \frac{1}{n!} W^{\mu_1}...W^{\mu_n} \nabla_{\mu_1}...\nabla_{\mu_n} \phi\, |_{IR}\,.
\ee

We are going to consider an explicit example in section \ref{sec:wymcov}, but we have explicitly checked that the results obtained with this covariant formalism are the same as the one obtained from the metric form (\ref{metricnlo}) in all the cases considered in this paper.

\subsubsection{A trivial extension: pp-waves}
\label{sec:ppwaves}

This formalism is readily extended to the RR-supported pp-wave cases, as already stressed in \cite{Harmark:2024ioq}.
The Lorentzian metric reads
\be 
ds^2 = - 2dx^+ dx^- - f^2 \left(a_1^2 \sum_{i=1}^4 x_i^2 + a_2^2 \sum_{i=5}^8 x_i^2 \right)(dx^+)^2 +\sum_{i=1}^8 dx_i^2\,.
\ee
The dilaton is constant.
In this case one finds that the problem is equivalent to eight decoupled harmonic oscillators in one dimension, with no further perturbations.
Four of these oscillators have frequencies fixed by $f a_{1}$, the other four by $f a_{2}$, where $a_1, a_2$ are just numbers and $f$ sets the scale.  
For each of these oscillators one has
\be 
\hat V = \partial_t \,, \qquad W_i = \partial_{x_i} \,,
\ee 
and
\be 
G = 1+ \frac12 f^2 a_i^2 x_i^2\,, \qquad g^{x_ix_i}\sqrt{g}=1\,.
\ee 
Thus one finds that $d=1$ and the only non-trivial coefficient is
\be 
\tilde G^{(0)}_2 = \frac{1}{2} \nabla_{x_i}\nabla_{x_i} G \, |_{x_i=0} = \frac12 f^2 a_i^2\,,
\ee 
which contributes to the coefficient $c_1$ in (\ref{nearHag}) as
\be 
-\frac{d}{\sqrt{2}} \sqrt{\tilde G^{(0)}_2} = -\frac{1}{2} f a_i\,.
\ee
Summing up the eight contributions one gets
\be 
c_0= - 2f (a_1+a_2) \,,
\ee 
and analogously
\be 
c_1= \frac12 c_0^2 = 2 f^2 (a_1+a_2)^2 \,,
\ee 
which coincide with the first terms in the expansion of the known result about the Hagedorn temperature of strings on pp-waves \cite{Harmark:2024ioq}. 
%%%%%%%%%%%%%%%%%%%%%%%%%%%%%%%%%%%%%%%%%%%%%%%
\section{Examples}
\label{sec:examples}
Generically, supergravity backgrounds are not exact solutions in $\alpha'$.
The corrections to the geometry would affect the value of the Hagedorn temperature at NNNL order \cite{Bigazzi:2023hxt}.
For this reason, we are going to limit the examples below to the NNL order, which is safe.
The exception is given by $AdS$ backgrounds, which are exact, so in those cases we present the results at NNNL order.
We are going to employ formula (\ref{nearHag}) for $\beta_H^2$ supplemented by the term which can be obtained only from the sigma-model analysis, i.~e.~the $\log{2}$ term in (\ref{geneffmass}).
Considering that for a metric in the form (\ref{metricnlo}) one has $R_{00}(0) = - d \tilde g_2^{(0)}$, this shifts
\be \label{shift1}
c_1 \rightarrow \tilde c_1 \equiv c_1 +  2 \left(d  \, \tilde{g}^{(0)}_{2}  +\frac{1}{2} |H_{(3)}(0)|_{00}^2 \right) \log{2}\,.
\ee
%In all the examples below we will have a vanishing $|H_{(3)}(0)|_{00}^2$.
For $AdS$ backgrounds we can also write down the effect at NNNL order, which consists in the shift
\be \label{shift2}
c_2 \rightarrow \tilde c_2 \equiv c_2 - \frac{3}{\sqrt{2}}  \, d \, \(\tilde{G}^{(0)}_{2}\)^{\hspace{-2pt}\frac12} \left(d\tilde{g}^{(0)}_{2} +\frac{1}{2} |H_{(3)}(0)|_{00}^2   \right) \log{2}\,.
\ee
From a world-sheet perspective, $\beta_H^2$ can be found as the solution of the mass-shell condition for the winding ground state at the Hagedorn point, which is an implicit equation. This is the reason why the shift of the NNNL order coefficient $c_2$ clearly originates from lower order contributions (see the shift of $c_1$).

In fact, using (\ref{shift1}), (\ref{shift2}) and eq. (\ref{nearHag}), the Hagedorn temperature to NNNL order turns out to be given by
\be
\label{thseries}
T_H = \sqrt{\frac{T_s}{4\pi}}\left[1- \frac{c_0}{2} \alpha'^{1/2} + \frac{\alpha'}{8}(3c_0^2-4 c_1 +8 \hat R_0 \log 2) +\frac{\alpha'^{3/2}}{16}(-5 c_0^3 + 12 c_0 c_1 - 8 c_2)+\cdots \right]\,,
\ee
where
\be
\hat R_0 \equiv - \left(d  \, \tilde{g}^{(0)}_{2}  +\frac{1}{2} |H_{(3)}(0)|_{00}^2 \right)\,.
\ee
Notice that, as a result of the above shifts, no $\log 2$ term enters in the NNNL coefficient. The latter, in principle, could also contain another piece proportional to what in \cite{Harmark:2024ioq} is denoted as $C-\tilde C$, where $C=2\log 2$ is the coefficient multiplying the curvature corrections to the scalar mass, as we have computed in (\ref{geneffmass}), while $\tilde C$ weights the curvature correction to the derivative term in the scalar equation. If, as argued in \cite{Harmark:2024ioq}, all these corrections can be accounted for by a field redefinition of the metric, then $C-\tilde C=0$ and our relations (\ref{shift2}) and (\ref{thseries}) hold. The already mentioned quantum spectral curve computations of the Hagedorn temperature in both $AdS_5$ and $AdS_4$ setups precisely agree with this expectation. Here we thus assume that $C-\tilde C=0$ continues to hold for other $AdS$ cases. 

\subsection{Global $AdS$}

As a first example, let us consider the global $AdS$ cases with no $B_{(2)}$ field. 
For simplicity we fix the $AdS$ radius to unity, so that all the quantities are dimensionless; in order to restore dimensions, one has to multiply all the quantities by the appropriate power of the $AdS$ radius.
The metric reads
\be\label{metricads}
ds^2 = (\cosh{r})^2 dt^2 + dr^2 + (\sinh{r})^2 d\Omega_{d-1} + ... \,.
\ee
The internal manifold ${\cal M}$ has an $r-$independent metric (e.g.~the rigid five-sphere for $d=4$) and the dilaton is constant.
Simply expanding the hyperbolic sine and cosine, one finds for these cases that 
\bea
&& p=0\,,\quad \tilde{g}^{(0)}_{2}=1\,,\quad  \tilde{g}^{(0)}_{4}=\frac13\,,\quad   \tilde{g}^{(0)}_{6}=\frac{2}{45}\,,\quad  \tilde{g}_{2}=\frac{2 + d}{6}\,,\nonumber \\
&& \tilde{g}_4 =  \frac{-8 + 18 d + 5 d^2}{360}\,,\quad  \phi_2 = \phi_4 = b_2 = b_4 =b_6=0\,, 
\eea
such that using (\ref{c0}), (\ref{c1}), (\ref{c2}), (\ref{shift1}) and (\ref{shift2}) in (\ref{nearHag}) one gets
\be \label{finalAdS}
\beta_H^2 = \frac{4\pi}{T_s} \( 1 - \frac{d}{\sqrt{2}} \, \alpha'^{\frac12} + \frac{d(d-2)+16d\log{2}}{2^3} \, \alpha'\ -  \frac{96\, d^2 \log{2} -d(5d+2)}{2^{\frac{11}{2}}}\alpha'^{\frac32} \,  + ...\) \, ,
\ee
which is the known result \cite{maldanotes,Urbach:2022xzw,Urbach:2023npi,Ekhammar:2023glu,Bigazzi:2023hxt,Ekhammar:2023cuj}.

In the $AdS_3$ case one can add the $B_{(2)}$ form as (in Euclidean time) $B_{\varphi 0}=i \lambda (\sinh{r})^2$, where\footnote{The case $\lambda=1$ corresponds to the pure NSNS background and is special. The construction presented in section \ref{sec:effnlo} breaks down, since the zero-order problem is not a harmonic oscillator anymore, due to the cancellation of the gravitational term with the $B_{(2)}$ one. A different equation has to be considered in this case \cite{Urbach:2023npi}.} $0 \leq \lambda < 1$ and $\varphi$ parameterizes the circle $\Omega_1$ (see e.g. \cite{Urbach:2023npi}). 
Thus in this case
\be 
b_2=-\lambda^2\,,\quad b_4=-\frac{\lambda^2}{3}\,,\quad b_6=-\frac{2\lambda^2}{45}\,,
\ee
and we get
\be 
\beta_H^2 = \frac{4\pi}{T_s} \left[ 1 - \sqrt{2(1-\lambda^2)}\,\alpha'^{\frac12} +\left(4(1-\lambda^2)\log{2} - \lambda^2\right)\alpha' -  \frac{2\lambda^4+48(1-\lambda^2)^2\log{2}-3}{2^{\frac52}\sqrt{1-\lambda^2}}\, \alpha'^{\frac32} + ...\right] .
\ee 
The leading and NL terms are known \cite{Urbach:2023npi}, while the NNL and NNNL terms are new. 

\subsubsection{$AdS$ world-sheet to quartic order in the zero-modes}
\label{pureworld-sheet}

%In a very recent and nice paper \cite{Harmark:2024ioq}, the author inferred the equation of motion for a thermal scalar coupled to a II supergravity background supported by RR and NSNS fluxes. The curvature corrections have been guessed relying on widespread expectations in the literature, along with the covariance and the consistency of the final result. Nevertheless, the latter displays some arbitrary coefficients.

%Of course, such coefficients affect the computation for the Hagedorn temperature in the cases of interest. In particular, for II string theory on AdS and pp-wave backgrounds, the NNNLO prediction features two arbitrary parameters. Crucially, they have been fixed from the interplay with the stringy sigma-model approach and the comparison with the numerical outcome coming from Quantum Spectral Curve techniques.

The coefficients in front of the $\log 2$-terms in the prediction \eqref{finalAdS} have been fixed from the interplay with the stringy sigma-model approach and the comparison with the numerical outcome coming from Quantum Spectral Curve techniques. Without this support, such coefficients would remain arbitrary parameters if we were dealing with the effective approach alone.

Here, we propose a pure world-sheet method showing how to get the whole result from the sigma model alone, at least up to NNLO. 
In a nutshell, our goal is to compute the first order correction $E_1$ to the ground state energy of the harmonic oscillator depicted in section \ref{sec:effnlo} by means of the VEV of $\Delta H_{(4)}$, i.e.~the contribution to the canonical string Hamiltonian $H$ that is quartic in the zero-mode fluctuations of the sigma model approach.

Given the mass-shell condition
\be
\frac{\beta_H}{2\pi\alpha'} \frac{d}{2} + \Delta E^{(1)} = \frac12 \[ \frac{2}{\alpha'} - \(\frac{\beta_H}{2\pi\alpha'}\)^{\hspace{-2pt}2} + \frac{\beta_H^2}{2\pi^2\alpha'} \, d \, \log2\] \, ,
\ee
the link between
\be
\Delta E^{(1)} = \langle 0 | \Delta H_{(4)} | 0 \rangle \, ,
\ee
and $E_1$ defined in \eqref{E1} is
\be
\label{link}
E_1 = \(\frac{\alpha'}{2}\)^{\hspace{-2pt}\frac{1}{2}} \[\Delta E^{(1)} + \frac{d^2}{4}\] \, .
\ee

%For the sake of simplicity, let us focus on global-$AdS_{d+1}$ cases
Let us parameterize global-$AdS_{d+1}$ with a set of Cartesian-like coordinates as
\be
\rmd s^2 = \frac{(1+\frac14 y^2)^2}{(1-\frac14 y^2)^2} \, \rmd t^2 + \frac{\rmd y^k \rmd y^k}{(1-\frac14 y^2)^2} \, , \quad k=1,...,d \, .
\ee
%featuring a non-running dilaton and a trivial Kalb-Ramond field. 
Since we want to deal with something which is quartic in the fluctuations, we need to the expand the string configuration up to the same order, that is
\be
x^p=X^p+\xi^p_{(1)}+\xi^p_{(2)}+\xi^p_{(3)}+\xi^p_{(4)} \, , \quad p=0,1,...,d \, .
\ee
The subscripts denote the order of the fluctuations.

Here, we are focusing just on the zero-modes part of the above expansion. Then, each fluctuation can be expressed in terms of the lower order ones solving order by order the bosonic equations of motion
\be
-\partial_\alpha \( \eta^{\alpha\beta} \, \partial_\beta x^i \, g_{ki}(x)\) + \frac12 \partial_\alpha x^i \partial_\beta x^j \eta^{\alpha\beta} \partial_k g_{pq}(x) = 0 \, , \quad k=1,...,d \, .
\ee
In the Hagedorn regime $X^0= \rho_H \sigma$, $Y^k=0$, the equations of motion for the zero-mode part of each fluctuation are\footnote{We omit the equation for $\xi^p_{(4)}$ since, as we will see, it does not contribute to the final result in the Hagedorn regime.}
\begin{subequations}
\begin{align}
&\partial_\tau^2 \xi^k_{(1)} + \rho_H^2 \, \xi^k_{(1)} = 0 \, , \\
&\partial_\tau^2 \xi^k_{(2)} + \rho_H^2 \, \xi^k_{(2)} = 0 \, , \\
&\partial_\tau^2 \xi^k_{(3)} + \rho_H^2 \, \xi^k_{(3)} = - \frac12 \, \rho_H^2 \, \xi^i_{(1)} \, \xi^k_{(1)} \, \xi^k_{(1)} - \partial_{\tau}\xi^i_{(1)} \, \xi^k_{(1)} \, \partial_{\tau} \xi^k_{(1)} + \frac12 \, \xi^i_{(1)} \, \partial_{\tau}\xi^k_{(1)} \, \partial_{\tau}\xi^k_{(1)} \, ,
\end{align}
\end{subequations}
which are solved by
\begin{subequations}
\label{classsol}
\begin{align}
&&\xi^k_{(1)} &= i \, \sqrt{\frac{\alpha'}{2\rho_H}} \( a^k \, e^{-i \, \rho_H \, \tau} - a^{i*} \, e^{+i \, \rho_H \, \tau}\) \, , \label{prima}\\
&&\xi^k_{(2)} &=0 \, , \\
&&\xi^k_{(3)} &= -\frac18 \, \xi^i_{(1)} \, \xi^k_{(1)} \, \xi^k_{(1)} + \frac\tau2 \, \xi^i_{(1)} \,  \xi^k_{(1)} \, \partial_{\tau}\xi^k_{(1)} - \frac\tau4 \, \partial_{\tau}\xi^i_{(1)} \, \xi^k_{(1)} \, \xi^k_{(1)} +\\
&& \, &\quad+ \frac{1}{8\,\rho_H^2} \, \xi^i_{(1)} \, \partial_{\tau}\xi^k_{(1)} \, \partial_{\tau}\xi^k_{(1)} - \frac{1}{4\,\rho^2_H} \, \partial_{\tau}\xi^i_{(1)} \, \xi^k_{(1)} \, \partial_{\tau} \xi^k_{(1)} + \frac{\tau}{4\,\rho_H^2} \, \partial_{\tau}\xi^i_{(1)} \, \partial_{\tau}\xi^k_{(1)} \, \partial_{\tau} \xi^k_{(1)} \, .
\end{align}
\end{subequations}

Let us stress that these are classical solutions which are not promoted to quantum operators yet. Let us define the zero-mode contribution to the classical canonical string Hamiltonian as\footnote{The normalization has been chosen so that the comparison with the results from the effective approach is as clear as possible.}
\be
H^{\text{class}} = \frac{1}{4\pi\alpha'^2} \int_0^{2\pi} d\sigma \, \left . \delta^{\alpha\beta} \partial_\alpha x^p \, \partial_\beta x^q \, g_{pq}(x) \right|_{\text{zero-modes}} \, .
\ee
Up to quartic order in the zero-modes fluctuations, it reads
\be
H^{\text{class}} = H^{\text{class}}_{0} + \Delta H^{\text{class}}_{(4)} \, ,
\ee
where
\begin{subequations}
\begin{align}
& H^{\text{class}}_{0} = \frac{1}{2\alpha'^2} \(\partial_\tau \xi^k_{(1)} \, \partial_\tau \xi^k_{(1)} + \rho_H^2 \, \xi^k_{(1)} \, \xi^k_{(1)}\) \, , \\
& \Delta H^{\text{class}}_{(4)} = \frac{1}{2\alpha'^2} \( 2 \, \partial_\tau \xi^k_{(1)} \, \partial_\tau \xi^k_{(3)} + 2 \, \rho_H^2 \, \xi^k_{(1)} \, \xi^k_{(3)} + \frac12 \, \partial_\tau \xi^k_{(1)}  \, \partial_\tau \xi^k_{(1)} \, \xi^i_{(1)}\, \xi^i_{(1)} + \frac12 \rho^2_H \, \xi^k_{(1)} \, \xi^k_{(1)}\, \xi^i_{(1)}\, \xi^i_{(1)}  \) \, .
\end{align}
\end{subequations}

Notice that, on the classical solutions \eqref{classsol}, the above quantities are manifestly conserved. Therefore, without loss of generality, they can be evaluated at $\tau=0$.
Introducing
\be
x^i = \sqrt{\frac{\rho_H}{\alpha'}} \, \left . \xi^i_{(1)} \right |_{\tau=0} \, , \quad p^i = \frac{1}{\sqrt{\alpha' \rho_H}} \, \left . \partial_\tau \xi^i_{(1)} \right |_{\tau=0} \, ,
\ee
the classical perturbation Hamiltonian takes the form
\be
\Delta H^{\text{class}}_{(4)} = \frac18 \( x^k x^k \, x^i x^i + 2 \, x^k x^k \, p^i p^i + p^k p^k \, p^i p^i   \) \, .
\ee

Switching to the quantum world, we have to promote the coefficients $a^k$, $a^{k*}$ in (\ref{prima}) to the ladder operators $\hat a^k$, $\hat a^{k\dagger}$ such that
\be
\label{commrules}
\[\hat a^k, \hat a^{k\dagger}\] = + \delta^{ki} \, , \quad \hat a^k |0\rangle = 0 \, ,  \quad \[\hat x^k, \hat p^i\] = +i \, \delta^{ki} \, .
\ee

Due to ordering problems, we do not know \emph{a priori} what the quantum version of $\Delta H^{\text{class}}_{(4)}$ is. Agnostically, we should promote the function $x^k x^k \, p^i p^i$ to a generic linear combination of all the hermitian operators with the correct classical limit, that is
\be
x^k x^k \, p^i p^i \mapsto h_1 \, \frac{\hat x^k \hat x^k \, \hat p^i \hat p^i + \hat p^i \hat p^i \,\hat x^k \hat x^k}{2} + h_2 \, \frac{\hat x^k \hat p^i \, \hat x^k \hat p^i + \hat p^i \hat \, x^k \hat p^i \hat x^k  }{2} + h_3 \, \hat x^k \, \hat p^i \hat p^i  \, \hat x^k + h_4 \, \hat p^i \, \hat x^k \hat x^k \, p^i \, ,
\ee
where
\be
h_1 + h_2 + h_3 + h_4 = 1 \, .
\ee

Anyway, thanks to the commutation rules in \eqref{commrules}, we have
\be
2 \(\hat x^k \hat p^i \, \hat x^k \hat p^i +  \hat p^i \hat \, x^k \hat p^i \hat x^k \) = \(\hat x^k \hat x^k \, \hat p^i \hat p^i + \hat p^i \hat p^i \,\hat x^k \hat x^k\) + 2 \, \hat p^i \, \hat x^k \hat x^k \, p^i \, , \quad \hat x^k \, \hat p^i \hat p^i  \, \hat x^k = \hat p^i \, \hat x^k \hat x^k \, p^i \, .
\ee
Therefore, the most general hermitian quantum operator which reduces to $\Delta H^{\text{class}}_{(4)}$ in the classical limit is, for instance,
\be
\Delta H_{(4)} = \frac18 \[ \hat x^k \hat x^k \, \hat x^i \hat x^i + 2 \( h \, \hat x^k \, \hat p^i \hat p^i x^k + (1-h)\frac{\hat x^k \hat x^k \, \hat p^i \hat p^i + \hat p^i \hat p^i \,\hat x^k \hat x^k}{2}\) + \hat p^k \hat p^k \, \hat p^i \hat p^i \] \, .
\ee
Its VEV is easy to be computed and gives
\be
\Delta E^{(1)} = \langle 0 | \Delta H_{(4)} | 0 \rangle = \frac{d(d+2\, h)}{8} \, .
\ee

%How can we fix the ordering parameter $h$? Is there any sigma model argument that could help us? For the time being, t
The comparison with both the effective approach through \eqref{link} and the outcome of the numerical analysis strongly suggests that the correct ordering is selected by
\be
h=1 \, .
\ee
A possible way to justify this choice is to consider that the eigenvalue equation for the winding string ground state, should correspond to the equation of motion for the thermal scalar field in the low energy effective description. In this perspective only for $h=1$ this equation might descend from a general covariant effective action.\footnote{We thank Troels Harmark for pointing out to us this possibility.}

\subsection{Witten-Yang-Mills model}
\label{sec:wym}

As a second example let us consider the WYM theory \cite{Witten:1998zw}.
In the strongly coupled large $N$ limit, the WYM theory at finite temperature (in the confining phase) is reliably described by the Type IIA supergravity theory on the background
\be \label{WbackOld}
\begin{split}
	&ds^2 =  \( \frac{u}{R} \)^{\hspace{-2pt}\frac{3}{2}} \( \eta_{\mu\nu} dx^\mu dx^\nu + \frac{4}{9 m_0^2} f(u) d\theta^2 \) + \( \frac{R}{u} \)^{\hspace{-2pt}\frac{3}{2}} \frac{du^2}{f(u)} +R^{\frac{3}{2}} u^{\frac{1}{2}} d\Omega_4^2 \, , \\
	&m_0^2 = \frac{u_0}{R^3} \, , \quad f(u) = 1 - \frac{u_0^3}{u^3} \, , \quad e^{\phi} = g_s \frac{u^{\frac{3}{4}}}{R^{\frac{3}{4}}} \, , \quad  R = \( \pi N g_s \)^{\frac{1}{3}} \alpha'^{\frac{1}{2}} \, , \quad F_4 = 3 R^3 \omega_4 \, ,
\end{split}
\ee
where $\mu$, $\nu = 0,1,2,3$ and $\omega_4$ is the volume form of the unit $S^4$. Notice that the dilaton $\phi$ has a running behavior, while the Ramond-Ramond field strength $F_4$ is constant. As usual, the time direction is compactified on a circle of length $\beta=1/T$, $u \in [ u_0, +\infty )$ is the holographic coordinate and $\theta$ is an angular coordinate such that $\theta \simeq \theta + 2\pi$. The latter two variables describe the so-called \emph{cigar} of the geometry, whose asymptotic radius is given by the inverse of the Kaluza-Klein (KK) mass scale 
\be \label{MKK}
M_{KK} = \frac{3}{2} \, m_0\,,
\ee 
and vanishes at $u=u_0$, that is at the so-called \emph{tip}. 

The region $u \sim u_0$ is dual to the IR regime of the WYM theory, exactly what we are interested in from an effective point of view. 
Passing to the coordinate $r$ to put the metric in the form (\ref{metricnlo}), one has
\be
u-u_0= \frac34  m_0 r^2 \left(1-\frac{1}{8\sqrt{u_0 R^3}}r^2 \right)\,,
\ee
and from the expansion of the metric in this coordinate one has the non-zero coefficients 
\bea \label{wymcoeffs}
&& d=2\,,\quad p=3\,,\quad  \phi_2 = \frac{9}{16 m_0 R^3}\,, \quad    \tilde{g}^{(0)}_{2}= \frac{9}{8 m_0 R^3}\,,\quad  \tilde{g}^{(0)}_{4}= \frac{9}{128 (m_0 R^3)^2}\,, \nonumber\\
&& \tilde{g}^{(\Omega)}_{2}= -\frac{1}{2 m_0 R^3}\,,\quad  \tilde{g}^{(\cal M)}_{2} =\frac{3}{8 m_0 R^3}\,.
\eea
Thus, considering that 
\be 
2\pi\alpha' T_s=\left(\frac{u_0}{R}\right)^{\hspace{-2pt}\frac32}\,,
\ee
and using (\ref{c0}), (\ref{c1}) and (\ref{shift1}) in (\ref{nearHag}) we find
\be \label{beta2wym}
\beta_H^2 = \frac{4\pi}{T_s} \( 1 - \frac{M_{KK}}{\sqrt{2\pi T_s}}+\frac{8\log{2}-1}{4} \left(\frac{M_{KK}}{\sqrt{2\pi T_s}}\right)^{\hspace{-2pt}2} + ...\) \,.
\ee
The NNL term of this expression is new.

\subsubsection{Covariant approach}
\label{sec:wymcov}

As a non-trivial example of the covariant form of the coefficients in section \ref{sec:covariant}, we consider the WYM background in the original form (\ref{WbackOld}).
In this case the metric depends on the radial variable $u$, which plays the role of $\rho$ of section \ref{sec:covariant} and has its IR value at $u_0$.
The normalized vectors are
\be 
\hat V = \left( \frac{R}{u_0}\right)^{\hspace{-2pt}\frac34} \partial_t \,, \qquad W= \( \frac{u}{R} \)^{\hspace{-2pt}\frac34} \sqrt{f(u)}\, \partial_u \,,
\ee
and the functions to be derived are
\be 
G = \left( \frac{R}{u_0}\right)^{\hspace{-2pt}\frac32} \left( \frac{u}{R}\right)^{\hspace{-2pt}\frac32}\,, \qquad g^{uu}\sqrt{|g|} = \frac{u^{\frac{13}{4}} \sqrt{u^3-u_0^3}}{R^{\frac34}} \,, \qquad \phi = \log{\left[g_s \left(\frac{u}{R} \right)^{\hspace{-2pt}\frac34}\right]}\,.
\ee

Deriving $G$, one finds
\bea 
&& \tilde G^{(0)}_1 = W^{u} \nabla_{u} G\, |_{u=u_0} = 0\,, \nonumber \\
&& \tilde G^{(0)}_2 = \frac{1}{2} (W^{u})^2 \nabla_{u}\nabla_{u} G \, |_{u=u_0} = \frac{9}{8 m_0 R^3 }\,,  \nonumber \\
&& \tilde G^{(0)}_3 = \frac16 (W^{u})^3 \nabla_{u}\nabla_{u}\nabla_{u} G\, |_{u=u_0} = 0\,, \nonumber \\
&& \tilde G^{(0)}_4 = \frac{1}{24} (W^{u})^4 \nabla_{u}\nabla_{u}\nabla_{u}\nabla_{u} G \, |_{u=u_0} = \frac{9}{128 (m_0 R^3)^2}\,. 
\eea
Concerning the coefficients $\tilde g_n$, while $g^{uu}\sqrt{|g|}=0$ at $u=u_0$, the first derivative is non-trivial
\be 
C = W^{u}\nabla_{u} \left(g^{uu}\sqrt{|g|} \right)\, |_{u=u_0} = \frac{3 u_0^{\frac92}}{2 R^{\frac32}}\,,
\ee
so that $d-1=1$.
Then one finds
\bea
&& \tilde g_{1} = \frac{1}{2C} (W^{u})^2 \nabla_{u}\nabla_{u} \left(g^{uu}\sqrt{|g|} \right)\, |_{u=u_0} =0\,, \nonumber \\
&& \tilde g_{2} = \frac{1}{6C} (W^{u})^3 \nabla_{u}\nabla_{u}\nabla_{u} \left(g^{uu}\sqrt{|g|} \right)\, |_{u=u_0} = \frac{11}{4 m_0 R^3}\,.
\eea
Finally
\be
\phi_1 =  W^{u}\nabla_{u} \phi\, |_{u=u_0}=0\,, \qquad 
\phi_2 = \frac{1}{2} (W^{u})^2\nabla_{u}\nabla_{u} \phi\, |_{u=u_0} = \frac{9}{16 m_0 R^3} \,.
\ee
All of these results coincides with the ones obtained above in (\ref{wymcoeffs}).

\subsection{Maldacena-N\'u\~nez and its softly broken susy deformation}

As a third example, let us consider the Maldacena-Nu\~nez (MN) model \cite{Maldacena:2000yy,Chamseddine:1997nm} and its soft supersymmetry breaking deformation \cite{Gubser:2001eg,Evans:2002mc}. The metric reads
\be
ds^2 = e^{\phi(\rho)} \left[ dx_{\mu} dx^{\mu} + \alpha' N \left( d\rho^2 +  e^{2g(\rho)}(d\theta_1^2 + \sin^2{\theta_1} d\varphi_1^2) + \frac14 \sum_a (w^a - A^a)^2 \right) \right]\,,
\ee
with $SU(2)$ gauge field (with generators $\sigma^a$)
\be 
A= \frac12 \left[\sigma^1 a(\rho) d\theta_1 + \sigma^2 a(\rho) \sin{\theta_1} d\varphi_1 + \sigma^3 \cos{\theta_1} d\varphi_1  \right]\,,
\ee
and one-forms
\be 
w^1 + i w^2 = e^{-i\psi} (d\theta_2 + i \sin{\theta_2}d\varphi_2)\,.
\ee
The functions of $\rho$ appearing in the formulas above have the following expressions in the supersymmetric MN case
\bea
e^{2\phi(\rho)} & = & g_s^2 \frac{\sinh{2\rho}}{2 e^{g(\rho)}}\,, \\
e^{2g(\rho)} & = & \rho \coth{2\rho} - \frac{\rho^2}{\sinh^2{2\rho}} -\frac14 \,, \\
a(\rho) & = &  \frac{2\rho^2}{\sinh{2\rho}}\,.
\eea 
In the softly-broken case their IR asymptotics read\footnote{The quartic term of the dilaton is new and it is determined by solving equations (76) of \cite{Gubser:2001eg} at the next order with respect to the solution in (78) of that paper. Note that there is a typo (a missing derivative on $w$) in the second term of the equation for $w$ in (76).}
\bea
\phi(\rho) & = & \log{g_s} +  \left(\frac13 +\frac{b^2}{4} \right)\rho^2 + 4 \left(\frac{b^4}{160}-\frac{b^3}{40}-\frac{b^2}{72}-\frac{1}{135}\right) \rho ^4 + ...\,, \\
e^{g(\rho)} & = & \rho - \left(\frac19 +\frac{b^2}{4} \right)\rho^3 + ... \,, \\
a(\rho) & = & 1-b \rho^2+ ...\,.
\eea 
In these expressions the parameter $b \in [0,2/3]$ measures the departure from the supersymmetric configuration, which is attained for $b=2/3$.
In the following we consider the general case with unfixed $b$.

The six-dimensional part of the geometry consists of a non-trivial fibration.
Employing for example the expression of the IR metric in formula (10) of \cite{Bigazzi:2004yt}, one gets for the determinant of the metric the expansion
\be\label{sqrt6}
\sqrt{|g|} = C r^2 \left(1+\frac{15 b^2+32 }{36 (g_s N \alpha')}\,  r^2 + ...   \right)\,,
\ee
where $r$ is related to the $\rho$ radial variable as
\be
\rho = \frac{1}{\sqrt{g_s N \alpha'}} r -\frac{\frac13 +\frac{b^2}{4}}{6(g_s N \alpha')^{\frac32}} r^3 +...\,.
\ee
By also expanding the $g_{tt}$ component of the metric one has 
\be
 d=3\,,\quad  \phi_2 = \tilde{g}^{(0)}_{2}= \frac{\frac13 +\frac{b^2}{4}}{\alpha' N g_s}\,,\quad   \tilde{g}^{(0)}_{4}= \frac{51b^4-144b^3-40b^2-16}{1440 (\alpha' N g_s)^2}\,,%\nonumber\\ &&
\quad \tilde{g}_{2} =\frac{32+15b^2}{36\alpha' N g_s}\,.
\ee
With these values and using the relations for the string tension and KK scale\footnote{Note that these values do not depend on the parameter $b$.} 
\be 
2\pi\alpha' T_s= g_s\,,\qquad M_{KK}=\frac{1}{\sqrt{N \alpha'}}\,,
\ee
one gets from  (\ref{c0}), (\ref{c1}), (\ref{shift1})
and (\ref{nearHag})  
\bea
\beta_H^2 & = & \frac{4\pi}{T_s} \Bigl[ 1 - \sqrt{\frac{12+9b^2}{8}}\,\frac{M_{KK}}{\sqrt{2\pi T_s}} +  \\
&& \quad \quad +\left( \frac{243 b^4+432 b^3+888 b^2+368}{576 b^2 +768 }+\frac{1}{2} \left(3 b^2+4\right) \log{2} \right) \left(\frac{M_{KK}}{\sqrt{2\pi T_s}}\right)^{\hspace{-2pt}2} + ...\Bigr] \,. \nonumber
\eea
%This expression is new.
In the supersymmetric MN case $b=2/3$ this reads
\be
\beta_H^2 = \frac{4\pi}{T_s} \( 1 - \sqrt{2}\,\frac{M_{KK}}{\sqrt{2\pi T_s}}+ \frac{11+32\log{2}}{12}\left(\frac{M_{KK}}{\sqrt{2\pi T_s}}\right)^{\hspace{-2pt}2} + ...\) \,.
\ee
Again, the NNL term of this expression is new.

\subsection{Klebanov-Strassler}
The Klebanov-Strassler model \cite{Klebanov:2000hb} has string-frame metric
\bea 
ds^2 & = & g_s^{\frac12} \Biggl[ h(\tau)^{-\frac12} dx_{\mu}dx^{\mu} +  \\
&& + h(\tau)^{\frac12} \frac{\epsilon^{\frac43}}{2}K(\tau) \left[\frac{d\tau^2 + g_5^2}{3K(\tau)^3} +\cosh^2{\left(\frac{\tau}{2}\right)}\,(g_3^2+g_4^2)+\sinh^2\left({\frac{\tau}{2}}\right)\,(g_2^2+g_3^2)\right]  \Biggr]\,,\nonumber
\eea
where the $g_i$ are a basis of one-forms and the expression of the function $K(\tau)$ is
\be 
K(\tau) = \frac{(\sinh{(2\tau)}-2\tau)^{\frac13}}{2^{\frac13}\sinh{\tau}}\,.
\ee 
The function $h(\tau)$ reads
\be 
h(\tau) = (g_s M \alpha')^2 2^{\frac23} \epsilon^{-\frac83} \int_{\tau}^{\infty} dx \frac{x \coth{x}-1}{\sinh^2{x}} (\sinh{(2x)}-2x)^{\frac13} \,.
\ee 
Note that $\epsilon^{\frac23}$ has the dimension of a length.

One can pass to the standard radius $r$ of (\ref{metricnlo}) with the change of variable 
\be 
\tau = \frac{2^{\frac23} 3^{\frac16}}{a_0^{\frac14}g_s^{\frac34}\sqrt{ M \alpha'}}\, r - \frac{6^{\frac43} a_0-5}{2^{\frac13} 3^{\frac{11}{6}} 5 a_0^{\frac74} g_s^{\frac94} ( M \alpha')^{\frac32}}\, r^3 + ...\,,
\ee 
where $a_0 \sim 0.718$ is the value of the integral in $h(0)$.

Considering that the dilaton is constant, the data one can extract from these expressions are 
\bea 
&& d=3\,,\quad \phi_2 = 0\,, \quad  \tilde{g}^{(0)}_{2}= \frac{1}{3 a_0^{\frac32}g_s^{\frac32} M \alpha'}\,, \quad  \tilde{g}^{(0)}_{4}= \frac{55-42\cdot 6^{\frac13}a_0}{270 a_0^3 g_s^3 ( M \alpha')^2}\,,\nonumber \\ 
&& \tilde{g}_{2} =\frac{6^{\frac43} a_0-1}{18 a_0^{\frac32} g_s^{\frac32} M \alpha'}\,.
\eea
Using the following relations for the string tension and KK scale \cite{Klebanov:2000hb}
\be  
2\pi\alpha' T_s= \frac{\epsilon^{\frac43}}{2^{\frac13}a_0^{\frac12}g_s^{\frac12} M \alpha'}\,,\qquad  M_{KK}=\frac{\epsilon^{\frac23}}{g_s M \alpha'}\,, 
\ee
one gets from  (\ref{c0}), (\ref{c1}),  (\ref{shift1}) and (\ref{nearHag}) 
\be
\beta_H^2 = \frac{4\pi}{T_s} \( 1 - \frac{\sqrt{3}}{2^{\frac23}a_0}\,\frac{M_{KK}}{\sqrt{2\pi T_s}}+ \frac{6^{\frac43} a_0-5+2^5\log{2}}{2^{\frac{13}{3}} a_0^{2}}\left(\frac{M_{KK}}{\sqrt{2\pi T_s}}\right)^2 + ...\) \,.
\ee
Again, the NNL term of this expression is new.

\subsection{New confining backgrounds from wrapped D5-branes}

In this section we consider the confining theories whose dual backgrounds have been recently derived in \cite{Nunez:2023xgl} from wrapped D5-branes.
For these theories there is no available study of the Hagedorn temperature yet.

The first example is a background dual to a confining $(2+1)$-dimensional field theory, whose metric is
\be 
ds^2 = \rho \left[ dx_{\mu} dx^{\mu} + \left( 1- \frac{m}{\rho^2} \right) d\mu^2 + \frac{2}{\rho^2 - m} d\rho^2 +d\theta^2 + \sin^2{\theta}d\varphi^2 + \sum_{i=1}^3 (\Theta^i)^2    \right]\,,
\ee
where $\{\Theta^i\}$ is a basis of one-forms.
The IR tip of the geometry is attained at $\rho = \sqrt{m}$.
The dilaton reads
\be 
\phi= \log{\left( \frac{4}{N}\rho \right)}\,.
\ee

The change of radial variable to the standard $r$ coordinate in (\ref{metricnlo}) reads
\be 
\rho -\sqrt{m} = \frac14 r^2 -\frac{1}{96\sqrt{m}} r^4 + ... \,.
\ee 
Then we find
\be
d=2\,,\quad \phi_2 = \tilde{g}^{(0)}_{2}= \frac{1}{4\sqrt{m}}\,, \quad \tilde{g}^{(0)}_{4}= -\frac{1}{96m}\,, \quad\tilde{g}_{2} =\frac{11}{12\sqrt{m}}\,, 
\ee 
from which, considering that 
\be 
2\pi T_s=\sqrt{m}\,, \qquad M_{KK}=\frac{1}{\sqrt{2}}\,,
\ee
we get
\be
\beta_H^2 = \frac{4\pi}{T_s} \( 1 - \frac{M_{KK}}{\sqrt{2\pi T_s}} + \frac{8\log{2}-1}{4}\left(\frac{M_{KK}}{\sqrt{2\pi T_s}}\right)^{\hspace{-2pt}2} + ...\) \,.
\ee
It is curious to note that the result coincides with the one for the WYM theory (\ref{beta2wym}).

Analogously, the background dual to a $(4+1)$-dimensional confining theory has metric
\be 
ds^2 = \rho \left[ dx_{\mu} dx^{\mu} + f_s(\rho) d\varphi^2 + \frac{N}{\rho^2 f_s(\rho)} d\rho^2 + \frac{N}{4} \left[ \omega_1^2 + \omega_2^2 + \left( \omega_3 -\sqrt{\frac{8}{N}} Q \zeta(\rho) d\varphi \right)^{\hspace{-2pt}2} \right]   \right] \,,
\ee
with one-forms $\omega_i$ and
\be 
f_s(\rho)= 1- \frac{m}{\rho^2}-\frac{2Q^2}{\rho^4}=\frac{(\rho^2-\rho_+^2)(\rho^2-	\rho_-^2)}{\rho^4}\,,\quad  \zeta(\rho) = \frac{1}{\rho^2} - \frac{1}{\rho_+^2}\,, \quad 2\rho_{\pm}^2 = m \pm \sqrt{m^2+8Q^2}\,,
\ee
while the dilaton is simply
\be 
\phi= \log{\rho}\,.
\ee
The metric has its IR tip at $\rho=\rho_+$.

The change of radial variable now is
\be 
\rho - \rho_+ =  \frac{\rho_+^2-\rho_-^2}{2N \rho_+^2}r^2 - \frac{5\rho_-^4+\rho_+^4-6\rho_-^2\rho_+^2}{24 N^2 \rho_+^5}r^4 + ...\,.
\ee
We get
\be 
d=2\,,\quad \phi_2 = \tilde{g}^{(0)}_{2}= \frac{\rho_+^2- \rho_-^2}{2N\rho_+^3}\,,\quad   \tilde{g}^{(0)}_{4}= -\frac{\rho_+^4+5\rho_-^4 -6\rho_+^2\rho_-^2}{24N^2\rho_+^6}\,, \quad  \tilde{g}_{2} =\frac{11\rho_+^2- 7\rho_-^2}{6N\rho_+^3}\,, 
\ee 
and, considering that 
\be 
2\pi T_s=\rho_+\,,\qquad M_{KK}=\frac{\rho_+^2-\rho_-^2}{\sqrt{N}\rho_+^2}\,, 
\ee
the result for the inverse Hagedorn temperature can be written as
\be
\beta_H^2 = \frac{4\pi}{T_s} \( 1 - b\,\frac{M_{KK}}{\sqrt{2\pi T_s}}+ \left(\frac34 b^2-b^4+2\log{2}\right)\left(\frac{M_{KK}}{\sqrt{2\pi T_s}}\right)^{\hspace{-2pt}2} + ...\) \,,
\ee
where we have defined the parameter 
\be 
b \equiv \frac{\rho_+}{\sqrt{\rho_+^2-\rho_-^2}}\,,
\ee
which is equal to $1/\sqrt{2}$ in the supersymmetric case $m=0$ ($\rho_+^2=-\rho_-^2$).

%%%%%%%%%%%%%%%%%%%%%%%%%%%%%%%%%%%%%%%%%%%
\section{Summary and conclusions}
\label{sec:conclu}
In this paper we have extended the recent progress in calculating the string Hagedorn temperature $T_H$ on non-trivial spaces to a very general class of backgrounds dual to ``confining'' theories, be them CFTs on spheres or truly confining models.
These backgrounds can include non-trivial RR-forms, dilaton and NSNS field.
We have derived general formulas for the value of $T_H$ up to NNNL order in the $\alpha'$ expansion (see section \ref{sec:results}) by exploiting two complementary approaches - the string sigma model one and the ``thermal scalar'' effective one.
In both cases $T_H$ is defined as the temperature at which the lightest mode winding the temporal direction becomes tachyonic. 

In the sigma-model approach in sections \ref{sec:quadmodel} and \ref{sec:extrapol} we have considered the quadratic fluctuations of the string modes around general classical configurations which admit a limit to the Hagedorn regime.
This allows for the calculation of the zero-point energy and consequently the energy of the lightest mode winding the temporal direction.
While the expression for the squared sum of bosonic fluctuation masses is well known, we have provided the explicit derivation of the corresponding squared sum of fermionic fluctuation masses, which is not easily found in the literature.

The analysis of the string quadratic fluctuations is enough to provide the full LO and NLO values of $T_H$, and, crucially, a $\log{2}$ term at NNLO which is not fixed by the effective approach. 
In principle, the sigma-model should allow to derive the complete NNLO result.
In the case of $AdS$ backgrounds we have indeed shown that the structure of the quartic fluctuations of the zero-modes encodes the full NNLO result. %, up to an ambiguity (which we have not resolved yet).  

Interestingly, %this ``ambiguous'' 
the NNLO term (but not the $\log{2}$ term!) is provided by the effective ``thermal scalar'' approach of section \ref{sec:effnlo}. 
The latter gives also the full LO and NLO terms and, combined with the $\log{2}$ term mentioned above, provides also the NNNLO result.
In this approach the mode becoming tachyonic at $T_H$ is described by an effective scalar field.
Vanishing of its bare mass provides the LO result, while quantum corrections extend the results to higher orders.
At NLO its equation of motion always gives a harmonic oscillator, while at higher orders one can simply use perturbation theory.

As already stressed, the combination of the two approaches gives easily the value of $T_H$ up to NNNLO for very general backgrounds.
In the last part of the paper (section \ref{sec:examples}) we have provided a number of examples.
In particular, we have derived the Hagedorn temperature up to NNNLO for the NSNS-supported $AdS_3$ background.
For the non-$AdS$ examples, we have limited the presentation of results to the NNLO, since one expects $\alpha'$ corrections to the geometry to come into play at NNNLO.
In turn, we have provided the NNLO value of $T_H$ for the WYM, (softly broken) MN and KS models, and for some confining theories recently constructed by means of wrapped D5-branes.

Most of the examples we have considered are not integrable theories.
Thus, on one hand the results of this paper do not constitute further precision tests of holography as the ones pertaining to the $AdS$ cases \cite{Ekhammar:2023cuj,Bigazzi:2004yt,Ekhammar:2023glu}.
On the other hand, they provide predictions for the value of the Hagedorn temperature in the strong coupling regime of the dual QFTs up to second order corrections, which is quite non-trivial.

It would be clearly interesting to extend the results of this paper to even more general models and, possibly, to some other regimes of parameters.
%Moreover, from a technical point of view it would be nice to resolve the ambiguity found in the sigma-model quartic fluctuation result.    
Finally, one might consider a different route to the sigma-model calculation where instead of beginning with a semiclassical configuration and extrapolate the results to the Hagedorn regime, one starts directly with a near-Hagedorn configuration.
The results of this approach, which are fully compatible with the ones of this paper up to NNLO, will be presented in \cite{w}. 

%%%%%%%%%%%%%%%%%%%%%%%%%%%%%%%%
\vskip 15pt \centerline{\bf Acknowledgments} \vskip 10pt 

\noindent 
We are grateful to Andrea Cavagli\'a, Matteo Ciardi, Lorenz Eberhardt, Valentina Giangreco M.~Puletti, Troels Harmark, Domenico Seminara, and Dima Sorokin for comments and very helpful discussions. 

%%%%%%%%%%%%%%%%%%%%%%%%%%%%%%%%%%%%%%%%%%%%%%%%%%%%%%%%%%

\appendix

\section{Notations and conventions}
\label{app:notations}

For the reader's convenience, we resume in this appendix the conventions and notations we have adopted throughout the paper. Here, we will refer to a generic $d$-dimensional background whose metric $g_{\mu\nu}$ has determinant $g$ and having $n_t$ time-like directions. Everything will be expressed in a covariant way, whereby nothing will depend on what the time-like directions are.

\subsubsection*{Antisymmetrizer and symmetrizer}
Let $\Omega$ be an object which depends in some way on a set of indices $\{\mu_1,...,\mu_k\}$, that is
\be
\Omega=\Omega\(\{\mu_1,...,\mu_k\}\) \, .
\ee
We define the symmetrizer $(...)$ as the operation
\be
\Omega\(\(\{\mu_1,...,\mu_k\}\)\) = \frac{1}{k!} \sum_{\pi} \Omega\(\{\mu_{\pi(1)},...,\mu_{\pi(k)}\}\) \, ,
\ee
where the sum runs over all the possible permutations $\pi$ of $k$ objects. 
Similarly, we define the antisymmetrizer $[...]$ as the operation
\be
\Omega\(\[\{\mu_1,...,\mu_k\}\]\) = \frac{1}{k!} \sum_{\pi} \sign\(\pi\) \Omega\(\{\mu_{\pi(1)},...,\mu_{\pi(k)}\}\) \, ,
\ee
being $\sign\(\pi\)$ the parity of the permutation $\pi$.

\subsubsection*{The Levi-Civita symbol}

We define the target space Levi-Civita symbol $\varepsilon_{\mu_1...\mu_d}$ as the total antisymmetric symbol such that
\be
\varepsilon_{1...d} = +1 \, .
\ee
%where the $n$($s$)-indices stand for the time(space)-like directions, sorted in ascending order.
All the other components are related to the above one by an even or odd permutation of the indices. Further,
\be
\varepsilon^{1\,...\,d} = g^{\mu_1 1} \cdots g^{\mu_d d} \, \varepsilon_{\mu_1...\mu_d} = g^{-1} \, .
\ee
Thus, also the symbol $g \, \varepsilon^{\mu_1\, ... \, \mu_d}$ has non-zero components equal to $\pm 1$ and it follows
\be
\varepsilon^{\mu_1 \, ... \,  \mu_p \nu_1 \, ... \, \nu_q} \, \varepsilon_{\mu_1 \, ... \, \mu_p \kappa_1 \, ... \, \kappa_q} = \frac{p! \, q!}{g} \, \delta^{[\nu_1}{}_{\kappa_1} \cdots \delta^{\nu_q]}{}_{\kappa_q} \, .
\ee

\subsubsection*{Operation with differential forms}

Now, let us recall well-known formulae about differential forms. Let $\alpha$ and $\beta$ be respectively a $k$-form and $l$-form, that is
\be
\alpha = \frac{1}{k!} \, \alpha_{\mu_1...\mu_k} \, \rmd x^{\mu_1} \wedge ... \wedge \rmd x^{\mu_k} \, , \quad \beta = \frac{1}{l!} \, \beta_{\mu_1...\mu_l} \, \rmd x^{\mu_1} \wedge ... \wedge \rmd x^{\mu_l} \, .
\ee

We define their wedge product as
\be
\(\alpha \wedge \beta\)_{\mu_1...\mu_k\nu_1...\nu_l} = \frac{(k+l)!}{k! \, l!} \, \alpha_{[\mu_1...\mu_k} \, \beta_{\nu_1...\nu_l]} \, ,
\ee
while the exterior derivative as
\be
\(\rmd \alpha\)_{\rho\mu_1...\mu_k} = \frac{(k+1)!}{k!} \, \partial_\rho \alpha_{\mu_1 ... \mu_k} \, .
\ee

Then, %in a generic Riemannian ($s=0$) or Lorentzian ($s=1$) $d$-dimensional background, 
we define the Hodge star operator as
\begin{subequations}
\begin{align}
&(*\alpha)_{\nu_1...\nu_{d-k}} = \frac{1}{k!} \, \sqrt{|g|} \, \alpha_{\mu_1...\mu_k} \, \varepsilon^{\mu_1...\mu_k}{}_{\nu_1...\nu_{d-k}} \, , \\
&*\hspace{-2pt}* \, \alpha = \sign(g) (-1)^{k(d-k)} \, \alpha \, , \\
&\label{covdiv}\(*d\hspace{-2pt}*\alpha\)_{\mu_1...\mu_{k-1}} = \sign (g) (-1)^{(k-1)(d-k)} \frac{1}{\sqrt{|g|}} \, \partial_\rho \(\sqrt{|g|} \, \alpha^{\rho\nu_1...\nu_{k-1}}\) g_{\nu_1\mu_1} \cdots g_{\nu_{k-1}\mu_{k-1}}\, .
\end{align} 
\end{subequations}
Here, $\sign(g)$ stands for the sign of determinant of the background metric. Finally, let us stress that these properties rely on the conventions for the Levi-Civita symbol we have reported in the previous subsection.

\subsubsection*{Gamma matrices}
We define the set of Gamma matrices $\{\Gamma^1, \ldots , \Gamma^d\}$ as the generators of a matrix representation of the Clifford algebra $\text{Cl}_{n_t,d-n_t}(\mathbb R)$. Therefore, they satisfy the anticommutation relations
\be
\{\Gamma_\mu, \Gamma_\nu\} = 2 \, g_{\mu\nu} \, \mathbb I_{2^{d/2}} \, .
\ee
Moreover, it holds that
\begin{subequations}
\begin{align}
& \Gamma_\mu^\dagger = - \Gamma_\mu \, , \quad \Gamma^2_\mu = - \mathbb I_{2^{d/2}} \quad \text{if $p$ is time-like} \, , \\
& \Gamma_\mu^\dagger = + \Gamma_\mu \, , \quad \Gamma^2_\mu = + \mathbb I_{2^{d/2}} \quad \text{if $p$ is space-like} \, .
\end{align}
\end{subequations}

In even dimension, we can also define a chirality operator $\Gamma^{(d+1)}$ as
\be
\Gamma^{(d+1)} = \(-i\)^{\frac{d}{2} + n_t} \frac{1}{d!} \, \varepsilon_{\mu_1 ... \mu_d} \, \Gamma^{\mu_1...\mu_d} \, ,
\ee
being
\be
\Gamma^{\mu_1...\mu_d} = \Gamma^{[\mu_1} \cdots \Gamma^{\mu_d]}= \frac{1}{d!} \sum_{\pi} \sign\(\pi\) \Gamma^{\mu_{\pi\(1\)}} \cdot ... \cdot \Gamma^{\mu_{\pi\(d\)}}\, .
\ee
It is such that
\be \label{chiralityproperties}
\(\Gamma^{(d+1)}\)^\dagger = \Gamma^{(d+1)} \, , \quad \{\Gamma^{(d+1)}, \Gamma^\mu\} = 0 \, , \, \forall \, \mu \, , \quad \(\Gamma^{(d+1)}\)^2 = \mathbb I_{2^{d/2}} \, .
\ee
Given these properties, it follows that there exist two eigenspaces of spinors corresponding to positive and negative eigenvalues of $\Gamma^{(d+1)}$. Finally, it holds that
\be \label{gammadualityreversed}
\Gamma^{\mu_1 ... \mu_k} = (-1)^{n_t} \, \frac{i^{\frac{d}{2}+n_t}}{(d-k)!} \, \varepsilon^{\mu_1 ... \mu_d} \, \Gamma^{(d+1)} \, \Gamma_{\mu_d ... \mu_{k+1}} \, .
\ee

All these properties are nicely collected in \cite{Kuusela:2019iok}. There, the reader can find further details on Gamma matrices and the associated Clifford algebra. What we reported here is just what we need for this work. Let us stress that the definitions of this subsection do not depend on the conventions for the Levi-Civita tensor. Indeed, everything has been written in a covariant way.

%%%%%%%%%%%%%%%
\section{Type II equations of motion in string frame}
\label{app:sugraeom}
The Einstein equations in Ricci form for type II theories in string frame can be written, using the dilaton equation of motion, as (see e.g. \cite{Hamilton:2016ito})
\be
R_{pq}-\frac12 |H_{(3)}|_{pq}^2 - \frac12 e^{2\phi} \sum_n \{G_{(n)}\}^2_{pq} = -2 \, \nabla_p\nabla_q \phi\,,
\label{Einsteq}
\ee
where we are considering the {\it non-democratic} formalism so that in type IIA $n=2,4$, while in type IIB $n=1,3,5$. Moreover we have defined
\bea
\{G_{(n)}\}^2_{pq}&=& |G_{(n)}|^2_{pq} - \frac12 g_{pq} |G_{(n)}|^2\,,\quad (n\neq 5)\,\nonumber \\
\{G_{(5)}\}^2_{pq}&=& \frac12 |G_{(5)}|^2_{pq}\,,
\label{eomsymb}
\eea
where for a generic $n$-form $W_{(n)}$
\bea
|W_{(n)}|^2_{pq} &=& \frac{1}{(n-1)!} W_{p p_2\cdots p_n} W_{q}^{p_2\cdots p_n}\,,\nonumber \\
|W_{(n)}|^2&=& \frac{1}{n!} W_{p_1 p_2\cdots p_n} W^{p_1 p_2\cdots p_n}\,.
\eea
Let us also recall that
\bea
G_{(1)}&=&F_{(1)}\,,\quad G_{(2)}=F_{(2)}\,,\nonumber \\
G_{(3)} &=& F_{(3)} - C_{(0)} H_{(3)}\,,\nonumber \\
G_{(4)} &=& F_{(4)} - C_{(1)}\wedge H_{(3)}\,,\nonumber \\
G_{(5)} &=& F_{(5)} - \frac12 H_{(3)}\wedge C_{(2)} +\frac12 F_{(3)}\wedge B_{(2)}\,.
\eea
As usual, the self-duality condition $G_{(5)}=* G_{(5)}$, which implies $|G_{(5)}|^2 =0$, is implemented at the level of the equations of motion.

The dilaton equation reads
\be
\Delta\phi= 2 |d\phi|^2 -\frac12 |H_{(3)}|^2 + \frac{e^{2\phi}}{4}\sum_n (5-n)|G_{(n)}|^2\,.
\ee
Finally, the Maxwell equations in the IIA case are
\begin{align}
\label{IIAmaxwell}
&d * G_{(2)}= H_{(3)} \wedge *G_{(4)} \,, \quad d * G_{(4)}=-H_{(3)} \wedge G_{(4)} \,,\nonumber \\
&d\(e^{-2\phi} *H_{(3)}\) = \frac12 G_{(4)} \wedge G_{(4)} - G_{(2)} \wedge *G_{(4)} \,.
\end{align}
On the other hand, in the IIB case they read
\begin{align}
\label{IIBmaxwell}
&d * G_{(1)}=- H_{(3)} \wedge *G_{(3)} \,, \quad d * G_{(3)}=- H_{(3)} \wedge G_{(5)}\,, \quad d * G_{(5)}= H_{(3)} \wedge G_{(3)} \,,\nonumber \\
&d\(e^{-2\phi} *H_{(3)}\) = G_{(1)} \wedge *G_{(3)} + G_{(3)} \wedge *G_{(5)} \,.
\end{align}

In the \emph{democratic} formalism equation \eqref{Einsteq} still holds, as long as we replace (\ref{eomsymb}) with
\be
\{G_{(n)}\}^2_{pq}= \frac12 |G_{(n)}|^2_{pq} \, , \quad \forall \, n = 1,...,9 \, .
\ee
Moreover, in this alternative formalism, the relevant Maxwell equation for our purposes takes the form
\be
\label{demomaxwell}
d\(e^{-2\phi} *H_{(3)}\) = \frac12 \sum_n G_{(n)} \wedge *G_{(n+2)} \, ,
\ee
where the sum runs over $n=2,4,6,8$ in the IIA case or $n=1,3,5,7,9$ in the IIB one. For the latter two formulae, see e.~g.~\cite{Cassani:2008rb,Baron:2023oxk}.

%%%%%%%%%%%%%%%%%%%%%%%%%%%%%%%%%%%%%%%%

\section{Useful formulae for Gamma matrices}
\label{app:useful}
In this appendix, we aim to collect and prove some useful formulae. The building blocks for what follows are
\begin{align}
&\widetilde \Gamma_I = \frac{1}{16} \, e^{\phi} \sum_n \, (-1)^{f_{(n,I)}}  \, \frac{1}{n!} \, G^{(n)}_{\underline{a_1...a_n}} \, \Gamma^{\underline{a_1...a_n}} \, , \nonumber  \\
& \widetilde\Gamma_{I\alpha} = \frac18 \, e^{\phi} \, \partial_\alpha X^p \, e^{\underline p}{}_{p} \, \sum_n (-1)^{f_{(n,I)}} \, \frac{1}{(n-1)!} \,  G^{(n)}_{\underline{p \, a_2 ... a_n}} \,  \Gamma^{\underline{a_2 ... a_n}} \, , \nonumber
\end{align}
with
\be
f_{(n,I)} = I \[\frac{n(n-1)}{2} + n +1\] + n+1 \, .
\ee
Here, $I$ takes the value $1$ or $2$, while $n=2,4,6,8$ for the IIA case and $n=1,3,5,7,9$ for the IIB one.

Let us start with
\begin{lemma} \label{lemma1}
Given a subset $\{\Gamma^{\underline{a_1}},..., \Gamma^{\underline{a_{2m-1}}}\}$ of the $d$-dimensional Gamma matrices in flat space with odd cardinality, we have
\benn
\{\Gamma^{\underline{a_1}} \cdot\cdot\cdot \Gamma^{\underline{a_{2m-1}}}, \Gamma_{\underline q}\} = \sum_{i=1}^{2m-1} (-1)^{i+1} \, 2 \, \delta^{\underline{a_i}}{}_{q} \prod_{\substack{j=1 \\ j \neq i}}^{2m-1} \Gamma^{\underline{a_j}} \, .
\eenn
\end{lemma}

\begin{proof}
Let us prove it by induction. For $m=1$ the above statement holds trivially thanks the defining property
\be \label{flatclifford}
\{\Gamma_{\underline p} , \Gamma_{\underline q}\} = 2 \, \eta_{\underline{pq}} \, \mathbb{I}_{2^{d/2}} \, .
\ee

Now, let us assume it holds for a certain $m$. Using
\be \label{ABCABCACB}
\{A \cdot B , C \} = A \[ B, C \] + \{A, C\} B \, ,
\ee
it follows that
\begin{align}
\{\Gamma^{\underline{a_1}} \cdot\cdot\cdot \Gamma^{\underline{a_{2m+1}}}, \Gamma_{\underline q}\} &= \Gamma^{\underline{a_1}} \cdot\cdot\cdot \Gamma^{\underline{a_{2m}}} \[\Gamma^{\underline{a_{2m+1}}}, \Gamma_{\underline q}\] + \{\Gamma^{\underline{a_1}} \cdot\cdot\cdot \Gamma^{\underline{a_{2m}}}, \Gamma_{\underline q}\} \Gamma^{\underline{a_{2m+1}}} \nonumber \\
\, &= 2 \, \Gamma^{\underline{a_1}} \cdot\cdot\cdot \Gamma^{\underline{a_{2m+1}}} \Gamma_{\underline q} - \, \Gamma^{\underline{a_1}} \cdot\cdot\cdot \Gamma^{\underline{a_{2m}}} \{\Gamma_{\underline q} , \Gamma^{\underline{a_{2m+1}}}\} + \nonumber\\
\, & \quad + \Gamma^{\underline{a_1}} \cdot\cdot\cdot \Gamma^{\underline{a_{2m-1}}} \[\Gamma^{\underline{a_{2m}}}, \Gamma_{\underline q}\] \Gamma^{\underline{a_{2m+1}}} 
 + \{\Gamma^{\underline{a_1}} \cdot\cdot\cdot \Gamma^{\underline{a_{2m-1}}}, \Gamma_{\underline q}\} \Gamma^{\underline{a_{2m}}} \Gamma^{\underline{a_{2m+1}}} \nonumber \\
\, & = \sum_{i=1}^{2m-1} (-1)^{i+1} \, 2 \, \delta^{\underline{a_i}}{}_{q} \prod_{\substack{j=1 \\ j \neq i}}^{2m-1} \Gamma^{\underline{a_j}} \Gamma^{\underline{a_{2m}}} \Gamma^{\underline{a_{2m+1}}} + \nonumber \\
\, & \quad - 2 \, \Gamma^{\underline{a_1}} \cdot\cdot\cdot \Gamma^{\underline{a_{2m-1}}} \, \delta^{\underline{a_{2m}}}{}_{q} \Gamma^{\underline{a_{2m+1}}} + 2 \, \Gamma^{\underline{a_1}} \cdot\cdot\cdot \Gamma^{\underline{a_{2m}}} \, \delta^{\underline{a_{2m+1}}}{}_{q} \nonumber \\
\, & = \sum_{i=1}^{2m+1} (-1)^{i+1} \, 2 \, \delta^{\underline{a_i}}{}_{q} \prod_{\substack{j=1 \\ j \neq i}}^{2m+1} \Gamma^{\underline{a_j}}\,. \nonumber
\end{align}
This concludes the proof.\\
\end{proof}

\begin{corollary} \label{anticommIIA}
Given a subset $\{\Gamma^{\underline{a_1}},..., \Gamma^{\underline{a_{2m}}}\}$ of the $d$-dimensional Gamma matrices in flat space with even cardinality, we have
\benn
\{\Gamma^{\underline{a_1}} \cdot\cdot\cdot \Gamma^{\underline{a_{2m}}}, \Gamma_{\underline q}\} = 2 \, \Gamma^{\underline{a_1}} \cdot\cdot\cdot \Gamma^{\underline{a_{2m}}} \, \Gamma_{\underline q} + \sum_{i=1}^{2m} (-1)^{i+1} \, 2 \, \delta^{\underline{a_i}}{}_{q} \prod_{\underset{j \neq i}{j=1}}^{2m} \Gamma^{\underline{a_j}} \, .
\eenn
\end{corollary}

\begin{proof}
It follows directly from the previous lemma, using again formulae \eqref{flatclifford} and \eqref{ABCABCACB}.\\
\end{proof}

\begin{lemma}
\label{lemma2}
Using the notation declared at the beginning of this appendix, it holds that%\footnote{As a remark, these formulae should be equivalent to the equations (12.1.16) of \cite{Polchinski:1998rr}, that is
%\begin{subequations}
%\begin{align}
%&\Gamma^{\nu}\Gamma^{\mu_1...\mu_p} = \Gamma^{\nu\mu_1...\mu_p} + p \, \eta^{\nu[\mu_1} \Gamma^{\mu_2...\mu_p]} \, , \nonumber\\
%&\Gamma^{\mu_1...\mu_p}\Gamma^{\nu} = (-1)^p \Gamma^{\nu\mu_1...\mu_p} + (-1)^{p+1} \, p \, \eta^{\nu[\mu_1} \Gamma^{\mu_2...\mu_p]} \, . \nonumber
%\end{align}
%\end{subequations}
%The latter provides an alternative way in the proof of formula \eqref{TrA1}.}
\benn
\{\widetilde\Gamma_I, \Gamma_\alpha\} = \(1+(-1)^{\mathfrak p}\) \widetilde\Gamma_I\Gamma_\alpha + \widetilde\Gamma_{I\alpha} \, ,%\begin{cases} 2 \, \widetilde\Gamma_I\Gamma_\alpha + \widetilde\Gamma_{I\alpha} \, , \quad &\text{IIA case} \, , \\ \widetilde\Gamma_{I\alpha} \, , \quad &\text{IIB case} \, . \end{cases}
\eenn
being $\mathfrak p = 0$ for the IIA case, while $\mathfrak p = 1$ for the IIB one.
\end{lemma}

\begin{proof}
It is a direct consequence of the previous Lemma and its Corollary. To begin with, let us stress that $G^{(n)}_{\underline{a_1...a_n}}$ is totally antisymmetric. Therefore, we have
\benn
G^{(n)}_{\underline{a_1...a_n}} \, \Gamma^{\underline{a_1...a_n}} = G^{(n)}_{\underline{a_1...a_n}} \, \Gamma^{\underline{a_1}} \cdot\cdot\cdot \Gamma^{\underline{a_n}} \, .
\eenn
For the same reason,
\begin{align}
\sum_{i=1}^{n} (-1)^{i+1} \, G^{(n)}_{\underline{a_1...a_{i-1}a_ia_{i+1}...a_n}} \, \delta^{\underline{a_i}}{}_{p} \, \Gamma^{\underline{a_1}} \cdot\cdot\cdot \Gamma^{\underline{a_{i-1}}}\Gamma^{\underline{a_{i+1}}} \cdot\cdot\cdot \Gamma^{\underline{a_n}} &= \sum_{i=1}^{n} G^{(n)}_{\underline{a_i \, b_2...b_n}} \, \delta^{\underline{a_i}}{}_{p} \, \Gamma^{\underline{b_2 ... b_n}} \nonumber \\
&= n \, G^{(n)}_{\underline{p \, b_2...b_n}} \, \Gamma^{\underline{b_2 ... b_n}} \nonumber
\end{align}
With this hints in mind, the derivation of the claim is straightforward.\\
\end{proof}

\begin{lemma}\label{reversing}
Reversing the order in the product of $n$ different Gamma matrices results in
\benn
\Gamma^{a_1 a_2 ... a_n} = (-1)^{\frac{n(n-1)}{2}} \, \Gamma^{a_n a_{n-1} ... a_1} \, .
\eenn
\end{lemma}

\begin{proof}
Thanks to the antisymmetrization of the indices we have
\benn
\Gamma^{a_1 ... a_n} = (-1)^{n-1} \, \Gamma^{a_n a_1 ... a_{n-1}} = (-1)^{n-1} (-1)^{n-2} \, \Gamma^{a_n a_{n-1} a_1 ... a_{n-2}} \, .
\eenn
Iterating this process, we finally get
\benn
\Gamma^{a_1 ... a_n} = (-1)^{(n-1)+(n-2)+\ldots+(n-(n-1))} \, \Gamma^{a_n a_{n-1}  ... a_1} \, .
\eenn
Clearly, the exponent of the above $(-1)$ is given by the sum of $n-1$ contributions, that is
\benn
\sum_{i=1}^{n-1} (n-i) = \frac{n(n-1)}{2} \, .
\eenn
This concludes the proof.
\end{proof}

\begin{corollary} \label{gammaduality}
In even dimension $d$, given a background with $n_t$ time-like directions, there exists a chirality operator $\Gamma^{(d+1)}$ which satisfies
\be
\Gamma^{a_1 ... a_n} \, \Gamma^{(d+1)} = (-1)^{\frac{d(d-1)}{2}+n_t} \, i^{\frac{d}{2}+n_t}\, (-1)^{\frac{n(n-1)}{2}} \,\frac{1}{(d-n)!} \, \varepsilon^{a_1 ... a_n \,b_1 ... b_{d-n}} \, \Gamma_{{b_1 ... b_{d-n}}} \, .
\ee
\end{corollary}

\begin{proof}
It is a direct consequence of the above Lemma applied to formula \eqref{gammadualityreversed}, given also the properties in \eqref{chiralityproperties}.
\begin{comment}
In \cite{Kuusela:2019iok}, the author dealt with Gamma matrices in $d$-dimensional flat space with $t$ time-like directions. In particular, the signature has been chosen such that the latter are the first $t$ ones. 

In (24), they defined a particular combination of Gamma matrices which we identify as ours $\Gamma^{(11)}$, let say
\benn
\Gamma^{(11)} = (-i)^{\frac{d}{2}+t} \, \Gamma^{\underline{n_1}} \ldots \Gamma^{\underline{n_t}} \, \Gamma^{\underline{s_1}} \ldots \Gamma^{\underline{s_{d-t}}} \, , \quad \(\Gamma^{(11)}\)^2 = \text{id} \, ,
\eenn
where the $n$($s$)-indices stand for the time(space)-like directions, sorted in ascending order. As a specific example, look at the background of interest in \eqref{asymetric}. There, the $0$-direction is space-like and corresponds to the thermal one. On the other hand, the $1$-direction has been taken time-like and it is the only one of this kind. Therefore, for $d=10$ and $t=1$, we have
\benn
\Gamma^{(11)} = - \, \Gamma^{\underline 1} \, \Gamma^{\underline 0} \, \Gamma^{\underline 2} \ldots \Gamma^{\underline 9} = + \, \Gamma^{\underline 0} \, \Gamma^{\underline 1} \ldots \Gamma^{\underline 9} \, .
\eenn
With the conventions established, the statement of the corollary easily follows.
\end{comment}
\end{proof}

\begin{corollary} \label{Gamma11duality}
The duality among RR-field strengths in the democratic formalism realizes in
\benn
\frac{1}{n!} \, \cancel{G}_{(n)} \Gamma^{(11)} = - \frac{1}{(10-n)!} \, \cancel{G}_{(10-n)} \, .
\eenn
It then follows that
\benn
\widetilde \Gamma_I \, \Gamma^{(11)} = (-1)^{I \, (\mathfrak p + 1) + 1} \, \widetilde \Gamma_I \, .
\eenn
\end{corollary}

\begin{proof}
It is an application of formula \eqref{gammaduality} for $d=10$ and $n_t=1$. Indeed, given
\benn
\cancel{G}_{(n)} \Gamma^{(11)} = - \cancel{G}^{(n)}_{\underline{a_1 ... a_n}} \Gamma^{\underline{a_1 ... a_n}} \Gamma^{(11)} = - \cancel{G}^{(n)}_{\underline{a_1 ... a_n}} (-1)^{\frac{n(n-1)}{2}} \,\frac{1}{(10-m)!} \, \varepsilon^{\underline{a_1 ... a_n \,b_1 ... b_{10-n}}} \, \Gamma_{b_1 ... b_{10-n}} \, ,
\eenn
we can recognize
\benn
n! \(*G_{(n)}\)_{\underline c_1 ... \underline c_{10-n}} = \cancel{G}^{(n)}_{\underline{a_1 ... a_n}} \, \varepsilon^{\underline{a_1 ... a_n}}{}_{\underline{b_1 ... b_{10-n}}} \, .
\eenn
Therefore, from the duality relation in \eqref{dualityconstraint}, the first statement of the corollary follows readily. The second one is just a direct implication.
\end{proof}

Now, let us discuss very important results for this work concerning the trace of particular combinations of gamma matrices.
\begin{lemma}
\label{lemma3}
\benn
\eta^{\alpha\beta} \, \Tr \(\widetilde\Gamma_1 \widetilde\Gamma_{2\alpha} \, \Gamma_\beta\) = \frac12 \, \eta^{\alpha\beta} \, \Tr \( \widetilde\Gamma_{1\alpha} \widetilde\Gamma_{2\beta}\) = \eta^{\alpha\beta} \, \Tr \(\widetilde\Gamma_2 \widetilde\Gamma_{1\alpha} \, \Gamma_\beta\)\, .
\eenn
\end{lemma}

\begin{proof}
First of all, let us compute
\benn
\eta^{\alpha\beta} \{\widetilde\Gamma_{I\alpha} , \Gamma_\beta\} = \frac18 \, e^{2\phi} \eta^{\alpha\beta} \, \partial_\alpha X^p \, \partial_\beta X^q \, e^{\underline p}{}_{p} \, e^{\underline q}{}_{q} \, \sum_n (-1)^{f_{(n,I)}} \, \frac{1}{(n-1)!} \,  G^{(n)}_{\underline{p \, a_2 ... a_n}} \, \{\Gamma^{\underline{a_2 ... a_n}}, \Gamma_{\underline q}\} \, .  
\eenn
Thanks to the previous lemma and its corollary, it is an easy task. Notice that
\benn
\frac14 \, e^{2\phi} \eta^{\alpha\beta} \, \partial_\alpha X^p \, \partial_\beta X^q \, e^{\underline p}{}_{p} \, e^{\underline q}{}_{q} \, \sum_n (-1)^{f_{(n,I)}} \, \frac{1}{(n-2)!} \,  G^{(n)}_{\underline{p \, q \, a_3 ... a_n}} \, \Gamma^{\underline{a_3 ... a_n}} = 0 \, ,  
\eenn
since $G^{(n)}_{\underline{p \, q \, a_3 ... a_n}}$ is antisymmetric in $p \leftrightarrow q$, but is contracted with something symmetric. Therefore, we can summarize the result as
\be
\label{anticommforproof}
\eta^{\alpha\beta} \{\widetilde\Gamma_{I\alpha} , \Gamma_\beta\} = (1+(-1)^{\mathfrak p + 1}) \, \eta^{\alpha\beta} \, \widetilde\Gamma_{I\alpha} \Gamma_\beta \, ,%\begin{cases} 0 \, , \quad & \text{IIA case} \, , \\ 2 \, \eta^{\alpha\beta} \, \widetilde\Gamma_{2\alpha} \Gamma_\beta \, , \quad & \text{IIB case} \, , \end{cases} 
\ee
being $\mathfrak p = 0$ $(\mathfrak p = 1)$ in the IIA (IIB) case.

Moreover, we can use the cyclicity of the trace as
\benn
\Tr \(\Gamma_\beta \widetilde\Gamma_1\widetilde\Gamma_{2\alpha}\) = \Tr \( \widetilde\Gamma_1\widetilde\Gamma_{2\alpha}\Gamma_\beta\) = - \Tr \( \widetilde\Gamma_1\Gamma_\beta\widetilde\Gamma_{2\alpha}\) + \Tr \(\widetilde\Gamma_1 \{\widetilde\Gamma_{2\alpha} , \Gamma_\beta\} \) \, ,
\eenn
from which
\benn
\Tr \( \{\widetilde\Gamma_1,\Gamma_\beta\} \widetilde\Gamma_{2\alpha}\) = \Tr \(\widetilde\Gamma_1 \{\widetilde\Gamma_{2\alpha} , \Gamma_\beta\} \) \, .
\eenn
Combining this equation with lemma \ref{lemma2} and \eqref{anticommforproof}, we get the first equality in the statement of the lemma. For the second one, the proof proceeds in the very same way (just switch $1$ and $2$ in the above formula).
\end{proof}

From now on, the formula (e.g., see \cite{Kuusela:2019iok})
\be
\label{orthostuff}
\Tr \( \Gamma^{a_1...a_n} \Gamma_{b_1 ... b_m} \) = 2^{\lfloor{\frac{d}{2}} \rfloor} (-1)^{\frac{n(n-1)}{2}} \, \delta_{mn} \, n! \, \delta^{[a_1}{}_{b_1} \cdot\cdot\cdot \delta^{a_n]}{}_{b_n}\,,
\ee
will play a crucial role. Here, $d$ stands for the dimension of the target space. In what follows, we will set $d=10$.

A first consequence is that
\be \label{traceHslash}
\Tr \, \cancel H_\alpha = \Tr \, \cancel \omega_\alpha = 0 \, .
\ee
Indeed,
\benn
\Tr \, \cancel H_\alpha = H_{\alpha\underline{ab}} \, \Tr \(\Gamma^{\underline a} \Gamma^{\underline b}\) = 32 \, H_{\alpha\underline{ab}} \, \delta^{\underline a}{}_{b} \, ,
\eenn
which vanishes since the components of $H_{(3)}$ are totally antisymmetric. The same holds for $\Tr \, \cancel \omega_\alpha$. Similarly, exploiting also the statement of corollary \ref{gammaduality}, if follows that
\be \label{tracedualHslash}
\Tr \(\cancel H_\lambda \, \Gamma^{(11)} \varepsilon^{\alpha\beta} \Gamma_{\alpha\beta}\) = \Tr \(\cancel \omega_\lambda \, \Gamma^{(11)} \varepsilon^{\alpha\beta} \Gamma_{\alpha\beta}\) = 0 \, .
\ee

Other relevant implications are
\begin{lemma}
\label{lemma4}
\benn
\Tr \(\widetilde\Gamma_1 \widetilde\Gamma_2\) = \frac18 \, e^{2\phi} \, (-1)^{\mathfrak p + 1} \sum_n |G_{(n)}|^2 = 0 \, ,
\eenn
where $\mathfrak p = 0$ $(\mathfrak p = 1)$ for the IIA (IIB) case and
\benn
|G_{(n)}|^2 = \frac{1}{n!} G^{(n)}_{\underline{a_1...a_n}} \,  G_{(n)}^{\underline{a_1...a_n}} \, .
\eenn
\end{lemma}

\begin{proof}
It is basically a direct consequence of formula \eqref{orthostuff}. Indeed,
\begin{align}
\Tr \(\widetilde\Gamma_1 \widetilde\Gamma_2\)
&= \frac{1}{256} \, e^{2\phi} \sum_{n,m} (-1)^{\frac{n(n-1)}{2} + m + 1} \, \frac{1}{n!} \, G^{(n)}_{\underline{a_1...a_n}} \,  \frac{1}{m!} \, G_{(m)}^{\underline{b_1...b_m}} \, \Tr \(\Gamma^{\underline{a_1...a_n}}\Gamma_{\underline{b_1...b_m}}\) \nonumber\\
&= \frac18 \, e^{2\phi} \sum_{n} (-1)^{n + 1} \, \frac{1}{n!} \, G^{(n)}_{\underline{a_1...a_n}} \,  G_{(n)}^{\underline{b_1...b_n}} \, \delta^{\underline{a_1}}{}_{b_1} \cdot\cdot\cdot \delta^{\underline{a_n}}{}_{b_n} \nonumber \\
&= \frac18 \, e^{2\phi} \sum_{n} (-1)^{n + 1} \, \frac{1}{n!} \, G^{(n)}_{\underline{a_1...a_n}} \,  G_{(n)}^{\underline{a_1...a_n}} \, . \nonumber
\end{align}
Notice that we remove the antisymmetrization on the indices $a_1, ..., a_n$ belonging to the deltas of Kronecker, since they are contracted with the indices of the Ramond-Ramond field strength. Moreover, in the IIA (IIB) case, the sum runs over even (odd) $n$.

Finally, we must not forget to impose the duality constraints in \eqref{dualityconstraint}. Remembering that $*\hspace{-2pt}*\hspace{-2pt}G_{(n)}=(-1)^{n+1} G_{(n)}$, they imply
\be
G_{(n)} \wedge *G_{(n)} = - G_{(10-n)} \wedge *G_{(10-n)} \, .
\ee
Therefore, in the democratic formalism, it holds that
\be
\sum_n |G_{(n)}|^2 = 0\,.
\ee  
\end{proof}

\begin{lemma}
\label{lemma5}
\benn
\Tr \(  \widetilde\Gamma_{1\alpha} \widetilde\Gamma_{2\beta} \) = \frac12 \, e^{2\phi} \,
\partial_{\alpha} X^p \, \partial_{\beta} X^q \, e^{\underline p}{}_{p} \, e^{\underline q}{}_{q} \,  \sum_n \, |G_{(n)}|^2_{\underline{pq}} \, ,
\eenn
where
\benn
|G_{(n)}|^2_{\underline{pq}} = \frac{1}{(n-1)!} G^{(n)}_{\underline{p \, a_2 ... a_n}} G_{(n)}^{\,\underline{k \, a_2 ... a_n}} \, \eta_{\underline{kq}} \,.
\eenn
\end{lemma}

\begin{proof}
Uplifting the definitions at the beginning of this appendix to the target space, what we have to compute is
\benn
\Tr \(  \widetilde\Gamma_{1\underline p} \widetilde\Gamma_{2\underline q} \) 
= \frac{e^{2\phi} }{64} \, \sum_{n,m} (-1)^{\frac{n(n-1)}{2} + m + 1} \, \frac{1}{(n-1)!} \, G^{(n)}_{\underline{p \, a_2 ... a_n}} \, \frac{1}{(m-1)!} \, G^{(m)}_{\underline{q\,b_2 ... b_m}} \Tr \(\Gamma^{\underline{a_2 ... a_n}}  \Gamma^{\underline{b_2 ... b_m}} \) \, .
\eenn
The proof proceeds in exactly the same way as the previous one and the pullback on the world-sheet of the above expression provides the statement of the lemma. The only thing to pay attention to is the number of Gamma matrices in the trace, $2(n-1)$ instead of $2n$. This implies
\benn
\Tr \(\Gamma^{\underline{a_2 ... a_n}}  \Gamma_{\underline{b_2 ... b_m}} \) = 32 \, (-1)^{\frac{n(n-1)}{2}} (-1)^{1-n} \, \delta_{mn} \, (n-1)! \, \delta^{[\underline{a_2}}{}_{\underline{b_2}} \cdot\cdot\cdot \delta^{\underline{a_n}]}{}_{\underline{b_n}} \, .
\eenn
\end{proof}

\begin{lemma} \label{lemma6}
The following results hold:
\begin{enumerate}
\item $\displaystyle \Tr \( \widetilde\Gamma_1 \widetilde\Gamma_{2\alpha} \Gamma_\beta \) =\frac14 \, e^{2\phi} \,
\partial_{\alpha} X^p \, \partial_{\beta} X^q \, e^{\underline p}{}_{p} \, e^{\underline q}{}_{q} \,  \sum_n \, \[|G_{(n)}|^2_{\underline{pq}} + \[*\(G_{(n)} \wedge *G_{(n+2)}\)\]_{\underline{pq}} \]$ ,
\item $\displaystyle \Tr \( \widetilde\Gamma_1 \widetilde\Gamma_{2\alpha} \Gamma_\beta \) = (-1)^{\mathfrak p + 1} \, \Tr \( \widetilde\Gamma_{1\alpha} \widetilde\Gamma_2 \Gamma_\beta \)$ ,
\item $\displaystyle \Tr \( \widetilde\Gamma_2 \widetilde\Gamma_{1\alpha} \Gamma_\beta \) = \frac14 \, e^{2\phi} \,
\partial_{\alpha} X^p \, \partial_{\beta} X^q \, e^{\underline p}{}_{p} \, e^{\underline q}{}_{q} \,  \sum_n \, \[|G_{(n)}|^2_{\underline{pq}} - \[*\(G_{(n)} \wedge *G_{(n+2)}\)\]_{\underline{pq}} \]$ ,
\end{enumerate}
being $\mathfrak p = 0$ $(\mathfrak p = 1)$ in the IIA (IIB) case.
\end{lemma}

\begin{proof}
Similarly to the previous proof, let us focus on
\benn
\Tr \( \widetilde\Gamma_1 \widetilde\Gamma_{2\underline p} \Gamma_{\underline q} \) = \frac{e^{2\phi}}{128} \, \sum_{n,m} (-1)^{\frac{n(n-1)}{2} + m + 1} \, \frac{1}{n!} \, G^{(n)}_{\underline{a_1...a_n}} \,  \frac{1}{(m-1)!} \, G^{(m)}_{\underline{p \, b_2...b_m}} \, \Tr \(\Gamma^{\underline{a_1...a_n}}\Gamma^{\underline{b_2...b_m}}\Gamma_{\underline q}\) \, .
\eenn

First of all, let us suppose that neither $G_{(n)}$ nor $G_{(m)}$ have legs along the $q$-direction, apart from possibly the case where $p=q$. This implies that $a_i$, $b_j \neq q$, $\forall i$, $j$. Therefore, we have
\benn
\Tr \(\Gamma_{\underline q}\Gamma^{\underline{a_1...a_n}}\Gamma^{\underline{b_2...b_m}}\) = \Tr \(\Gamma^{\underline{a_1...a_n}}\Gamma^{\underline{b_2...b_m}}\Gamma_{\underline q}\) = (-1)^{n+m-1} \, \Tr \(\Gamma_{\underline q}\Gamma^{\underline{a_1...a_n}}\Gamma^{\underline{b_2...b_m}}\) \, .
\eenn
Here, the first equality holds thanks to the cyclicity of the trace, while the second one is a consequence of the anticommutation rules in \eqref{flatclifford}. Since $n+m$ is even both in the IIA case and in the IIB one, it follows that
\benn
\Tr \(\Gamma_{\underline q}\Gamma^{\underline{a_1...a_n}}\Gamma^{\underline{b_2...b_m}}\) = 0 \, .
\eenn

The above result holds also if both $G_{(n)}$ and $G_{(m)}$ have legs along the $q$-direction, excluding the case $p=q$. Indeed, if so, due to the antisymmetric nature of $G_{(m)}$, none of the $b$-indices could be equal to $q$.

Going further, let assume now that $G_{(n)}$ flows on the $q$-direction. Therefore, for every permutation of the indices, there exists a $k$ such that $a_k = q$. Since $G_{(n)}$ is totally antisymmetric, it is unique. In other words, $a_i \neq q$ $\forall i \neq k$. Then, it follows that
\benn
G^{(n)}_{\underline{a_1...a_n}} \, \Tr \(\Gamma_{\underline q}\Gamma^{\underline{a_1...a_n}}\Gamma^{\underline{b_2...b_m}}\) = n \, G^{(n)}_{\underline{q \, a_2...a_n}} \Tr \(\Gamma^{\underline{a_2...a_n}}\Gamma^{\underline{b_2...b_m}}\) \, .
\eenn
From now on, the computation is the same as the previous proof and we can write the contribution to $\Tr \( \widetilde\Gamma_1 \widetilde\Gamma_{2\underline p} \Gamma_{\underline q} \)$ in these cases as
\be
\label{sympart}
\frac14 \, e^{2\phi} \, \sum_n \,  \frac{1}{(n-1)!} G^{(n)}_{\underline{p \, a_2 ... a_n}} G_{(n)}^{\,\underline{k \, a_2 ... a_n}} \, \eta_{\underline{kq}} = \frac14 \, e^{2\phi} \, \sum_n \, |G_{(n)}|^2_{\underline{pq}} \, .
\ee

Notice that there are no contradictions with the previous case. Indeed, here $n=m$ and if one of the $a$-indices was equal to $q$ then everything would vanish due to the antisymmetric nature of $G_{(n)}$. On the contrary, $p=q$ is not excluded. Moreover, $|G_{(n)}|^2_{\underline{pq}}$ can also be defined as
\be
|G_{(n)}|^2_{\underline{pq}} = \frac{\delta |G_{(n)}|^2}{\delta \eta^{\underline{pq}}} \, .
\ee
This contribution is thus manifestly symmetric in $p \leftrightarrow q$.

Finally, let us assume that $G_{(m)}$ flows on the $q$-direction, besides along the $p$ one. Then, by construction, we have $p \neq q$. As before, we can state
\benn
G^{(m)}_{\underline{p \, b_2...b_m}} \, \Tr \(\Gamma^{\underline{a_1...a_n}}\Gamma^{\underline{b_2...b_m}}\Gamma_{\underline q}\) = (-1)^{\mathfrak p} \, (m-1) \, G^{(m)}_{\underline{p q \, b_3...b_m}} \Tr \(\Gamma^{\underline{a_1...a_n}}\Gamma^{\underline{b_3...b_m}}\) \, .
\eenn
All in all, given that
\benn
\Tr \(\Gamma^{\underline{a_1 ... a_n}}  \Gamma_{\underline{b_3 ... b_m}} \) = 32 \, (-1)^{\frac{n(n-1)}{2}} \, \delta_{m,n+2} \, n! \, \delta^{[\underline{a_1}}{}_{\underline{b_3}} \cdot\cdot\cdot \delta^{\underline{a_n}]}{}_{\underline{b_{n+2}}} \, ,
\eenn
the final result for the contribution to $\Tr \( \widetilde\Gamma_1 \widetilde\Gamma_{2\underline p} \Gamma_{\underline q} \)$ in these cases is
\be
\label{antisympart}
- \frac14 \, e^{2\phi} \, \sum_n \frac{1}{n!} \, G_{ \underline{pq} \, \underline{a_1 ... a_n}}^{(n+2)} \, G_{(n)}^{\underline{a_1...a_n}} \, .
\ee
Notice that the above formula is manifestly antisymmetric in $p \leftrightarrow q$. 

Using the notations and conventions of the previous appendix, we get
\benn
\[*\(G_{(n)} \wedge *G_{(n+2)}\)\]_{\underline{pq}} = - \frac{1}{n!} \, G_{ \underline{pq} \, \underline{a_1 ... a_n}}^{(n+2)} \, G_{(n)}^{\underline{a_1...a_n}} \, .
\eenn
Therefore, collecting all the cases, the pullback on the world-sheet of the sum of \eqref{sympart} and \eqref{antisympart} provides the first part of the statement of the lemma.

To prove the second one, we have to compute explicitly also
\benn
\Tr \( \widetilde\Gamma_{1\underline{p}} \widetilde\Gamma_2 \Gamma_{\underline q} \) = \frac{e^{2\phi}}{128} \, \sum_{n,m} (-1)^{\frac{n(n-1)}{2} + m + 1} \, \frac{1}{(n-1)!} \, G^{(n)}_{\underline{p \, a_2...a_n}} \,  \frac{1}{m!} \, G^{(m)}_{\underline{b_1...b_m}} \, \Tr \(\Gamma^{\underline{a_2...a_n}}\Gamma^{\underline{b_1...b_m}}\Gamma_{\underline q}\) \, .
\eenn
The steps are completely analogous to the previous ones. To give some hints, if $G_{(n)}$ flows on the $q$-direction given $p \neq q$, we have
\benn
G^{(n)}_{\underline{p \, a_2...a_n}} \, \Tr \(\Gamma^{\underline{a_2...a_n}}\Gamma^{\underline{b_1...b_m}}\Gamma_{\underline q}\) = (m-1) \, G^{(n)}_{\underline{p q \, a_3...a_n}} \Tr \(\Gamma^{\underline{a_3...a_n}}\Gamma^{\underline{b_1...b_m}}\)
\eenn
and
\benn
\Tr \(\Gamma^{\underline{a_3 ... a_n}}  \Gamma_{\underline{b_1 ... b_m}} \) = 32 \, (-1)^{\frac{m(m-1)}{2}} \, \delta_{n,m+2} \, m! \, \delta^{[\underline{a_3}}{}_{\underline{b_1}} \cdot\cdot\cdot \delta^{\underline{a_{m+2}}]}{}_{\underline{b_m}} \, .
\eenn

Similarly we also get the third part of the statement. This concludes the proof. En-passant, this computation provides an explicit check of lemma \ref{lemma3}, given the result in lemma \ref{lemma5}.
\end{proof}

\section{The fermionic mass operator}
\label{app:trace}

In this appendix we provide details to support section \ref{sec:fermsector}. First of all, we will compute explicitly the trace of the fermionic mass operator in \eqref{finalkgeom}. Then, we discuss how our general formulae can be applied to a specific example. To conclude, we gives details about working on-shell on the solution of the Virasoro constraints. 

\subsection{The reduced trace of the fermionic mass operator}
\label{app:redtrace}

Now, we are ready to compute the sum of the fermionic masses squared by means of $\Tr A_1$ (see \eqref{reducedtrace}). For the reader's convenience, we report here the main results of appendix \ref{app:useful}:
\begin{subequations}
\begin{align}
&\hspace{-4pt}\{\widetilde\Gamma_I, \Gamma_\alpha\} = \(1+(-1)^{\mathfrak p}\) \widetilde\Gamma_I\Gamma_\alpha + \widetilde\Gamma_{I\alpha}  \, ,\\%\begin{cases} 2 \, \widetilde\Gamma_I\Gamma_\alpha + \widetilde\Gamma_{I\alpha} \, , \quad &\text{IIA case} \, , \\ \widetilde\Gamma_{I\alpha} \, , \quad &\text{IIB case} \, , \end{cases}\\%\(1+(-1)^{\mathfrak p}\) \, \widetilde\Gamma_I\Gamma_\alpha + \widetilde\Gamma_{I\alpha} \, ,\\
&\eta^{\alpha\beta} \, \Tr \(\widetilde\Gamma_1 \widetilde\Gamma_{2\alpha} \Gamma_\beta\) = \frac12 \, \eta^{\alpha\beta} \, \Tr \( \widetilde\Gamma_{1\alpha} \widetilde\Gamma_{2\beta}\) \, , \\
&\Tr \(\widetilde \Gamma_1 \widetilde \Gamma_2\) = 0 \, , \\
&\label{tracelessHw}\Tr \, \cancel H_\alpha = \Tr \, \cancel \omega_\alpha = 0 \, ,
\end{align}
\end{subequations}
where
\be
\widetilde\Gamma_{I\alpha} = \frac18 \, e^{2\phi} \, \partial_\alpha X^p \, e^{\underline p}{}_{p} \, \sum_n (-1)^{f_{(n,I)}} \, \frac{1}{(n-1)!} \,  G^{(n)}_{\underline{p \, a_2 ... a_n}} \,  \Gamma^{\underline{a_2 ... a_n}} \, .
\ee
Given these relations, it is easy to derive\footnote{We also use that $\eta^{\alpha\beta} \, \Gamma_\alpha \Gamma_\beta = \eta^{\alpha\beta}g_{\alpha\beta}$.}
\begin{subequations}
\begin{align}
%&\hspace{-4pt}-\eta^{\alpha\beta} \, \Tr \(\widetilde \Gamma_1 \Gamma_\alpha \widetilde \Gamma_2 \Gamma_\beta\) = (-1)^{\mathfrak p + 1} \, \eta^{\alpha\beta} g_{\alpha\beta} \, \Tr \(\widetilde\Gamma_1\widetilde\Gamma_2\) 
%- \frac12 \, \eta^{\alpha\beta} \, \Tr \(\widetilde\Gamma_{1\alpha}\widetilde\Gamma_{2\beta}\) \, , \\
&\eta^{\alpha\beta} \, \Tr \(\widetilde \Gamma_1 \Gamma_\alpha \widetilde \Gamma_2 \Gamma_\beta\) = \frac12 \, \eta^{\alpha\beta} \, \Tr \(\widetilde\Gamma_{1\alpha}\widetilde\Gamma_{2\beta}\) \, , \\
&\label{antisym}\varepsilon^{\alpha\beta} \, \Tr \(\widetilde\Gamma_1\Gamma_\alpha\widetilde\Gamma_2\Gamma_\beta\) = \frac12 \, \varepsilon^{\alpha\beta} \, \Tr \(\widetilde\Gamma_1\widetilde\Gamma_{2\alpha}\Gamma_\beta\) + (-1)^{\mathfrak p + 1} \frac12 \, \varepsilon^{\alpha\beta} \, \Tr \(\widetilde\Gamma_{1\alpha}\widetilde\Gamma_2\Gamma_\beta\) \, .
\end{align}
\end{subequations}

Then, the statements of lemma \ref{lemma4}, \ref{lemma5} and \ref{lemma6} lead to
\be
\label{TrA1}
 \Tr A_1 = - \frac14 e^{2\phi} \, \partial_\alpha X^p \partial_\beta X^q \[\eta^{\alpha\beta} \sum_n \, |G_{(n)}|^2_{pq} - \varepsilon^{\alpha\beta} \,  \sum_n \[*\(G_{(n)} \wedge *G_{(n+2)}\)\]_{pq}\] ,
\ee
where
\be
|G_{(n)}|^2_{pq} = \frac{1}{(n-1)!} G^{(n)}_{p \, \underline{a_1...a_n}} G_{(n)}^{k \, \underline{a_1...a_n}} \, g_{kq} \, .
\ee
Notice that $p$ and $q$ are curved indices.
 
If we want to translate this result into the non-democratic formalism, we have to break the democracy among the RR potentials expressing everything in terms of the physical degrees of freedom encoded in $G_{(n)}$ with $n=2,4$ in the IIA case, or $n=1,3,5$ in the IIB one.

It is easy to show that
\be \label{dualmodulus}
|G_{(10-n)}|^2_{\underline{pq}} = - \frac{n+1}{n!} \, \delta^{\underline{b_0}}{}_{\underline p} \, G^{\underline{b_1 ... b_n}}_{(n)} \, \eta_{\underline{c_0 q}} \, G^{(n)}_{\underline{c_1 ... c_n}} \, \delta^{[\underline{c_0}}{}_{\underline{b_0}} \, \delta^{\underline{c_1}}{}_{\underline{b_1}} \cdot\cdot\cdot \delta^{\underline{c_n}]}{}_{\underline{b_n}} \, .
\ee
Indeed, notice that the above antisymmetrization of the deltas of Kronecker can be expanded as
\begin{multline}
\label{expandedantsymdelta}
(n+1) \, \delta^{[c_0}{}_{b_0} \delta^{c_1}{}_{b_1} \cdot\cdot\cdot \, \delta^{c_n]}{}_{b_n} - \delta^{c_0}{}_{b_0} \delta^{[c_1}{}_{b_1} \cdot\cdot\cdot \,  \delta^{c_n]}{}_{b_n} = \\ \sum_{i=1}^n \, (-1)^i \, \delta^{c_0}{}_{b_i} \, \frac{1}{n!} \hspace{-14pt} \sum_{{\substack{\pi \\ \pi(j) \in \{1,...,n\} \\ \, j \neq i}}} \hspace{-18pt} \sign(\pi_n) \, \delta^{c_{\pi(0)}}{}_{b_0} \, \delta^{c_{\pi(1)}}{}_{b_1} \cdot\cdot\cdot \, \delta^{c_{\pi(i-1)}}{}_{b_{i+1}} \, \delta^{c_{\pi(i+1)}}_{\, . \, b_{i+1}} \cdot\cdot\cdot \,
\delta^{c_{\pi(n)}}{}_{b_n} \, ,
\end{multline}
being $\sign(\pi_n)$ the parity of the permutation $\pi$ of $n$ objects.
To prove it, from its definition, we can expand the antisymmetrization of the indices as
\be
\delta^{[c_0}{}_{b_0} \delta^{c_1}{}_{b_1} \cdot\cdot\cdot \delta^{c_n]}{}_{b_n} = \frac{1}{(n+1)!} \[\sum_{\substack{\pi \\ \pi(0)=0}} + \sum_{\substack{\pi \\ \pi(0) \neq 0}} \] \sign(\pi_{n+1}) \, \delta^{c_{\pi(0)}}{}_{b_0} \, \delta^{c_{\pi(1)}}{}_{b_1} \cdot\cdot\cdot \delta^{c_{\pi(n)}}{}_{b_n}\,,
\ee
where $\sign(\pi_{n+1})$ is the parity of the permutation $\pi$ of $n+1$ objects. The first contribution includes $n!$ terms and it easy to be computed. Indeed, since $\pi(0)$ is fixed to $0$, it follows that $\sign(\pi_{n+1})$ is exactly equal to the parity of the permutation of the remaining $n$ objects. Then,
\begin{align}
&&\hspace{-12pt}\frac{1}{(n+1)!}  \sum_{\substack{\pi \\ \pi(0)=0}} \sign(\pi_{n+1}) \, \delta^{c_{\pi(0)}}{}_{b_0} \, \delta^{c_{\pi(1)}}{}_{b_1} \cdot\cdot\cdot \delta^{c_{\pi(n)}}{}_{b_n} &= \frac{\delta^{c_0}{}_{b_0}}{n+1} \, \frac{1}{n!} \,  \sum_{\pi} \, \sign(\pi_n) \, \delta^{c_{\pi(1)}}{}_{b_1} \cdot\cdot\cdot \delta^{c_{\pi(n)}}{}_{b_n} \nonumber \\
&&\hspace{-12pt} \,  &= \frac{\delta^{c_0}{}_{b_0}}{n+1} \, \delta^{[c_1}{}_{b_1} \cdot\cdot\cdot \delta^{c_n]}{}_{b_n}  \, . 
\end{align}
On the other hand, the second contribution includes $(n+1)!-n! = n \cdot n!$ terms and it is more tricky. To begin with, since $\pi(0) \neq 0$, it follows that exists one and only one $i \in \{1,...,n\}$ such that $\pi(i)=0$. So, this time we have
\begin{align}
&& \sign(\pi_{n+1}) & = \sign(\pi(0) \, \pi(1)  \cdot\cdot\cdot \pi(i-1) \cdot 0 \cdot \pi(i+1) \cdot\cdot\cdot \pi(n)) \nonumber \\
&& \, \ &= (-1)^i \, \sign(\pi(0) \, \pi(1)  \cdot\cdot\cdot \pi(i-1) \pi(i+1) \cdot\cdot\cdot \pi(n)) \nonumber \\
&& \, \ &= (-1)^i \, \sign(\pi_n) \, ,
\end{align}
where by construction $\pi(j) \in \{1,...,n\}$, $\forall \, j \neq i$. Now, formula \eqref{expandedantsymdelta} can be derived easy.

Therefore, plugging this result in \eqref{dualmodulus}, we get
\be
|G_{(10-n)}|^2_{\underline{pq}} = - |G_{(n)}|^2 \, \eta_{\underline{pq}} + |G_{(n)}|^2_{\underline{pq}} \, , \quad \forall \, n = 1, ... , 9 \, .
\ee
Crucially, it implies that, in the IIA case,
\be
\frac12 \sum_{n=2,4,6,8} |G_{(n)}|^2_{\underline{pq}} = \sum_{n=2,4} \[|G_{(n)}|^2_{\underline{pq}} - \frac12 |G_{(n)}|^2_{\underline{pq}} \] \, ,
\ee
or, in the IIB one,
\be
\frac12 \sum_{n=1,3,5,7,9} |G_{(n)}|^2_{\underline{pq}} = \sum_{n=1,3} \[|G_{(n)}|^2_{\underline{pq}} - \frac12 |G_{(n)}|^2_{\underline{pq}} \] + \frac12 \, |G_{(5)}|^2_{\underline{pq}} \, .
\ee
This proves the equivalence of the symbol $\{G_{n}\}^2_{\underline{pq}}$ in both the formalisms, democratic or not (see appendix \ref{app:sugraeom}). To sum up, we can express the final result as \eqref{finalsummuf}.

\subsection{Example: the WYM case}
\label{app:wym}

To fix the ideas, let us discuss a specific example. For the sake of simplicity, let us focus on the Witten background described in section \ref{sec:wym}. It displays just one non-self-dual RR field-strength (along with its electromagnetic dual) and no Kalb-Ramond field. If we consider a background string configuration which extends just along the first four Minkowskian directions, the pull-back of the spin-connection vanishes too. Therefore, no rotation $\mathcal R$ has to be applied to the spinors in this case.

The duality constraint links the RR field strengths as (see Corollary \ref{Gamma11duality})
\be
-\frac{1}{6!} \, \cancel{F}_{(6)} = \frac{1}{4!} \, \cancel{F}_{(4)} \, \Gamma^{(11)} \, .
\ee
Then, it follows that
\be
\begin{cases}
\displaystyle \widetilde \Gamma_1 = + \frac18 \, e^\phi \, \frac{1}{4!} \, \cancel{F}_{(4)} \, \Pi_1\, , \\[2ex]
\displaystyle \widetilde \Gamma_2 = - \frac18 \, e^\phi \, \frac{1}{4!} \, \cancel{F}_{(4)} \, \Pi_2 \, , 
\end{cases}
\ee
where
\be
\Pi_I = \frac12 \(\mathbb I_{32} + (-1)^{I+1} \Gamma^{(11)}\)\,,
\ee
are the projectors over the subspaces with definite chirality.

At low energies, the four-sphere looks flat for a closed string located at the tip of the cigar. Therefore the background configuration for the RR field strength $F_{(4)}$ is 
\be
\frac{1}{4!} \, \cancel F_{(4)} = \frac{3}{u_0} \, \frac{1}{4!} \, \varepsilon_{\underline{abcd}} \, \Gamma^{\underline{abcd}}\, , \quad a,b,c,d=6,7,8,9 \, .
\ee
Moreover, for the background string configuration at hand, it holds that
\be
\Gamma_\alpha \, \cancel F_{(4)} = \cancel F_{(4)} \, \Gamma_\alpha \, , \quad \Gamma_\alpha \, \Gamma^{(11)} = - \Gamma^{(11)} \, \Gamma_\alpha \, .
\ee
As a consequence
\be
\Gamma_\alpha \, \widetilde \Gamma_2 = - \widetilde \Gamma_1 \, \Gamma_\alpha \, , \quad \Gamma_\alpha \, \widetilde \Gamma_1 = - \widetilde \Gamma_2 \, \Gamma_\alpha \, .
\ee

All in all, we get
\be
\begin{cases}
\displaystyle A_1 = k^2 \, P_- \, \Pi_1 \, , \\
\displaystyle A_2 = k^2 \,  P_+ \, \Pi_2 \, , 
\end{cases}
\ee
where
\be \label{mufwym}
k^2 \equiv \frac{9}{32} \frac{1}{m_0 R^3} \, \eta^{\alpha\beta} g_{\alpha\beta} = \frac{9}{32} \(\rho^2 + M^2\) \, ,
\ee
and $P_\pm$ are the operators (\emph{not} projectors) defined as
\be
P_\pm = \frac12 \(\mathbb{I}_{32} \pm \frac{1}{\eta^{\lambda\delta} g_{\lambda\delta}} \, \varepsilon^{\alpha\beta} \Gamma_{\alpha\beta}\) \, .
\ee

To be more concrete, let us choose as background configuration
\be
X^0 = \rho \sigma \, , \quad X^1 = M \tau \, .
\ee
Moreover, let us deal with a specific representation of the Gamma matrices, that is
\be \label{rep}
\Gamma_{\underline 0} = \sigma_1 \otimes \mathbb{I}_{16} \, , \quad \Gamma_{\underline 1} = i \, \sigma_2 \otimes \mathbb{I}_{16} \, , \quad \Gamma_{\underline A} = \sigma_3 \otimes \gamma_{\underline A} \,, \quad {\underline A} = 2, ... , 9 \, ,
\ee
where $\sigma_1 = \begin{pmatrix} 0 & 1 \\ 1 & 0 \end{pmatrix}$, $\sigma_2 = \begin{pmatrix} 0 & -i \\ i & 0 \end{pmatrix}$ and $\sigma_3 = \begin{pmatrix} 1 & 0 \\ 0 & -1 \end{pmatrix}$ are the Pauli matrices and $\gamma_{\underline A}$ are Euclidean Dirac matrices in eight dimensions, such that $\{\gamma_{\underline A}, \gamma_{\underline B}\}=2 \, \mathbb{I}_{16} \, \delta_{\underline A \underline B}$, $\underline A$, $\underline B = 2, ..., 9$. They can be expressed as
\be \label{gammaA}
\gamma_{\underline A} = \begin{pmatrix} 0 & \Lambda_{\underline A} \\ \Lambda_{\underline A}^T & 0 \end{pmatrix} \, , \quad \Lambda_{\underline A} \Lambda_{\underline B}^T + \Lambda_{\underline B} \Lambda_{\underline A}^T = 2 \, \mathbb{I}_{8} \, \delta_{\underline A \underline B} \, , \quad \underline A, \underline B = 2, ..., 9 \, ,
\ee
where \cite{Green:2012oqa}
\begin{align} \label{Lambda}
&\Lambda_2 = \sigma_1 \otimes i \sigma_2 \otimes \mathbb{I}_2\,, & &\Lambda_6 = i \sigma_2 \otimes i \sigma_2 \otimes i \sigma_2\,, \nb \\
&\Lambda_3 = \sigma_3 \otimes i \sigma_2 \otimes \mathbb{I}_2\,, & &\Lambda_7 = \mathbb{I}_2 \otimes \sigma_1 \otimes i \sigma_2\,, \nb \\
&\Lambda_4 = i \sigma_2 \otimes \mathbb{I}_2 \otimes \sigma_3\,, & &\Lambda_8 = \mathbb{I}_2 \otimes \sigma_3 \otimes i \sigma_2\,, \nb \\
&\Lambda_5 = \mathbb{I}_2 \otimes \mathbb{I}_2 \otimes \mathbb{I}_2\,, & &\Lambda_8 = i \sigma_2 \otimes \mathbb{I}_2 \otimes \sigma_1 \, .
\end{align}

With this choice we have
\be \label{gamma0gamma1}
\frac{1}{2\sqrt{-g}} \, \varepsilon^{\alpha\beta} \Gamma_{\alpha\beta} = \Gamma^{\underline 0} \Gamma^{\underline 1} = \begin{pmatrix} +\mathbb{I}_{16} & 0 \\ 0 & -\mathbb{I}_{16} \end{pmatrix} \, , \quad \Gamma^{(11)} = \begin{pmatrix} -\gamma^{(11)} & 0 \\ 0 & +\gamma^{(11)} \end{pmatrix} \, ,
\ee
where
\be \label{8dgamma11}
\gamma^{(11)} = \gamma_{\underline 2} \cdot ... \cdot \gamma_{\underline 9} = \begin{pmatrix} + \, \mathbb{I}_{8} & 0 \\ 0 & - \, \mathbb{I}_{8} \end{pmatrix}\,,
\ee
is the eight-dimensional chirality operator.

Then, $A_1$ takes the form
\be
A_1 = k^2 \begin{pmatrix} 0 & 0 & 0 & 0\\ 0 & \frac12 \(1-\frac{2 M \rho}{\rho^2 + M^2}\) \mathbb I_8 & 0 & 0\\ 0 & 0 & \frac12 \(1+\frac{2 M \rho}{\rho^2 + M^2}\) \mathbb{I}_{8} & 0\\ 0 & 0 & 0 & 0 \end{pmatrix} \, ,
\ee
while $A_2$ is
\be
A_2 = k^2 \begin{pmatrix} \frac12 \(1+\frac{2 M \rho}{\rho^2 + M^2}\) \mathbb{I}_{8} & 0 & 0 & 0\\ 0 & 0  & 0 & 0\\ 0 & 0 & 0 & 0\\ 0 & 0 & 0 & \frac12 \(1-\frac{2 M \rho}{\rho^2 + M^2}\) \mathbb I_8 \end{pmatrix} \, .
\ee
%In both cases, 
%\be
%\mu_f^2 = \frac{9}{32} \frac{1}{m_0 R^3} \, \eta^{\alpha\beta} g_{\alpha\beta} = \frac{9}{32} \(\rho^2 + M^2\) \, .
%\ee

Considering the coupling between $\psi_1$ and $\psi_2$, the trace of the fermionic mass squares \eqref{reducedtrace} is found to be   
\be
\sum_f \mu_f^2 = \Tr A_1 = 8 k^2 =  \frac94 \(\rho^2+M^2\) \, .
\ee
It matches with the sum of the bosonic mass squares, which are \cite{Bigazzi:2023oqm}
\be
\sum_b \mu_b^2 = \frac94 \(\rho^2+M^2\) \, .
\ee

%Equivalently,
%\be
%\sum_f \mu_f^2 = \Tr_2 \mathcal M_f = \Tr A_2 = \frac94 \(\rho^2+M^2\) \, .
%\ee

Notice that these results are ready to be evaluated both in the semiclassical and in the Hagedorn regimes. Indeed, in the semiclassical regime, we have $\rho=M$, from which 
\be \label{semiclasssummuf2}
\sum_b \mu_{b,\text{class}}^2 = \sum_f \mu_{f,\text{class}}^2 = \frac92 \, \rho^2 \, .
\ee
On the other hand, in the Hagedorn limit $M\to0$, we get
\be \label{hagsummub2}
\sum_b \mu_{b,\text{Hag}}^2 = \sum_f \mu_{f,\text{Hag}}^2 =\frac94 \, \rho_H^2 \, .
\ee

\providecommand{\href}[2]{#2}\begingroup\raggedright
\endgroup

\end{document}